\documentclass[amsart,11pt,oneside,english]{article}
\usepackage{subcaption} 
\usepackage{amssymb,amsmath,amsfonts, amsmath,mathtools, 
amsthm,euscript,mathrsfs
,verbatim
,enumerate,multirow
,bbding,slashed
,multicol,color,array, 
esint,babel, tikz,tikz-cd,tikz-3dplot,tkz-graph,pgfplots,ytableau,graphicx,float,rotating,hyperref,geometry,mathdots,savesym,cite, wasysym,amscd,graphicx,pifont,float,setspace,multirow,wrapfig,picture, amsthm,hepth, enumitem, enumerate,xcolor
}
\usepackage{array}
\usepackage[utf8]{inputenc}
\usepackage{thm-restate}
 \usetikzlibrary{math}
\usepackage{thmtools}
\usepackage{prettyref}
\usepackage[T1]{fontenc}
\usepackage{datetime}
\numberwithin{equation}{section}
\usetikzlibrary{arrows,positioning,decorations.pathmorphing,   decorations.markings, matrix, patterns,shapes}
\hypersetup{colorlinks=true}
\hypersetup{linkcolor=black}
\hypersetup{citecolor=black}
\hypersetup{urlcolor=black}

\newcommand{\ee}{ \end{equation}}
\newcommand{\be}{\begin{equation}}

\usepackage{titletoc,tocloft}

\theoremstyle{plain}
\newtheorem*{thm*}{Theorem}
\theoremstyle{plain}
\newtheorem{thm}{Theorem}[section]
\newtheorem{conjec}{Conjecture}[section]
\newtheorem{lem}[thm]{Lemma}

\theoremstyle{definition}
\newtheorem{defn}[thm]{Definition}
\newtheorem*{defn*}{Definition}

\newtheorem{rem}[thm]{Remark}

\usepackage[normalem]{ulem}
\usepackage{makecell}
\numberwithin{equation}{section}

	\definecolor{myblue}{rgb}{.72,.83,.97}
	\definecolor{myred}{rgb}{0.9, 0.44, 0.37}
\begin{document}

\begin{titlepage}
\begin{center}
\vspace{4cm}
{\Huge\bfseries Flops and Fibral Geometry \\ of E$_7$-models \\  }
\vspace{2cm}
{
\LARGE  Mboyo Esole$^{\spadesuit}$ and Sabrina Pasterski$^\dagger$\\}
\vspace{1cm}

{\large $^{\spadesuit}$ Department of Mathematics, Northeastern University}\par
{ \em 360 Huntington Avenue, Boston, MA 02115, USA}\par
 \scalebox{.95}{Email: \quad {\tt  j.esole@northeastern.edu }}\par
\vspace{.3cm}

 \noindent{ $^\dagger$  Princeton Center for Theoretical Science,}\\
{\em Jadwin Hall, Princeton, NJ 08544, USA}\\
 \scalebox{.95}{Email:\quad {\tt  sabrina.pasterski@princeton.edu}}\par 
 \vspace{2cm}
{ \bf{Abstract:}}\\
\end{center}
{\date{\today\  \currenttime}}

An E$_7$-Weierstrass  model is conjectured to have eight distinct crepant resolutions whose flop diagram is a Dynkin diagram of type E$_8$.  In previous work, we explicitly constructed four distinct resolutions, for which the flop diagram formed a D$_4$  sub-diagram.  The goal of this paper is to explore those properties of a resolved E$_7$-model which are not invariant under flops. In particular, we examine the fiber degenerations, identify the fibral divisors up to isomorphism, and study violation of flatness appearing over certain codimension-three loci in the base, where a component of the fiber grows in dimension from a rational curve to a rational surface.  For each crepant resolution, we compute the triple intersection polynomial and the linear form induced by the second Chern class, as well as the holomorphic and ordinary Euler characteristics, and the signature of each fibral divisor.  We identify the isomorphism classes of the rational surfaces that break the flatness of the fibration.   Moreover, we explicitly show that the D$_4$ flops correspond to the crepant resolutions of the orbifold given by $\mathbb{C}^3$ quotiented by the Klein four-group.

 \vfill


\end{titlepage}

\tableofcontents

\newpage 

\section{Introduction and summary}

The study of  crepant resolutions  of singular Weierstrass models lies at the crossroads  of  algebraic  geometry, number theory, and string theory.
In mathematics, interest in elliptic fibrations  started with the pioneering work of Kodaira, N\'eron, Tate, Deligne, and others. Ever since,  elliptic fibrations  have appeared in a variety of situations  from algebraic geometry to number theory. In Calabi--Yau compactifications, elliptic fibrations are ubiquitous, as a large majority of known Calabi--Yau varieties are elliptically fibered. These hold a special place in birational geometry.   Meanwhile elliptic fibrations  play a key role in M-theory and  F-theory compactifications.  
In both F-theory and M-theory, elliptic fibrations offer elegant geometrizations of aspects of supersymmetric gauge theories. 
In particular,   elliptically fibrations are at the heart of the constructions of new super-conformal field theories that often have no alternative description. 
 The extended K\"ahler cones of Calabi--Yau threefolds are closely related to the Coulomb phases of  five-dimensional supersymmetric gauge theories with eight supersymmetric charges.

Simple types of elliptic fibrations are the so called $G$-models where $G$ is a simply connected compact Lie group associated with a Kodaira fiber whose dual graph is the affine version of the Dynkin diagram of the Lie algebra $\mathfrak{g}$ of $G$. 
Among the $G$-models with $G$ a compact exceptional Lie group, those not fully well understood are E$_6$ and E$_7$, since they are the only ones allowing  flops. 
Unfortunately,  not all possible minimal models corresponding to a Weierstrass model of type E$_6$ and E$_7$ are known  explicitly.  
 An E$_7$-model  is conjectured to have eight distinct crepant resolutions whose flop diagram is a Dynkin diagram of type E$_8$  (see Figure \ref{Figure:IG}).  
 Any two crepant resolutions of the same variety are related by a finite sequence of flops and have the same Euler characteristic and Hodge numbers.\footnote{The generating function for the Euler characteristic of  $G$-models, as well as the Hodge numbers for the Calabi-Yau threefold case can be found in \cite{Euler}. Additional characteristic invariants preserved by flops are given in \cite{EK.Invariant}.  } 
The goal of this paper is to explore those properties of a resolved E$_7$-model which are not invariant under flops.  
Examples include  the geometry of its fibral divisors as well as its intersection ring and the geometry of its fat fibers.  
The key results of this paper are as follows: 
\begin{enumerate}[label=\alph*)]
\item {\it Fiber degenerations of an E$_7$-model} (see Table \ref{Table:E7.Split})

We study degenerations of the generic curve of the fibral divisors using the hyperplane arrangement I($\mathbf{56}$, E$_7$). 
In this way, we avoid the box graph method used in \cite{Box} and correct a few discrepancies in the literature \cite{Diaconescu:1998cn,Box,DelZotto, Bhardwaj:2018yhy,E7}. 
In particular, for the chamber where the affine node can degenerate, we identify the correct splitting which was missing in \cite{Diaconescu:1998cn} and inaccurate in \cite{Box, DelZotto}.  We regard this result as a completion of the work of  Diaconescu and Entin  \cite{Diaconescu:1998cn}.

\item {\it D$_4$-flops of the E$_7$-model as flops of the orbifold  $\mathbb{C}^3/(\mathbb{Z}_2\times \mathbb{Z}_2)$} (see Figure \ref{Fig:C322})

In  \cite{E7}, we explicitly constructed  four of the eight conjectured E$_7$ minimal models and showed that their flops define a Dynkin diagram of type D$_4$. We now give a direct answer  to a question raised in that paper.  Namely, we show that the flops between the minimal models Y$_4$, Y$_5$,  Y$_6$, and Y$_8$ correspond to flops between the four crepant resolutions of the orbifold  $\mathbb{C}^3/(\mathbb{Z}_2\times \mathbb{Z}_2)$, which is isomorphic to 
 the binomial variety $$
\mathbb{C}[u_1,u_2,u_3,t]/(t^2- u_1 u_2 u_3 ).
$$

\item {\it Triple intersection numbers} (see Theorem \ref{thm:TripleCY3}) 

We compute the triple intersection polynomial of the fibral divisors in all chambers for which we have an explicit crepant resolution of the singularities
$$
F_m(\phi)= \int_{Y_m} (\sum_{a=0}^7 D_a \phi_a)^3, \quad m=4, 5, 6, 8. 
$$

The triple intersection depends on the chamber but not the blowups used to reach it. This data is useful for determining the matter representations which appear in $F$-theory compactifications. We consider specializations to the case of Calabi-Yau manifolds, and further to $S^2=-8$ and $g=0$, relevant to the CFT literature.

\item {\it Isomorphism classes of fibral divisors} (see Table \ref{Table:Div}, Section~\ref{isomcl})

 We identify the fibral divisors of an E$_7$-model up to isomorphism by exploiting the known crepant resolutions. In doing so, we also learn something about those chambers for which we do not have an explicit geometric construction. When two chambers are connected by a flop that does not change D$_i$, the isomorphism class of D$_i$ remains the same. This implies that we can easily move from chamber to chamber by flops and learn about the fiber geometry (modulo some empty entries).

\item  {\it Characteristic numbers of fibral divisors}  (see Theorems \ref{thm:Chern2} and \ref{thm:char.fibral})

We give the linear functions induced on $H^2(Y, \mathbb{Z})$ by the second Chern class of the minimal models Y=\{Y$_4$, Y$_5$, Y$_6$, Y$_8$\}
$$\mu:\quad H^2(Y,\mathbb{Z})\to \mathbb{Z}\quad \quad D\mapsto\int_Y D\cdot c_2 (TY)$$
as well as characteristic numbers of the fibral divisiors $D_a$ for each of these varieties.  In particular, we consider the signature $\tau(D$)  as well as the holomorphic $\chi_0(D)$ and ordinary $\chi(D)$ Euler characteristics.
These characteristic numbers provide precious information about the structure of the fibral divisors. For instance, the signature and the Euler characteristic also give information on the number of charged hypermultiplets and the number of rational curves appearing when the E$_7$ fibers degenerate.

\item{\it Fat fibers and loss of flatness}   (see Figure \ref{Fig:Q})

In each minimal model  Y$_a$ we analyze, the generic fiber C$_6$ of the fibral divisor D$_6$ specializes to a rational surface Q$_a$ over a codimension-three locus in the base, and does not give a flat fibration. 
The rational surfaces Q$_8$ and Q$_6$ are isomorphic to the Hirzebruch surfaces $\mathbb{F}_2$ and $\mathbb{F}_1$, respectively, and are related by  the usual Nagata transformation with  Q$_5$ serving as the intermediate surface.
{The rational surface Q$_5$ is obtained by blowing-up a point of the ($-1$)-curve of Q$_6\cong  \mathbb{F}_1$ or by blowing-up a point of the curve of self-intersection $2$ in $Q_8\cong\mathbb{F}_2$.}
The rational surface Q$_4$ is obtained by blowing-up the intersection of the two ($-1$)-curves of Q$_5$.

\end{enumerate}

This paper is organized as follows.  We spend Section 2 reviewing the necessarily preliminaries. We then present results a)-f) summarized above in sections 3-8, respectively.

\section{Preliminaries}

  In this section, we introduce the E$_7$ Weierstrass model, give our conventions for the Dynkin diagrams of E$_7$ and E$_8$, write out the weights for the fundamental representation $\mathbf{56}$ of E$_7$, and review 
  the structure of the hyperplane arrangement  I($\text{E}_7, \mathbf{56}$)   as analyzed in \cite{E7}.

\subsection{Defining the E$_7$-model}\label{sec:defe7}

Consider a smooth variety  $B$, a line bundle $\mathscr{L}\to B$,  and define the projective bundle 
$$\pi: X_0=\mathbb{P}_B[\mathscr{O}_B\oplus \mathscr{L}^{\otimes 2}\oplus \mathscr{L}^{\otimes 3}]\to B.$$
A Weierstrass model is the zero scheme of a section of the bundle\footnote{Here $\mathscr{O}_{X_0}(1)$ is the dual of the tautological line bundle of $X_0$ 
.} $\mathscr{O}_{X_0}(3)\otimes \pi^* \mathscr{L}^{\otimes 6}$.
We can make this more explicit by denoting the relative projective coordinates of $X_0$ as $[z:x:y]$.  Then a Weierstrass model can be written as the vanishing locus\footnote{Given a set of line bundles $\mathscr{L}_i$ with sections $f_i$ we denote their zero scheme $f_1=f_2=\cdots=f_r=0$ as $V(f_1, \ldots, f_r)$.} 
\begin{equation}\label{eq:w1}
V(y^2z-x^3- f xz^2 -g z^3),
\end{equation}
where $f$ is a section of $\mathscr{L}^{\otimes 4}$ and $g$ is  a section of $\mathscr{L}^{\otimes 6}$.
The discriminant and the $j$-invariant are 
 $$
\Delta = 4 f^3 +27 g^2, \quad j=1728 \frac{4f^3}{\Delta}.
$$
The discriminant locus $V(\Delta)$ consists of points in $B$ over which the fiber is singular.

 Let $B$ be a smooth variety and  $S=V(s)$ be a smooth prime divisor in  $B$ given by  the zero locus of a section $s$ of a line bundle $\mathscr{S}$.  
 An E$_7$-model is given by a Weierstrass model such that (see Proposition 4 of \cite{Neron} and Step 9 of Tate's algorithm)
\begin{equation}\label{eq:E7}
 y^2z = x^3 + a s^3 x z^2 + b s^5 z^3,
\end{equation}
where   $a$ is a section of $\mathscr{L}^{\otimes 4} \otimes \mathscr{S}^{-\otimes 3}$, 
and $b$ is a section of $\mathscr{L}^{\otimes 6} \otimes \mathscr{S}^{-\otimes 5}$.  Moreover, we assume that $a$ and $b$ have zero valuation along $S$ and $V(a)$ and $V(b)$ are smooth divisors in $B$ which intersect transversally. 
For this model the discriminant 
\begin{equation}
\Delta=s^9 (4a^3 + 27 b^2 s)
\end{equation}
factorizes into components $S$ and $\Delta'= V(4a^3 + 27b^2 s)$.
The generic fiber over $S$ is of Kodaira type  III$^*$ and that over $\Delta'$ is of type I$_1$. 
The divisor $\Delta'$ has cuspidal singularities at $V(a,b)$ which worsen to triple point singularities over $V(a,b,s)$. $S$ and $\Delta'$ do not intersect transversally,  but rather at the triple points  $(s,a^3)$. At the support of this intersection, we have the following degeneration: 
\begin{equation}
\Delta'\cap S = V(s,a):\quad  \text{III}^*+\text{I}_1\to   \text{II}^*.
\end{equation}

\subsection{Root system of E$_7$  and the weights of its fundamental representation $\mathbf{56}$}  

\label{sec:IE756}

The Lie algebra of type E$_7$ has dimension $133$, and Weyl group of order $2^{10}\cdot 3^4 \cdot 5 \cdot 7$ \cite[Plate VI]{Bourbaki.GLA46}.
The Cartan matrix of E$_7$ is
\begin{equation}
\scalebox{1}{$
\begin{array}{c}
\alpha_1 \\
\alpha_2 \\
\alpha_3 \\
\alpha_4 \\
\alpha_5 \\
\alpha_6 \\
\alpha_7 
\end{array}
\left(
\begin{array}{ccccccc}
 2 & -1 & 0 & 0 & 0 & 0 & 0 \\
 -1 & 2 & -1 & 0 & 0 & 0 & 0 \\
 0 & -1 & 2 & -1 & 0 & 0 & -1 \\
 0 & 0 & -1 & 2 & -1 & 0 & 0 \\
 0 & 0 & 0 & -1 & 2 & -1 & 0 \\
 0 & 0 & 0 & 0 & -1 & 2 & 0 \\
 0 & 0 & -1 & 0 & 0 & 0 & 2 \\
\end{array}
\right)$}
\end{equation}
where the $i$th row gives the coordinates of the simple root $\alpha_i$ in the basis of fundamental weights. 
As compared to Bourbaki's tables,  our ($\alpha_1,\alpha_2,\alpha_3,\alpha_4,\alpha_5,\alpha_6,\alpha_7)$  are 
denoted ($\alpha_1,\alpha_3,\alpha_4,\alpha_5,\alpha_6, \alpha_7, \alpha_2)$, respectively.

The affine  Dynkin diagrams for $\widetilde{\text{E}}_7$ and $\widetilde{\text{E}}_8$ are provided in Figure~\ref{Fig:E7} and Figure~\ref{Fig:E8}, respectively.  The Hasse diagram for the representation $\bf{56}$ of E$_7$ is given in Figure~\ref{Fig:Hasse56}.  The  affine Dynkin diagram of type $\widetilde{\text{E}}_7$ appears as the dual graph of the generic fiber over $S$ for the E$_7$-model.

\begin{figure}[H]
\begin{center}
\scalebox{.9}{
\begin{tikzpicture}
				\node[draw,circle,thick,scale=1,fill=black,label=below:{\scalebox{1.2}{ $\alpha_0$}}] (0) at (0,0){$1$};
				\node[draw,circle,thick,scale=1,label=below:{\scalebox{1.2}{$\alpha_1$}}] (1) at (1.2,0){$2$};
				\node[draw,circle,thick,scale=1,label=below:{\scalebox{1.2}{$\alpha_2$}}] (2) at (2.4,0){$3$};
				\node[draw,circle,thick,scale=1,label=below:{\scalebox{1.2}{$\alpha_3$}}] (3) at (3.6,0){$4$};
				\node[draw,circle,thick,scale=1,label=below:{\scalebox{1.2}{$\alpha_4$}}] (4) at (4.8,0){$3$};
				\node[draw,circle,thick,scale=1,label=below:{\scalebox{1.2}{$\alpha_5$}}] (5) at (6,0){$2$};
				\node[draw,circle,thick,scale=1,label=below:{\scalebox{1.2}{$\alpha_6$}}] (6) at (7.2,0){$1$};
				\node[draw,circle,thick,scale=1, label=above:{\scalebox{1.2}{$\alpha_7$}}] (7) at (3.6,1.2){$2$};
				\draw[thick] (0)--(1)--(2)--(3)--(4)--(5)--(6);
				\draw[thick]  (3)--(7);
					\end{tikzpicture}}
					\caption{Affine  Dynkin diagram of type $\widetilde{\text{E}}_7$, which reduces to the Dynkin diagram of type E$_7$ when the 
					 black node is removed.
					The numbers inside the nodes are the multiplicities of the Kodaira fiber of type III$^*$ and  the Dynkin labels of the highest root. The root $\alpha_1$ is the highest weight of the adjoint representation while $\alpha_6$ is the highest weight of the fundamental representation $\mathbf{56}$. \label{Fig:E7} 
					}
\end{center}
\end{figure}
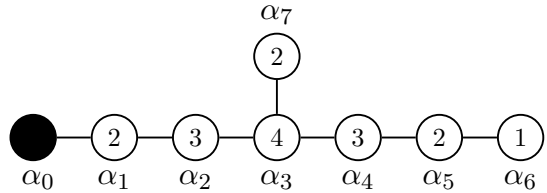

\begin{figure}[H]
\begin{center}
\scalebox{.9}{
\def\arraystretch{1.5}
\begin{tikzpicture}
				\node[draw,circle,thick,scale=1,fill=black,label=below:{\scalebox{1.2}{ $\alpha_0$}}] (0) at (0,0){$1$};
				\node[draw,circle,thick,scale=1,label=below:{\scalebox{1.2}{$\alpha_1$}}] (1) at (1.2,0){$2$};
				\node[draw,circle,thick,scale=1,label=below:{\scalebox{1.2}{$\alpha_2$}}] (2) at (2.4,0){$3$};
				\node[draw,circle,thick,scale=1,label=below:{\scalebox{1.2}{$\alpha_3$}}] (3) at (3.6,0){$4$};
				\node[draw,circle,thick,scale=1,label=below:{\scalebox{1.2}{$\alpha_4$}}] (4) at (4.8,0){$5$};
				\node[draw,circle,thick,scale=1,label=below:{\scalebox{1.2}{$\alpha_5$}}] (5) at (6,0){$6$};
				\node[draw,circle,thick,scale=1,label=below:{\scalebox{1.2}{$\alpha_6$}}] (6) at (7.2,0){$4$};
				\node[draw,circle,thick,scale=1, label=below:{\scalebox{1.2}{$\alpha_7$}}] (7) at (8.4,0){$2$};
				\node[draw,circle,thick,scale=1, label=above:{\scalebox{1.2}{$\alpha_8$}}] (8) at (6,1.2){$3$};
				\draw[thick] (0)--(1)--(2)--(3)--(4)--(5)--(6)--(7);
				\draw[thick]  (5)--(8);
					\end{tikzpicture}}
					\caption{Affine  Dynkin diagram of type $\widetilde{\text{E}}_8$, which reduces to the Dynkin diagram of type E$_8$
					when the black node is removed. 
					The numbers in the nodes are the multiplicities of the Kodaira fiber of type II$^*$. 
					\label{Fig:E8}
					}
\end{center}
\end{figure}
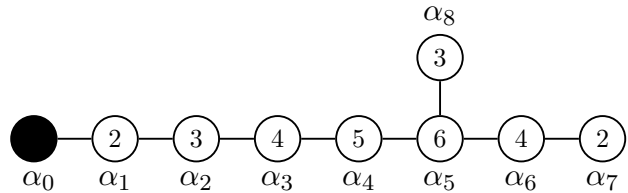

\clearpage

\def\nrx{1} 

\begin{figure}[htb]
\begin{minipage}[c]{0.35\textwidth}
\scalebox{.6}{
\begin{tikzpicture}[ 
x=1.1cm,y=1.4cm,rotate=0,transform shape, color=black]
\tikzstyle{weight}=[circle,thick,draw,minimum size=5mm,inner sep=1pt];
\tikzstyle{root}=[minimum size=0.5cm];
\tikzstyle{sign}=[circle,thick,draw,minimum size=5mm,inner sep=1pt];
\tikzstyle{positive}=[circle,thick,draw,minimum size=5mm,inner sep=1pt,fill=myred];
\tikzstyle{negative}=[circle,thick,draw,minimum size=5mm,inner sep=1pt,fill=myblue];
				\node[positive] (1) at (0,0) {$1$};
			         \node[positive] (2) at (0,-1) {$2$};
			         \node[positive] (3) at (0,-2) {$3$};
			         \node[positive] (4) at (0,-3) {$4$};
				\node[positive] (5) at (0,-4) {$5$};
				
			  \node[positive] (7) at (-1,-5) {$7$} ;         \node[positive] (6) at (1,-5) {$6$};
			         			         \node[positive] (8) at (-2,-6) {$8$} ; \node[positive] (9) at (0,-6) {$9$};
						           \node[positive] (10) at (-1,-7) {\scalebox{.8}{$10$}} ; \node[positive] (11) at (1,-7) {\scalebox{.8}{$11$}};
			           \node[positive] (12) at (0,-8) {\scalebox{.8}{$12$}} ; \node[positive] (13) at (2,-8) {\scalebox{.8}{$13$}};
			           \node[positive] (14) at (1,-9) {\scalebox{.8}{$14$}} ; \node[positive] (15) at (-1,-9) {\scalebox{.8}{$15$}}; ; \node[positive] (16) at (3,-9) {\scalebox{.8}{$16$}};
  \node[positive] (17) at (2,-10) {\scalebox{.8}{$17$}} ; \node[positive] (18) at (0,-10) {\scalebox{.8}{$18$}};
  \node[positive] (21) at (1,-11) {\scalebox{.8}{$21$}}; 
   \node[positive] (22) at (-1,-11) {\scalebox{.8}{$22$}};						         
 \node[positive] (24) at (0,-12) {\scalebox{.8}{$24$}};

   \node[sign] (19) at (4,-10) {\scalebox{.8}{$19$}};						         
  \node[sign] (20) at (3,-11) {\scalebox{.8}{$20$}} ; 
    \node[sign] (23) at (2,-12) {\scalebox{.8}{$23$}} ;

     \node[sign] (25) at (-2,-12) {\scalebox{.8}{$25$}};	
     \node[sign] (26) at (1,-13) {\scalebox{.8}{$26$}} ;   
	  \node[sign] (27) at (0,-13) {\scalebox{.8}{$27$}}; 

	  \node[sign] (28) at (-1,-13) {\scalebox{.8}{$28$}};	

        \node[sign] (29) at (1,-14) {\scalebox{.8}{$29$}} ;

    \node[sign] (30) at (0,-14) {\scalebox{.8}{$30$}} ;

	 \node[sign] (31) at (-1,-14) {\scalebox{.8}{$31$}};

    \node[sign] (32) at (2,-15) {\scalebox{.8}{$32$}} ;   
  
     \node[sign] (34) at (-2,-15) {\scalebox{.8}{$34$}};		

  \node[sign] (37) at (-3,-16) {\scalebox{.8}{$37$}};	
		
	   \node[sign] (38) at (-4,-17) {\scalebox{.8}{$38$}} ;

	\node[negative] (33) at (0,-15) {\scalebox{.8}{$33$}}; 			         
            \node[negative] (35) at (1,-16) {\scalebox{.8}{$35$}} ; \node[negative] (36) at (-1,-16) {\scalebox{.8}{$36$}}; 
           \node[negative] (39) at (0,-17) {\scalebox{.8}{$39$}} ; \node[negative] (40) at (-2,-17) {\scalebox{.8}{$40$}};	
                \node[negative] (41) at (-3,-18) {\scalebox{.8}{$41$}} ; \node[negative] (42) at (1,-18) {\scalebox{.8}{$42$}} ; \node[negative] (43) at (-1,-18) {\scalebox{.8}{$43$}};	
            
               \node[negative] (56) at (0,-27) {\scalebox{.8}{$56$}}; 
			         \node[negative] (55) at (0,-26) {\scalebox{.8}{$55$}}; 
			         \node[negative] (54) at (0,-25) {\scalebox{.8}{$54$}}; 
			         \node[negative] (53) at (0,-24) {\scalebox{.8}{$53$}}; 
				\node[negative] (52) at (0,-23) {\scalebox{.8}{$52$}}; 
			         \node[negative] (51) at (-1,-22) {\scalebox{.8}{$51$}};  \node[negative] (50) at (1,-22) {\scalebox{.8}{$50$}}; ;
			         			         \node[negative] (49) at (2,-21){\scalebox{.8}{$49$}};  \node[negative] (48) at (0,-21)  {\scalebox{.8}{$48$}}; 
						           \node[negative] (47) at (1,-20) {\scalebox{.8}{$47$}} ; \node[negative] (46) at (-1,-20) {\scalebox{.8}{$46$}};
						       \node[negative] (44) at (-2,-19) {\scalebox{.8}{$44$}};      \node[negative] (45) at (0,-19) {\scalebox{.8}{$45$}} ;

\draw[thick]  (1)--   node[root,right] {$\alpha_{6}$}  (2);  
\draw[thick] (2)-- node[root,right] {$\alpha_{5}$} (3);
\draw[thick] (3)-- node[root,right] {$\alpha_{4}$} (4);      

\draw[thick] (4)-- node[root,right] {$\alpha_{3}$} (5);

\draw[thick] (5)-- node[root,right] {$\alpha_{7}$} (6);

\draw[thick] (5)-- node[root,right] {$\alpha_{2}$} (7);

\draw[thick] (6)-- node[root,right] {$\alpha_{2}$} (9);

\draw[thick] (7)-- node[root,right] {$\alpha_{1}$} (8);

\draw[thick] (7)-- node[root,right] {$\alpha_{7}$} (9);

\draw[thick] (8)-- node[root,right] {$\alpha_{7}$} (10);

\draw[thick] (9)-- node[root,right] {$\alpha_{1}$} (10);

\draw[thick] (9)-- node[root,right] {$\alpha_{3}$} (11);

\draw[thick] (10)-- node[root,right] {$\alpha_{3}$} (12);

\draw[thick] (11)-- node[root,right] {$\alpha_{1}$} (12);

\draw[thick] (11)-- node[root,right] {$\alpha_{4}$} (13);

\draw[thick] (12)-- node[root,right] {$\alpha_{4}$} (14);

\draw[thick] (12)-- node[root,right] {$\alpha_{2}$} (15);

\draw[thick] (13)-- node[root,right] {$\alpha_{1}$} (14);

\draw[thick] (13)-- node[root,right] {$\alpha_{5}$} (16);

\draw[thick] (14)-- node[root,right] {$\alpha_{5}$} (17);

\draw[thick] (14)-- node[root,right] {$\alpha_{2}$} (18);

\draw[thick] (15)-- node[root,right] {$\alpha_{4}$} (18);

\draw[thick] (16)-- node[root,right] {$\alpha_{1}$} (17);

\draw[thick] (16)-- node[root,right] {$\alpha_{6}$} (19);

\draw[thick] (17)-- node[root,right] {$\alpha_{6}$} (20);

\draw[thick] (17)-- node[root,right] {$\alpha_{2}$} (21);

\draw[thick] (18)-- node[root,right] {$\alpha_{5}$} (21);

\draw[thick] (18)-- node[root,right] {$\alpha_{3}$} (22);

\draw[line width=1mm] (19)-- node[root,right] {$\alpha_{1}$} (20);

\draw[line width=1mm] (20)-- node[root,right] {$\alpha_{2}$} (23);

\draw[thick] (21)-- node[root,right] {$\alpha_{6}$} (23);

\draw[thick] (21)-- node[root,right] {$\alpha_{3}$} (24);

\draw[thick] (22)-- node[root,right] {$\alpha_{5}$} (24);

\draw[thick] (22)-- node[root,right] {$\alpha_{7}$} (25);

\draw[line width=1mm] (23)-- node[root,right] {$\alpha_{3}$} (26);

\draw[thick] (24)-- node[root,right] {$\alpha_{6}$} (26);

\draw[thick] (24)-- node[root,right, near end] {$\alpha_{4}$} (27);

\draw[thick] (24)-- node[root,left] {$\alpha_{7}$} (28);

\draw[thick] (25)-- node[root,left] {$\alpha_{5}$} (28);

\draw[line width=1mm] (26)-- node[root,right] {$\alpha_{4}$} (29);

\draw[line width=1mm] (26)-- node[root,below, near end] {$\alpha_{7}$} (30);

\draw[thick] (27)-- node[root,above, near start] {$\alpha_{6}$} (29);

\draw[thick] (27)-- node[root,above, near start] {$\alpha_{7}$} (31);

\draw[thick] (28)-- node[root,right] {$\alpha_{6}$} (30);

\draw[thick] (28)-- node[root,left] {$\alpha_{4}$} (31);

\draw[line width=1mm] (29)-- node[root,right] {$\alpha_{5}$} (32);

\draw[thick] (29)-- node[root,right] {$\alpha_{7}$} (33);

\draw[thick] (30)-- node[root,right, near start] {$\alpha_{4}$} (33);

\draw[thick] (31)-- node[root,left] {$\alpha_{6}$} (33);

\draw[thick] (31)-- node[root,left] {$\alpha_{3}$} (34);

\draw[thick] (32)-- node[root,right] {$\alpha_{7}$} (35);

\draw[thick] (33)-- node[root,right] {$\alpha_{5}$} (35);

\draw[thick] (33)-- node[root,right] {$\alpha_{3}$} (36);

\draw[thick] (34)-- node[root,right] {$\alpha_{6}$} (36);

\draw[thick] (34)-- node[root,right] {$\alpha_{2}$} (37);

\draw[thick] (35)-- node[root,right] {$\alpha_{3}$} (39);

\draw[thick] (36)-- node[root,right] {$\alpha_{5}$} (39);

\draw[thick] (36)-- node[root,right] {$\alpha_{2}$} (40);

\draw[thick] (37)-- node[root,right] {$\alpha_{1}$} (38);

\draw[thick] (37)-- node[root,right] {$\alpha_{6}$} (40);

\draw[thick] (38)-- node[root,right] {$\alpha_{6}$} (41);

\draw[thick] (39)-- node[root,right] {$\alpha_{4}$} (42);

\draw[thick] (39)-- node[root,right] {$\alpha_{2}$} (43);

\draw[thick] (40)-- node[root,right] {$\alpha_{1}$} (41);

\draw[thick] (40)-- node[root,right] {$\alpha_{5}$} (43);
\draw[thick] (41)-- node[root,right] {$\alpha_{5}$} (44);

\draw[thick] (42)-- node[root,right] {$\alpha_{2}$} (45);

\draw[thick] (43)-- node[root,right] {$\alpha_{1}$} (44);

\draw[thick] (43)-- node[root,right] {$\alpha_{4}$} (45);

\draw[thick] (44)-- node[root,right] {$\alpha_{4}$} (46);

\draw[thick] (45)-- node[root,right] {$\alpha_{1}$} (46);

\draw[thick] (45)-- node[root,right] {$\alpha_{3}$} (47);

\draw[thick] (46)-- node[root,right] {$\alpha_{3}$} (48);

\draw[thick] (47)-- node[root,right] {$\alpha_{1}$} (48);

\draw[thick] (47)-- node[root,right] {$\alpha_{7}$} (49);

\draw[thick] (48)-- node[root,right] {$\alpha_{7}$} (50);

\draw[thick] (48)-- node[root,right] {$\alpha_{2}$} (51);

\draw[thick] (49)-- node[root,right] {$\alpha_{1}$} (50);

\draw[thick] (50)-- node[root,right] {$\alpha_{2}$} (52);

\draw[thick] (51)-- node[root,right] {$\alpha_{7}$} (52);

\draw[thick] (52)-- node[root,right] {$\alpha_{3}$} (53);

\draw[thick] (53)-- node[root,right] {$\alpha_{4}$} (54);

\draw[thick] (54)-- node[root,right] {$\alpha_{5}$} (55);

\draw[thick] (55)-- node[root,right] {$\alpha_{6}$} (56);

					\end{tikzpicture}
					}
					\end{minipage}
					\hfill
					\begin{minipage}[c]{0.65\textwidth}
					\begin{center}
					\scalebox{.7}{
					
					\begin{tabular}{c}
					\begin{tikzpicture}[scale=.75]
			\tikzmath{\x1 = 2.8;};
				\node[draw,circle,thick,scale=1] (1) at (1.2*\x1,2){$\varpi_{19}$};
				\node[draw,circle,thick,scale=1] (2) at (2.4*\x1,2){$\varpi_{20}$};
				\node[draw,circle,thick,scale=1] (3) at (3.6*\x1,2){$\varpi_{23}$};
				\node[draw,circle,thick,scale=1] (4) at (4.8*\x1,2){$\varpi_{26}$};
				\node[draw,circle,thick,scale=1] (5) at (6*\x1,2){$\varpi_{29}$};
				\node[draw,circle,thick,scale=1] (6) at (7.2*\x1,2){$\varpi_{32}$};
				\node[draw,circle,thick,scale=1] (8) at (4.8*\x1,2*\x1){$\varpi_{30}$};
				\draw[thick, ->] (1)--node[below] {$-\alpha_{1}$}(2);
				\draw[thick, ->] (2)--node[below] {$-\alpha_2$}(3);
				\draw[thick, ->](3)--node[below] {$-\alpha_3$}(4);
				\draw[thick, ->] (4)--node[below] {$-\alpha_4$}(5);
				\draw[thick, ->] (5)--node[below] {$-\alpha_5$}(6);
				\draw[thick,->] (4)--node[right] {$-\alpha_7$}(8);
					\end{tikzpicture}
				
					\vspace{.5cm}
\\
					\\~\\
					$
  \begin{array}{c}
 \varpi_{1}\\
\varpi_{2}\\
\varpi_{3}\\
\varpi_{4}\\
\varpi_{5}\\
\varpi_{6}\\
\varpi_{7}\\
\varpi_{8}\\
\varpi_{9}\\
\varpi_{10}\\
\varpi_{11}\\
\varpi_{12}\\
\varpi_{13}\\
\varpi_{14}\\
\varpi_{15}\\
\varpi_{16}\\
\varpi_{17}\\
\varpi_{18}\\
\varpi_{19}\\
\varpi_{20}\\
\varpi_{21}\\
\varpi_{22}\\
\varpi_{23}\\
\varpi_{24}\\
\varpi_{25}\\
\varpi_{26}\\
\varpi_{27}\\
\varpi_{28}
\end{array}
 \left[
\begin{array}{ccccccc}
 0 & 0 & 0 & 0 & 0 & 1 & 0 \\
 0 & 0 & 0 & 0 & 1 & -1 & 0 \\
 0 & 0 & 0 & 1 & -1 & 0 & 0 \\
 0 & 0 & 1 & -1 & 0 & 0 & 0 \\
 0 & 1 & -1 & 0 & 0 & 0 & 1 \\
 0 & 1 & 0 & 0 & 0 & 0 & -1 \\
 1 & -1 & 0 & 0 & 0 & 0 & 1 \\
 -1 & 0 & 0 & 0 & 0 & 0 & 1 \\
 1 & -1 & 1 & 0 & 0 & 0 & -1 \\
 -1 & 0 & 1 & 0 & 0 & 0 & -1 \\
 1 & 0 & -1 & 1 & 0 & 0 & 0 \\
 -1 & 1 & -1 & 1 & 0 & 0 & 0 \\
 1 & 0 & 0 & -1 & 1 & 0 & 0 \\
 -1 & 1 & 0 & -1 & 1 & 0 & 0 \\
 0 & -1 & 0 & 1 & 0 & 0 & 0 \\
 1 & 0 & 0 & 0 & -1 & 1 & 0 \\
 -1 & 1 & 0 & 0 & -1 & 1 & 0 \\
 0 & -1 & 1 & -1 & 1 & 0 & 0 \\
 1 & 0 & 0 & 0 & 0 & -1 & 0 \\
 -1 & 1 & 0 & 0 & 0 & -1 & 0 \\
 0 & -1 & 1 & 0 & -1 & 1 & 0 \\
 0 & 0 & -1 & 0 & 1 & 0 & 1 \\
 0 & -1 & 1 & 0 & 0 & -1 & 0 \\
 0 & 0 & -1 & 1 & -1 & 1 & 1 \\
 0 & 0 & 0 & 0 & 1 & 0 & -1 \\
 0 & 0 & -1 & 1 & 0 & -1 & 1 \\
 0 & 0 & 0 & -1 & 0 & 1 & 1 \\
 0 & 0 & 0 & 1 & -1 & 1 & -1 \\
\end{array}
\right]
\quad 
 \begin{array}{c}
 \varpi_{56}\\
\varpi_{55}\\
\varpi_{54}\\
\varpi_{53}\\
\varpi_{52}\\
\varpi_{51}\\
\varpi_{50}\\
\varpi_{49}\\
\varpi_{48}\\
\varpi_{47}\\
\varpi_{46}\\
\varpi_{45}\\
\varpi_{44}\\
\varpi_{43}\\
\varpi_{42}\\
\varpi_{41}\\
\varpi_{40}\\
\varpi_{39}\\
\varpi_{38}\\
\varpi_{37}\\
\varpi_{36}\\
\varpi_{35}\\
\varpi_{34}\\
\varpi_{33}\\
\varpi_{32}\\
\varpi_{31}\\
\varpi_{30}\\
\varpi_{29}\\
\end{array}
\left[
\begin{array}{ccccccc}

 0 & 0 & 0 & 0 & 0 & -1 & 0 \\

 0 & 0 & 0 & 0 & -1 & 1 & 0 \\

 0 & 0 & 0 & -1 & 1 & 0 & 0 \\

 0 & 0 & -1 & 1 & 0 & 0 & 0 \\

 0 & -1 & 1 & 0 & 0 & 0 & -1 \\

 0 & -1 & 0 & 0 & 0 & 0 & 1 \\

 -1 & 1 & 0 & 0 & 0 & 0 & -1 \\

 1 & 0 & 0 & 0 & 0 & 0 & -1 \\

 -1 & 1 & -1 & 0 & 0 & 0 & 1 \\

 1 & 0 & -1 & 0 & 0 & 0 & 1 \\

 -1 & 0 & 1 & -1 & 0 & 0 & 0 \\

 1 & -1 & 1 & -1 & 0 & 0 & 0 \\

 -1 & 0 & 0 & 1 & -1 & 0 & 0 \\

 1 & -1 & 0 & 1 & -1 & 0 & 0 \\

 0 & 1 & 0 & -1 & 0 & 0 & 0 \\

 -1 & 0 & 0 & 0 & 1 & -1 & 0 \\

 1 & -1 & 0 & 0 & 1 & -1 & 0 \\

 0 & 1 & -1 & 1 & -1 & 0 & 0 \\

 -1 & 0 & 0 & 0 & 0 & 1 & 0 \\

 1 & -1 & 0 & 0 & 0 & 1 & 0 \\

 0 & 1 & -1 & 0 & 1 & -1 & 0 \\

 0 & 0 & 1 & 0 & -1 & 0 & -1 \\

 0 & 1 & -1 & 0 & 0 & 1 & 0 \\

 0 & 0 & 1 & -1 & 1 & -1 & -1 \\

 0 & 0 & 0 & 0 & -1 & 0 & 1 \\

 0 & 0 & 1 & -1 & 0 & 1 & -1 \\

 0 & 0 & 0 & 1 & 0 & -1 & -1 \\

 0 & 0 & 0 & -1 & 1 & -1 & 1 \\
\end{array}
\right]
$
\\
\\
\end{tabular}
}
\end{center}

\caption{{\it Left:} Hasse diagram for the weights, $\varpi$, of the representation $\bf{56}$ of E$_7$.   A blue (resp. red) node corresponds to a weight for which $\langle \varpi,\phi \rangle$ is always strictly negative (resp. strictly positive). White nodes  correspond to weights such that the form $\langle \varpi,\phi \rangle$ can be either positive or negative. 
 Each chamber of   I($\bf{56}$, E$_7$) is uniquely determined by the signs taken by the white nodes  (see Section \eqref{sec:IE756}). {\it Top right:}  Up to an overall sign, there are seven  weights of the representation $\mathbf{56}$ of E$_7$ which intersect the interior of the dual fundamental Weyl chamber. 
The partial order of weights corresponds to a decorated Dynkin diagram of type E$_7$.   We write $\varpi_i \xrightarrow{-\alpha_\ell} \varpi_j$ to indicate that  $\varpi_i-\alpha_\ell= \varpi_j$. {\it Bottom right:} Weights $\varpi_i$ of the representation $\bf{56}$ expressed in the basis of fundamental weights of E$_7$.
  \label{Fig:SV}
 \label{Fig:Hasse56}}

\end{minipage}
						\end{figure}
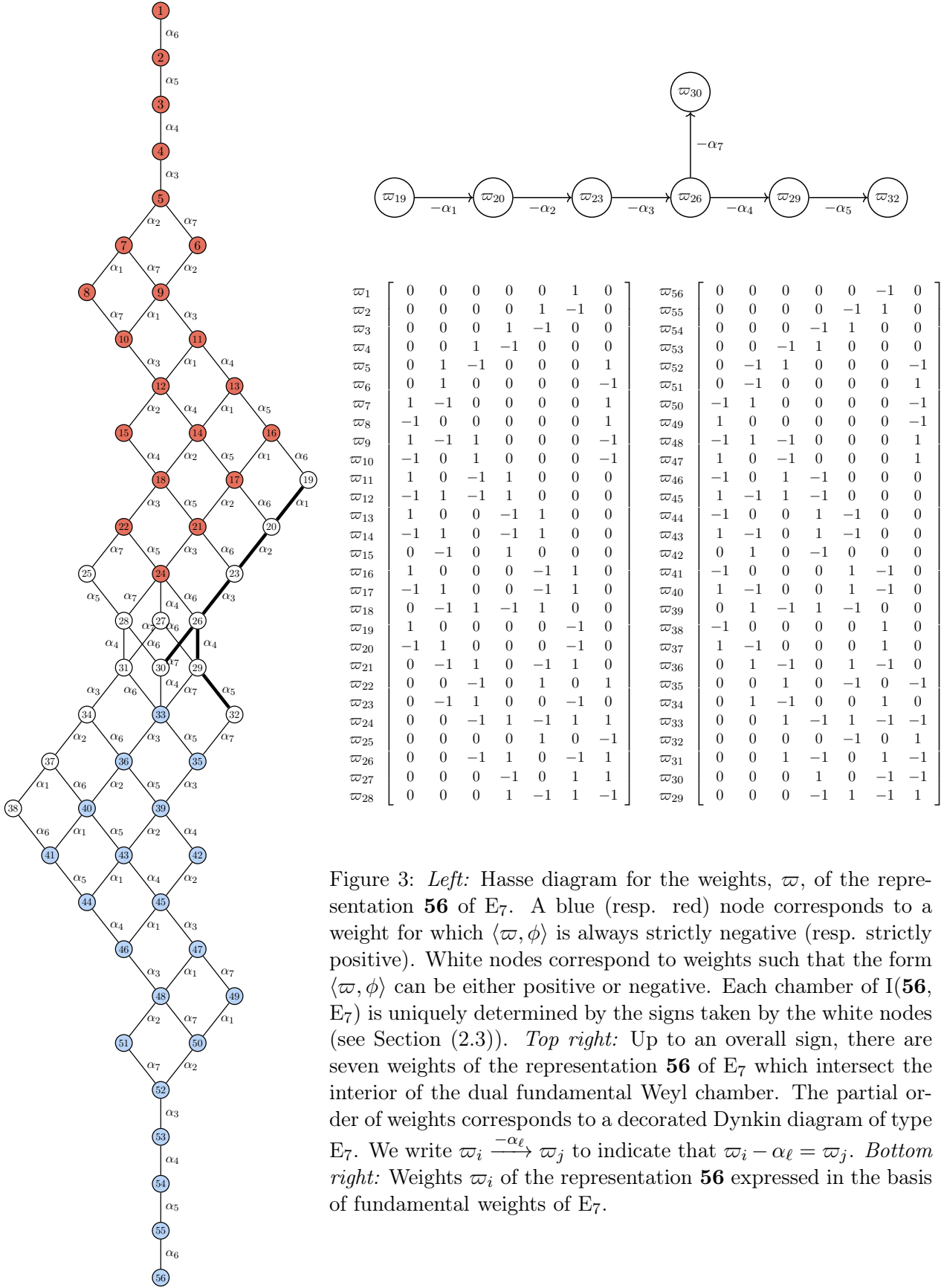

\clearpage

\subsection{Chamber structure of the hyperplane I(E$_7$,$\mathbf{56}$) \label{sec:IE756}}
In this section, we review the structure of the hyperplane arrangement I(E$_7$,$\mathbf{56})$, following the presentation of \cite{E7}. 

 \begin{defn} For a choice of positive simple roots $\alpha_i$, the {\em open dual fundamental Weyl chamber} is the cone of coroots $\phi$ such that:
$\langle \alpha_i, \phi\rangle>0~~ \mathrm{for} ~~i=1,\ldots, r$. \end{defn}
 Given a Lie algebra $\mathfrak{g}$ of rank $r$ and a representation $\bf{R}$ of $\mathfrak{g}$, 
 the set of weights of  $\bf{R}$ is a poset with the usual ordering relation:
$$\varpi_a\preceq \varpi_b\iff \varpi_a-\varpi_b \text{  is a sum of positive roots}.$$ 
The kernel of a weight $\varpi$ is the hyperplane $\varpi^\bot:=\{\phi |\langle \varpi,\phi\rangle =0\}$ in the space of coroots.

\begin{defn}
A weight $\varpi$ of representation $\bf{R}$ is {\em extremal} if the hyperplane $\varpi^\bot$ intersects the interior 
of the dual fundamental Weyl chamber of $\mathfrak{g}$. 
\end{defn}

If we restrict the ambient space to the open dual fundamental Weyl chamber, the only weights giving 
 hyperplanes which intersect this space are the extremal weights, by definition. 
The correspondence between weights and perpendicular hyperplanes is not one-to-one so long as two weights can be parallel.
We therefore make a choice of extremal weight for each hyperplane, fix an order $(\varpi_1,\cdots, \varpi_q)$, and define  a {\em sign vector} $v(\phi)$ whose $k$th entry is  $$v_k(\phi)=\text{Sign}( \langle \varpi_k,\phi\rangle)\quad \text{ where}\quad \text{Sign}(x)=\begin{cases} 
-1\quad if \quad x<0\\
\  0 \quad if \ \quad x=0\\
\  1 \quad if \ \quad x>0.
 \end{cases}$$
 
A simple way to tell if a weight is  extremal  is to write it in the basis of simple roots and use the following theorem.
\begin{thm}
A weight is extremal if and only if at least two of its coefficients in the basis of simple roots have different signs.  
\end{thm}

\begin{defn}
An {\em open chamber} of the hyperplane arrangement  I($\mathfrak{g},\mathbf{R}$) is a  connected component of  the dual open Weyl chamber minus the union of the hyperplanes $\varpi^\bot_m$. 
 \end{defn}
 Each open chamber is uniquely determined by the entries of the sign vector, which take the values $\pm 1$, and are constant within each open chamber.
 
 In particular, the partial order for the weights that are interior walls is (see Figure \ref{Fig:SV}): 
\begin{equation}
\varpi_{19}\succ\varpi_{20}\succ\varpi_{23}\succ\varpi_{26}\succ\varpi_{29}\succ\varpi_{32}, \quad \varpi_{26}\succ\varpi_{30}.
\end{equation}
Our choice of sign vector for the hyperplane arrangement I$(E_7, \mathbf{56})$ is as follows:\footnote{Each  weight of the representation $\mathbf{56}$ has norm square $3/2$ and has scalar product $\pm 1/2$ with any other weight of $\mathbf{56}$. 
Our choice of signs for the entries of the sign vector is such that the  highest weight $\boxed{0\  0 \  0 \  0\  0\  1\  0}$ has a sign $(-1,-1,-1,-1,-1,-1,-1)$.} 
\begin{equation}\label{eq:sgn}
\phi\mapsto ( \langle \varpi_{19}, \phi\rangle , \langle\varpi_{20} , \phi  \rangle, \langle\varpi_{23}, \phi  \rangle, 
\langle \varpi_{26}, \phi \rangle , 
\langle \varpi_{29}, \phi \rangle,
\langle \varpi_{32}, \phi \rangle, \langle\varpi_{30}, \phi \rangle ),
\end{equation}
which expands to 
\begin{equation}\label{eq.SV}
\begin{aligned}
v(\phi)= & \text{Sign}(\phi_1-\phi_6, 
-\phi_1+\phi_2-\phi_6,
-\phi_2+\phi_3-\phi_6, \\
& \   -\phi_3+\phi_4-\phi_6+\phi_7,
-\phi_4+\phi_5-\phi_6+\phi_7,
-\phi_5+\phi_7,
\phi_4-\phi_6-\phi_7
). 
\end{aligned}
\end{equation}

\begin{defn} 
Two chambers $\Pi_1$ and $\Pi_2$ are said to be {\em incident} if they share a common wall $\varpi^\bot_k$, in which case their sign vectors differ only in their $k^{th}$ component. 
\end{defn}
The incidence matrix of the chambers has a dual graph which gives the geography of chambers of the hyperplane arrangement I($\mathfrak{g},\mathbf{R}$).  The incidence graph for the hyperplane arrangement I$(E_7, \mathbf{56})$ is given in Figure~\ref{Figure:IG}.

\begin{thm}
The hyperplane arrangement  I$(E_7, \mathbf{56})$ has eight chambers, each of which is simplicial.  The adjacency graph of the chambers is isomorphic to the Dynkin diagram  of type E$_8$. 
\end{thm}

Explicitly our sign vector in equation \eqref{eq:sgn} obeys the following rules:
\begin{enumerate}
\item The  negative sign  flows as the arrows of Figure \ref{Fig:SV}. 
\item 
The forms $\langle \varpi_{30}, \phi\rangle$ and $\langle\varpi_{29}, \phi\rangle$ cannot both be positive at the same time.
\end{enumerate}
For example, if  $\langle \varpi_{19}, \phi \rangle$ is negative, the same is true of all the $\langle\varpi_{i}, \phi \rangle$ with $i=\{20,23,26,29,32,30\}$. The second rule arises from the fact that  $\varpi_{30}+\varpi_{29}=-\alpha_6$ and  $\langle \alpha_6, \phi\rangle >0$, since we are restricted to the interior of the  dual  fundamental Weyl chamber.
To define a chamber, we just need to name which one of the $\langle \varpi_i, \phi \rangle$ is the first negative one with respect to the order given above. For the case where both $\langle \varpi_{26}, \phi \rangle$ and  $\langle \varpi_{30}, \phi \rangle$  are positive, 
then  $\langle \varpi_{19}, \phi \rangle$,  $\langle \varpi_{20}, \phi \rangle$, and  $\langle \varpi_{23}, \phi \rangle$ are all positive.  Since $\langle \varpi_{30}, \phi \rangle$ is positive, $\langle \varpi_{29}, \phi \rangle$ is necessarily negative, which forces $\langle \varpi_{32}, \phi\rangle$ to also be negative. 
There are exactly  eight possibilities satisfying these two rules.  They are listed in Figure \ref{Figure:Ch}.

~\\

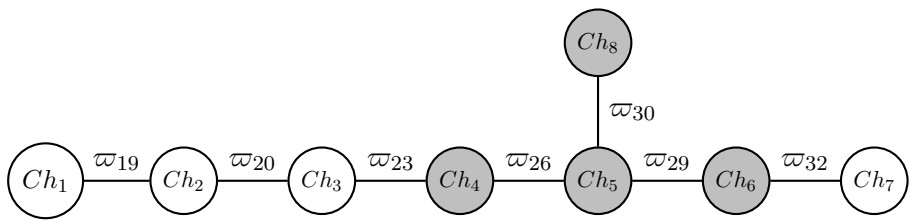
\begin{figure}[H]
{	

\begin{center}
\begin{tikzpicture}[scale=.85]
 \node[draw,circle,thick,scale=.9] (0) at (0,0){$Ch_1$};
				\node[draw,circle,thick,scale=.8] (1) at (1.2*1.8,0){$Ch_2$};
				\node[draw,circle,thick,scale=.8] (2) at (2.4*1.8,0){$Ch_3$};
				\node[draw,circle,thick,scale=.8, fill=lightgray] (3) at (3.6*1.8,0){$Ch_4$};
				\node[draw,circle,thick,scale=.8,fill=lightgray] (4) at (4.8*1.8,0){$Ch_5$};
				\node[draw,circle,thick,scale=.8,fill=lightgray] (5) at (6*1.8,0){$Ch_6$};
				\node[draw,circle,thick,scale=.8] (6) at (7.2*1.8,0){$Ch_7$};
				\node[draw,circle,thick,scale=.8,fill=lightgray] (7) at (4.8*1.8,1.2*1.8){$Ch_8$};
				\draw[thick] (0)-- node[above] {$\varpi_{19}$} (1)--node[above] {$\varpi_{20}$} (2)--node[above] {$\varpi_{23}$} (3)--node[above] {$\varpi_{26}$} (4)--node[above] {$\varpi_{29}$} (5)--node[above] {$\varpi_{32}$} (6);
				\draw[thick]  (4)--   node[right] {$\varpi_{30}$}  (7);
					\end{tikzpicture}
					
					\end{center}
					}

 \caption{Incidence graph of the chambers of the hyperplane arrangement I$(E_7, \mathbf{56})$. 
   A weight $\varpi$ between two nodes indicates that the corresponding  chambers  are separated by the hyperplane $\varpi^\bot$:  for example, one goes from Ch$_1$ to Ch$_2$ by crossing the hyperplane $\varpi_{19}^\bot$.  The colored chambers forming a subgraph of type D$_4$ are those corresponding to the nef-cone of the crepant resolutions constructed by  explicit blowups in  \cite{E7}. 
\label{Figure:IG}}
\end{figure}

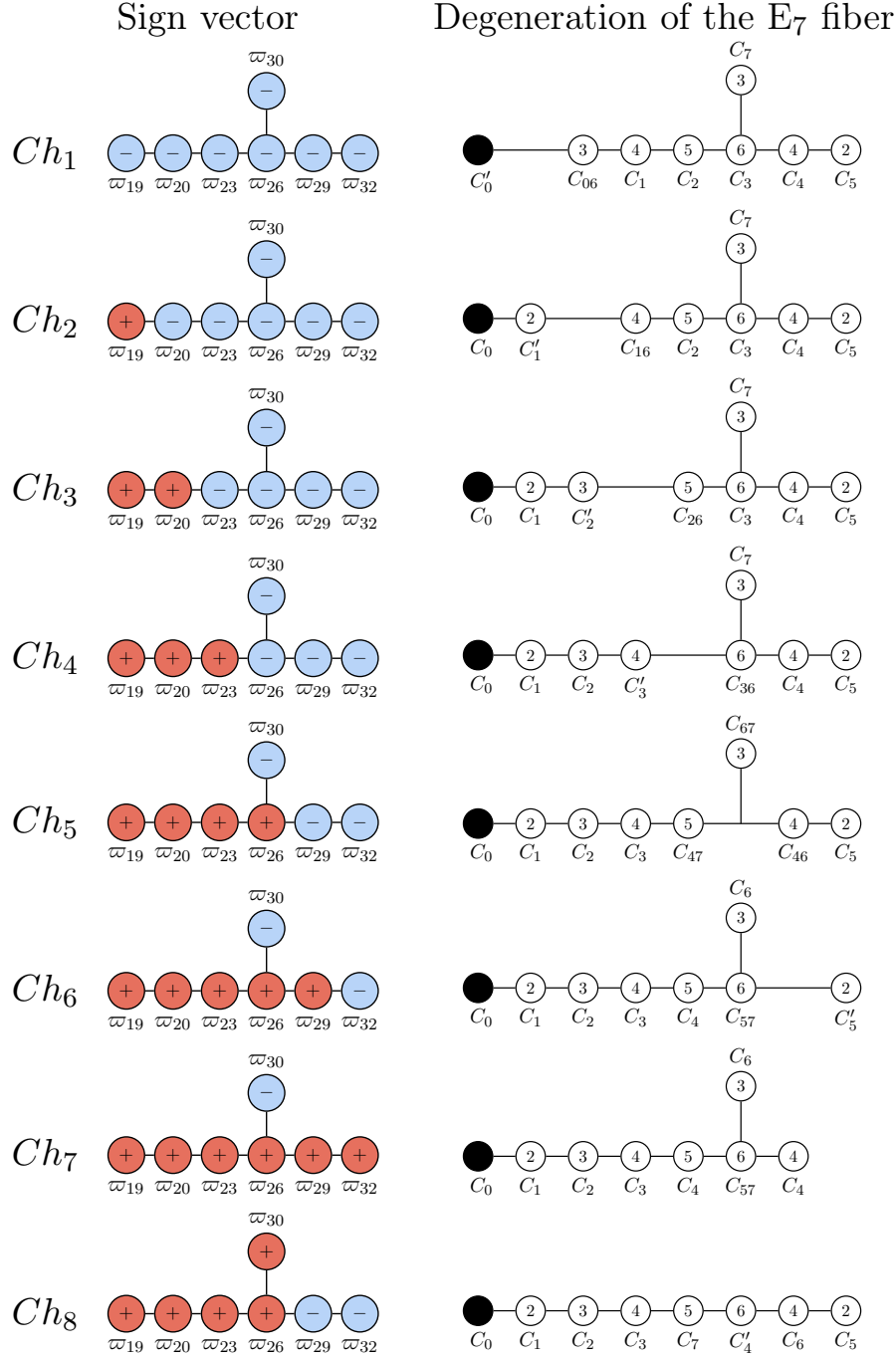
\begin{figure}[!p]
\begin{center}
\scalebox{.65}{
\begin{tabular}{l l}
\quad  \quad\quad\quad\scalebox{2}{\quad Sign vector} & \scalebox{2}{Degeneration of the E$_7$ fiber} \\
\begin{tikzpicture}[scale=.8]
				\node at (-.4,0){\scalebox{2}{$Ch_1$\quad }};
				\node[fill=myblue, draw,circle,thick,scale=1,label=below:{\scalebox{1.2}{$\varpi_{19}$}}] (1) at (1.2,0){$-$};
				\node[fill=myblue,draw,circle,thick,scale=1,label=below:{\scalebox{1.2}{$\varpi_{20}$}}] (2) at (2.4,0){$-$};
				\node[fill=myblue,draw,circle,thick,scale=1,label=below:{\scalebox{1.2}{$\varpi_{23}$}}] (3) at (3.6,0){$-$};
				\node[fill=myblue,draw,circle,thick,scale=1,label=below:{\scalebox{1.2}{$\varpi_{26}$}}] (4) at (4.8,0){$-$};
				\node[fill=myblue,draw,circle,thick,scale=1,label=below:{\scalebox{1.2}{$\varpi_{29}$}}] (5) at (6,0){$-$};
				\node[fill=myblue,draw,circle,thick,scale=1,label=below:{\scalebox{1.2}{$\varpi_{32}$}}] (6) at (7.2,0){$-$};
				\node[fill=myblue,draw,circle,thick,scale=1, label=above:{\scalebox{1.2}{$\varpi_{30}$}}] (8) at (4.8,1.6){$-$};
				\draw[thick] (1)--(2)--(3)--(4)--(5)--(6);
				\draw[thick]  (4)--(8);
					\end{tikzpicture}& 	\quad 
					
					\scalebox{.9}{
\begin{tikzpicture}
				\node[draw,circle,thick,scale=1,fill=black,label=below:{\scalebox{1.2}{ $C'_0$}}] (0) at (0,0){$1$};
				\node[draw,circle,thick,scale=1,label=below:{\scalebox{1.2}{$C_{06}$}}] (2) at (2.4,0){$3$};
				\node[draw,circle,thick,scale=1,label=below:{\scalebox{1.2}{$C_1$}}] (3) at (3.6,0){$4$};
				\node[draw,circle,thick,scale=1,label=below:{\scalebox{1.2}{$C_2$}}] (4) at (4.8,0){$5$};
				\node[draw,circle,thick,scale=1,label=below:{\scalebox{1.2}{$C_3$}}] (5) at (6,0){$6$};
				\node[draw,circle,thick,scale=1,label=below:{\scalebox{1.2}{$C_4$}}] (6) at (7.2,0){$4$};
				\node[draw,circle,thick,scale=1, label=below:{\scalebox{1.2}{$C_5$}}] (7) at (8.4,0){$2$};
				\node[draw,circle,thick,scale=1, label=above:{\scalebox{1.2}{$C_7$}}] (8) at (6,1.6){$3$};
				\draw[thick] (0)--(2)--(3)--(4)--(5)--(6)--(7);
				\draw[thick]  (5)--(8);
					\end{tikzpicture}}

					\\
					
					\begin{tikzpicture}[scale=.8]
				\node at (-.4,0){\scalebox{2}{$Ch_2$\quad }};
				\node[fill=myred, draw,circle,thick,scale=1,label=below:{\scalebox{1.2}{$\varpi_{19}$}}] (1) at (1.2,0){$+$};
				\node[fill=myblue,draw,circle,thick,scale=1,label=below:{\scalebox{1.2}{$\varpi_{20}$}}] (2) at (2.4,0){$-$};
				\node[fill=myblue,draw,circle,thick,scale=1,label=below:{\scalebox{1.2}{$\varpi_{23}$}}] (3) at (3.6,0){$-$};
				\node[fill=myblue,draw,circle,thick,scale=1,label=below:{\scalebox{1.2}{$\varpi_{26}$}}] (4) at (4.8,0){$-$};
				\node[fill=myblue,draw,circle,thick,scale=1,label=below:{\scalebox{1.2}{$\varpi_{29}$}}] (5) at (6,0){$-$};
				\node[fill=myblue,draw,circle,thick,scale=1,label=below:{\scalebox{1.2}{$\varpi_{32}$}}] (6) at (7.2,0){$-$};
				\node[fill=myblue,draw,circle,thick,scale=1, label=above:{\scalebox{1.2}{$\varpi_{30}$}}] (8) at (4.8,1.6){$-$};
				\draw[thick] (1)--(2)--(3)--(4)--(5)--(6);
				\draw[thick]  (4)--(8);
					\end{tikzpicture}& 	\quad
					
					\scalebox{.9}{
\begin{tikzpicture}
				\node[draw,circle,thick,scale=1,fill=black,label=below:{\scalebox{1.2}{ $C_0$}}] (0) at (0,0){$1$};
				\node[draw, circle,thick,scale=1, label=below:{\scalebox{1.2}{$C'_{1}$}}] (1) at (1.2,0){$2$};
			\node[draw,circle,thick,scale=1,label=below:{\scalebox{1.2}{$C_{16}$}}] (3) at (3.6,0){$4$};
				\node[draw,circle,thick,scale=1,label=below:{\scalebox{1.2}{$C_2$}}] (4) at (4.8,0){$5$};
				\node[draw,circle,thick,scale=1,label=below:{\scalebox{1.2}{$C_3$}}] (5) at (6,0){$6$};
				\node[draw,circle,thick,scale=1,label=below:{\scalebox{1.2}{$C_4$}}] (6) at (7.2,0){$4$};
				\node[draw,circle,thick,scale=1, label=below:{\scalebox{1.2}{$C_5$}}] (7) at (8.4,0){$2$};
				\node[draw,circle,thick,scale=1, label=above:{\scalebox{1.2}{$C_7$}}] (8) at (6,1.6){$3$};
				\draw[thick] (0)--(1)--(3)--(4)--(5)--(6)--(7);
				\draw[thick]  (5)--(8);					\end{tikzpicture}}
					\\

										\begin{tikzpicture}[scale=.8]
			\node at (-.4,0){\scalebox{2}{$Ch_3$\quad }};
				\node[fill=myred, draw,circle,thick,scale=1,label=below:{\scalebox{1.2}{$\varpi_{19}$}}] (1) at (1.2,0){$+$};
				\node[fill=myred,draw,circle,thick,scale=1,label=below:{\scalebox{1.2}{$\varpi_{20}$}}] (2) at (2.4,0){$+$};
				\node[fill=myblue,draw,circle,thick,scale=1,label=below:{\scalebox{1.2}{$\varpi_{23}$}}] (3) at (3.6,0){$-$};
				\node[fill=myblue,draw,circle,thick,scale=1,label=below:{\scalebox{1.2}{$\varpi_{26}$}}] (4) at (4.8,0){$-$};
				\node[fill=myblue,draw,circle,thick,scale=1,label=below:{\scalebox{1.2}{$\varpi_{29}$}}] (5) at (6,0){$-$};
				\node[fill=myblue,draw,circle,thick,scale=1,label=below:{\scalebox{1.2}{$\varpi_{32}$}}] (6) at (7.2,0){$-$};
				\node[fill=myblue,draw,circle,thick,scale=1, label=above:{\scalebox{1.2}{$\varpi_{30}$}}] (8) at (4.8,1.6){$-$};
				\draw[thick] (1)--(2)--(3)--(4)--(5)--(6);
				\draw[thick]  (4)--(8);
					\end{tikzpicture}& 	\quad
					\scalebox{.9}{
\begin{tikzpicture}
				\node[draw,circle,thick,scale=1,fill=black,label=below:{\scalebox{1.2}{ $C_0$}}] (0) at (0,0){$1$};
				\node[draw,circle,thick,scale=1,label=below:{\scalebox{1.2}{$C_{1}$}}] (2) at (1.2,0){$2$};
				\node[draw,circle,thick,scale=1,label=below:{\scalebox{1.2}{$C'_{2}$}}] (3) at (2.4,0){$3$};
				\node[draw,circle,thick,scale=1,label=below:{\scalebox{1.2}{$C_{26}$}}] (4) at (4.8,0){$5$};
				\node[draw,circle,thick,scale=1,label=below:{\scalebox{1.2}{$C_3$}}] (5) at (6,0){$6$};
				\node[draw,circle,thick,scale=1,label=below:{\scalebox{1.2}{$C_4$}}] (6) at (7.2,0){$4$};
				\node[draw,circle,thick,scale=1, label=below:{\scalebox{1.2}{$C_5$}}] (7) at (8.4,0){$2$};
				\node[draw,circle,thick,scale=1, label=above:{\scalebox{1.2}{$C_7$}}] (8) at (6,1.6){$3$};
				\draw[thick] (0)--(2)--(3)--(4)--(5)--(6)--(7);
				\draw[thick]  (5)--(8);
					\end{tikzpicture}}

					\\

								\begin{tikzpicture}[scale=.8]
			\node at (-.4,0){\scalebox{2}{$Ch_4$\quad }};
				\node[fill=myred, draw,circle,thick,scale=1,label=below:{\scalebox{1.2}{$\varpi_{19}$}}] (1) at (1.2,0){$+$};
				\node[fill=myred,draw,circle,thick,scale=1,label=below:{\scalebox{1.2}{$\varpi_{20}$}}] (2) at (2.4,0){$+$};
				\node[fill=myred,draw,circle,thick,scale=1,label=below:{\scalebox{1.2}{$\varpi_{23}$}}] (3) at (3.6,0){$+$};
				\node[fill=myblue,draw,circle,thick,scale=1,label=below:{\scalebox{1.2}{$\varpi_{26}$}}] (4) at (4.8,0){$-$};
				\node[fill=myblue,draw,circle,thick,scale=1,label=below:{\scalebox{1.2}{$\varpi_{29}$}}] (5) at (6,0){$-$};
				\node[fill=myblue,draw,circle,thick,scale=1,label=below:{\scalebox{1.2}{$\varpi_{32}$}}] (6) at (7.2,0){$-$};
				\node[fill=myblue,draw,circle,thick,scale=1, label=above:{\scalebox{1.2}{$\varpi_{30}$}}] (8) at (4.8,1.6){$-$};
				\draw[thick] (1)--(2)--(3)--(4)--(5)--(6);
				\draw[thick]  (4)--(8);
					\end{tikzpicture}& \quad	
					\scalebox{.9}{
\begin{tikzpicture}
				\node[draw,circle,thick,scale=1,fill=black,label=below:{\scalebox{1.2}{ $C_0$}}] (0) at (0,0){$1$};
				\node[draw,circle,thick,scale=1,label=below:{\scalebox{1.2}{$C_{1}$}}] (2) at (1.2,0){$2$};
				\node[draw,circle,thick,scale=1,label=below:{\scalebox{1.2}{$C_{2}$}}] (3) at (2.4,0){$3$};
				\node[draw,circle,thick,scale=1,label=below:{\scalebox{1.2}{$C'_{3}$}}] (4) at (3.6,0){$4$};
				\node[draw,circle,thick,scale=1,label=below:{\scalebox{1.2}{$C_{36}$}}] (5) at (6,0){$6$};
				\node[draw,circle,thick,scale=1,label=below:{\scalebox{1.2}{$C_4$}}] (6) at (7.2,0){$4$};
				\node[draw,circle,thick,scale=1, label=below:{\scalebox{1.2}{$C_5$}}] (7) at (8.4,0){$2$};
				\node[draw,circle,thick,scale=1, label=above:{\scalebox{1.2}{$C_7$}}] (8) at (6,1.6){$3$};
				\draw[thick] (0)--(2)--(3)--(4)--(5)--(6)--(7);
				\draw[thick]  (5)--(8);
					\end{tikzpicture}}
					\\

					\begin{tikzpicture}[scale=.8]
			\node at (-.4,0){\scalebox{2}{$Ch_5$\quad }};
				\node[fill=myred, draw,circle,thick,scale=1,label=below:{\scalebox{1.2}{$\varpi_{19}$}}] (1) at (1.2,0){$+$};
				\node[fill=myred,draw,circle,thick,scale=1,label=below:{\scalebox{1.2}{$\varpi_{20}$}}] (2) at (2.4,0){$+$};
				\node[fill=myred,draw,circle,thick,scale=1,label=below:{\scalebox{1.2}{$\varpi_{23}$}}] (3) at (3.6,0){$+$};
				\node[fill=myred,draw,circle,thick,scale=1,label=below:{\scalebox{1.2}{$\varpi_{26}$}}] (4) at (4.8,0){$+$};
				\node[fill=myblue,draw,circle,thick,scale=1,label=below:{\scalebox{1.2}{$\varpi_{29}$}}] (5) at (6,0){$-$};
				\node[fill=myblue,draw,circle,thick,scale=1,label=below:{\scalebox{1.2}{$\varpi_{32}$}}] (6) at (7.2,0){$-$};
				\node[fill=myblue,draw,circle,thick,scale=1, label=above:{\scalebox{1.2}{$\varpi_{30}$}}] (8) at (4.8,1.6){$-$};
				\draw[thick] (1)--(2)--(3)--(4)--(5)--(6);
				\draw[thick]  (4)--(8);
					\end{tikzpicture}	&\quad	
					\scalebox{.9}{
\begin{tikzpicture}
				\node[draw,circle,thick,scale=1,fill=black,label=below:{\scalebox{1.2}{ $C_0$}}] (0) at (0,0){$1$};
				\node[draw,circle,thick,scale=1,label=below:{\scalebox{1.2}{$C_{1}$}}] (2) at (1.2,0){$2$};
				\node[draw,circle,thick,scale=1,label=below:{\scalebox{1.2}{$C_{2}$}}] (3) at (2.4,0){$3$};
				\node[draw,circle,thick,scale=1,label=below:{\scalebox{1.2}{$C_{3}$}}] (4) at (3.6,0){$4$};
				\node[draw,circle,thick,scale=1,label=below:{\scalebox{1.2}{$C_{47}$}}] (5) at (4.8,0){$5$};
				\node[draw,circle,thick,scale=1,label=below:{\scalebox{1.2}{$C_{46}$}}] (6) at (7.2,0){$4$};
				\node[draw,circle,thick,scale=1, label=below:{\scalebox{1.2}{$C_5$}}] (7) at (8.4,0){$2$};
				\node[draw,circle,thick,scale=1, label=above:{\scalebox{1.2}{$C_{67}$}}] (8) at (6,1.6){$3$};
				\draw[thick] (0)--(2)--(3)--(4)--(5)--(6)--(7);
				\draw[thick]  (6,0)--(8);
					\end{tikzpicture}}

\\				\begin{tikzpicture}[scale=.8]
				\node at (-.4,0){\scalebox{2}{$Ch_6$\quad }};
				\node[fill=myred, draw,circle,thick,scale=1,label=below:{\scalebox{1.2}{$\varpi_{19}$}}] (1) at (1.2,0){$+$};
				\node[fill=myred,draw,circle,thick,scale=1,label=below:{\scalebox{1.2}{$\varpi_{20}$}}] (2) at (2.4,0){$+$};
				\node[fill=myred,draw,circle,thick,scale=1,label=below:{\scalebox{1.2}{$\varpi_{23}$}}] (3) at (3.6,0){$+$};
				\node[fill=myred,draw,circle,thick,scale=1,label=below:{\scalebox{1.2}{$\varpi_{26}$}}] (4) at (4.8,0){$+$};
				\node[fill=myred,draw,circle,thick,scale=1,label=below:{\scalebox{1.2}{$\varpi_{29}$}}] (5) at (6,0){$+$};
				\node[fill=myblue,draw,circle,thick,scale=1,label=below:{\scalebox{1.2}{$\varpi_{32}$}}] (6) at (7.2,0){$-$};
				\node[fill=myblue,draw,circle,thick,scale=1, label=above:{\scalebox{1.2}{$\varpi_{30}$}}] (8) at (4.8,1.6){$-$};
				\draw[thick] (1)--(2)--(3)--(4)--(5)--(6);
				\draw[thick]  (4)--(8);
					\end{tikzpicture}	& \quad	
					\scalebox{.9}{
\begin{tikzpicture}
				\node[draw,circle,thick,scale=1,fill=black,label=below:{\scalebox{1.2}{ $C_0$}}] (0) at (0,0){$1$};
				\node[draw,circle,thick,scale=1,label=below:{\scalebox{1.2}{$C_{1}$}}] (2) at (1.2,0){$2$};
				\node[draw,circle,thick,scale=1,label=below:{\scalebox{1.2}{$C_{2}$}}] (3) at (2.4,0){$3$};
				\node[draw,circle,thick,scale=1,label=below:{\scalebox{1.2}{$C_{3}$}}] (4) at (3.6,0){$4$};
				\node[draw,circle,thick,scale=1,label=below:{\scalebox{1.2}{$C_{4}$}}] (5) at (4.8,0){$5$};
				\node[draw,circle,thick,scale=1,label=below:{\scalebox{1.2}{$C_{57}$}}] (6) at (6,0){$6$};
				\node[draw,circle,thick,scale=1, label=below:{\scalebox{1.2}{$C'_5$}}] (7) at (8.4,0){$2$};
				\node[draw,circle,thick,scale=1, label=above:{\scalebox{1.2}{$C_6$}}] (8) at (6,1.6){$3$};
				\draw[thick] (0)--(2)--(3)--(4)--(5)--(6)--(7);
				\draw[thick]  (6)--(8);
					\end{tikzpicture}}
					
					\\
					
					\begin{tikzpicture}[scale=.8]
			\node at (-.4,0){\scalebox{2}{$Ch_7$\quad }};
				\node[fill=myred, draw,circle,thick,scale=1,label=below:{\scalebox{1.2}{$\varpi_{19}$}}] (1) at (1.2,0){$+$};
				\node[fill=myred,draw,circle,thick,scale=1,label=below:{\scalebox{1.2}{$\varpi_{20}$}}] (2) at (2.4,0){$+$};
				\node[fill=myred,draw,circle,thick,scale=1,label=below:{\scalebox{1.2}{$\varpi_{23}$}}] (3) at (3.6,0){$+$};
				\node[fill=myred,draw,circle,thick,scale=1,label=below:{\scalebox{1.2}{$\varpi_{26}$}}] (4) at (4.8,0){$+$};
				\node[fill=myred,draw,circle,thick,scale=1,label=below:{\scalebox{1.2}{$\varpi_{29}$}}] (5) at (6,0){$+$};
				\node[fill=myred,draw,circle,thick,scale=1,label=below:{\scalebox{1.2}{$\varpi_{32}$}}] (6) at (7.2,0){$+$};
				\node[fill=myblue,draw,circle,thick,scale=1, label=above:{\scalebox{1.2}{$\varpi_{30}$}}] (8) at (4.8,1.6){$-$};
				\draw[thick] (1)--(2)--(3)--(4)--(5)--(6);
				\draw[thick]  (4)--(8);
					\end{tikzpicture}	&\quad	
					\scalebox{.9}{
\begin{tikzpicture}
				\node[draw,circle,thick,scale=1,fill=black,label=below:{\scalebox{1.2}{ $C_0$}}] (0) at (0,0){$1$};
				\node[draw,circle,thick,scale=1,label=below:{\scalebox{1.2}{$C_{1}$}}] (2) at (1.2,0){$2$};
				\node[draw,circle,thick,scale=1,label=below:{\scalebox{1.2}{$C_{2}$}}] (3) at (2.4,0){$3$};
				\node[draw,circle,thick,scale=1,label=below:{\scalebox{1.2}{$C_{3}$}}] (4) at (3.6,0){$4$};
				\node[draw,circle,thick,scale=1,label=below:{\scalebox{1.2}{$C_{4}$}}] (5) at (4.8,0){$5$};
				\node[draw,circle,thick,scale=1,label=below:{\scalebox{1.2}{$C_{57}$}}] (6) at (6,0){$6$};
				\node[draw,circle,thick,scale=1,label=below:{\scalebox{1.2}{$C_4$}}] (7) at (7.2,0){$4$};
				\node[draw,circle,thick,scale=1, label=above:{\scalebox{1.2}{$C_6$}}] (8) at (6,1.6){$3$};
				\draw[thick] (0)--(2)--(3)--(4)--(5)--(6)--(7);
				\draw[thick]  (6)--(8);
					\end{tikzpicture}}

					\\

					\begin{tikzpicture}[scale=.8]
			\node at (-.4,0){\scalebox{2}{$Ch_8$\quad }};
				\node[fill=myred, draw,circle,thick,scale=1,label=below:{\scalebox{1.2}{$\varpi_{19}$}}] (1) at (1.2,0){$+$};
				\node[fill=myred,draw,circle,thick,scale=1,label=below:{\scalebox{1.2}{$\varpi_{20}$}}] (2) at (2.4,0){$+$};
				\node[fill=myred,draw,circle,thick,scale=1,label=below:{\scalebox{1.2}{$\varpi_{23}$}}] (3) at (3.6,0){$+$};
				\node[fill=myred,draw,circle,thick,scale=1,label=below:{\scalebox{1.2}{$\varpi_{26}$}}] (4) at (4.8,0){$+$};
				\node[fill=myblue,draw,circle,thick,scale=1,label=below:{\scalebox{1.2}{$\varpi_{29}$}}] (5) at (6,0){$-$};
				\node[fill=myblue,draw,circle,thick,scale=1,label=below:{\scalebox{1.2}{$\varpi_{32}$}}] (6) at (7.2,0){$-$};
				\node[fill=myred,draw,circle,thick,scale=1, label=above:{\scalebox{1.2}{$\varpi_{30}$}}] (8) at (4.8,1.6){$+$};
				\draw[thick] (1)--(2)--(3)--(4)--(5)--(6);
				\draw[thick]  (4)--(8);
					\node (9) at (8.4,0){};
					\end{tikzpicture} 
					&\quad
					\scalebox{.9}{
\begin{tikzpicture}
				\node[draw,circle,thick,scale=1,fill=black,label=below:{\scalebox{1.2}{ $C_0$}}] (0) at (0,0){$1$};
				\node[draw,circle,thick,scale=1,label=below:{\scalebox{1.2}{$C_{1}$}}] (2) at (1.2,0){$2$};
				\node[draw,circle,thick,scale=1,label=below:{\scalebox{1.2}{$C_{2}$}}] (3) at (2.4,0){$3$};
				\node[draw,circle,thick,scale=1,label=below:{\scalebox{1.2}{$C_{3}$}}] (4) at (3.6,0){$4$};
				\node[draw,circle,thick,scale=1,label=below:{\scalebox{1.2}{$C_{7}$}}] (5) at (4.8,0){$5$};
				\node[draw,circle,thick,scale=1,label=below:{\scalebox{1.2}{$C'_{4}$}}] (6) at (6,0){$6$};
				\node[draw,circle,thick,scale=1,label=below:{\scalebox{1.2}{$C_6$}}] (7) at (7.2,0){$4$};
				\node[draw,circle,thick,scale=1, label=below:{\scalebox{1.2}{$C_5$}}] (8) at (8.4,0){$2$};
				\draw[thick] (0)--(2)--(3)--(4)--(5)--(6)--(7)--(8);
			
					\end{tikzpicture}}
\\

					\end{tabular}}
					\end{center}
					\caption{The eight chambers of I(E$_7$, $\bf{56})$. Each chamber is uniquely defined by the signs taken by the seven linear functions $\langle \varpi_i, \phi\rangle$ for $i=\{19, 20,23,26,29,32,30\}$, which together define a sign vector for the hyperplane arrangement. 
					{The left column gives the entries of the sign vector for each chamber. The right column gives the singular fibers observed or expected over V($s,a$). 
					In Chamber $i$, the singular fiber over V$(s,a)$ is expected to have as a dual graph the  affine $\widetilde{\text{E}}_8$ Dynkin diagram with the node $i$ contracted to a point \cite{Box}. The singular fibers are observed directly in an explicit crepant resolution in chambers 4,5,6,8 in \cite{E7} and need to be confirmed geometrically in Chambers 1, 2, 3, and 7. 
					}
					\label{Figure:Ch}}

\end{figure}

\section{Fiber degenerations of an E$_7$-model}\label{sec:chcomp}

Figure \ref{Figure:Ch} also summarizes the structure of the singular fiber over $V(s,a)$ in each chamber. In this section, we explore this degeneration of the fiber III$^*$ (with dual graph the affine $\text{E}_7$ Dynkin diagram) to an incomplete II$^*$ (with dual graph E$_8$).  We begin with the physical motivation for this computation, as well as a a few more definitions needed to explain our strategy.

In a five-dimensional supersymmetric gauge theory with gauge algebra $\mathfrak{g}$ and hypermultiplets transforming in the representation $\bf{R}$ of $\mathfrak{g}$, each chamber of I($\mathfrak{g},\mathbf{R})$  corresponds to a unique Coulomb phase of the Coulomb of the theory. 
Such a gauge theory can be obtained by a compactification of M-theory on an elliptic fibration with associated Lie algebra $\mathfrak{g}$ and representation $\bf{R}$. 

The fibral divisors D$_i$ of the elliptic fibration correspond to the roots $\alpha_i$ of $\mathfrak{g}$.  
In codimension-two, the generic  curve C$_i$ of D$_i$ can degenerate into a collection of rational curves. 
Each of these rational curves has intersections defining a weight $\varpi$, which will be an extremal weight of I($\mathfrak{g},\mathbf{R})$. 
Each crepant resolution of the underlying Weierstrass model, $Y$, corresponds to a relative  minimal model over $Y$.  
Each of these relative minimal models corresponds to a unique chamber of the hyperplane arrangement I($\mathfrak{g},\mathbf{R})$. 
The extremal weights depend on the minimal model.

\begin{defn}
{Given a curve $C$, its associated weight with respect to the fibral divisor D$_i$ is the intersection number  $-D_i \cdot C$.   To any curve $C$, we can associate a weight vector  $\varpi(C)$ with components $\varpi(C)_i= -D_i \cdot C$.
 }
\end{defn}

\begin{rem}
The decomposition of the curve $C_i$ corresponding root $\alpha_i$ in a chamber $\Pi$ with face $\varpi_m^\bot$ are deduced using the linear relations connecting the extremal weights $\varpi_m$.
\end{rem}

\begin{rem}\label{rem11}
 The intersection with $D_0$ can be deduced by linearity, using  $D_0\cong - \sum(m_iD_i)$, where $m_i$ are the Dynkin coefficients of the highest root of $\mathfrak{g}$ and $D_i$ is the fibral divisor corresponding to the root $\alpha_i$.  
If a curve $C$ has negative intersection number with $D_0$, this implies that $C$ is contained in $D_0$.  In this case, $C_0$ will split with $C$ being one of the components. 
\end{rem}

 Figure \ref{Fig:Hasse56}, showed the Hasse diagram of the representation $\bf{56}$, with a clear identification of the simple root between any two adjacent weights.  We summarize the relevant data for the extremal weights here:
\begin{equation}
	\begin{aligned}\begin{array}{llll}
\varpi_{19}&\quad (1, 1, 1, \frac{1}{2}, 0, \text{-}\frac{1}{2}, \frac{1}{2})\quad&\boxed{\ 1\ \ 0\ 0 \ 0 \ 0  \ $-1$ \ \ 0 } &\\
\varpi_{20}=\varpi_{19}-\alpha_1&\quad (0, 1, 1, \frac{1}{2}, 0, \text{-}\frac{1}{2}, \frac{1}{2}) \quad&\boxed{$-1$ \  \  1 \ 0 \ 0  \ 0 \ $-1$ \ \  0 }  & \\
\varpi_{23}=\varpi_{19}-\alpha_1-\alpha_2&\quad (0, 0, 1, \frac{1}{2}, 0, \text{-}\frac{1}{2}, \frac{1}{2}) \quad&\boxed{\ 0 \  $-1$ \ 1 \ 0  \ 0 \ $-1$ \ \  0 } & \\
\varpi_{26}= \varpi_{19}-\alpha_1-\alpha_2-\alpha_3& \quad(0, 0, 0, \frac{1}{2}, 0, \text{-}\frac{1}{2}, \frac{1}{2}) \quad&\boxed{\  0 \ \  0 \ $-1$\ 1   \ 0 $\ -1$ \   1}& \\
\varpi_{29}=\varpi_{19}-\alpha_1-\alpha_2-\alpha_3-\alpha_4& \quad(0, 0, 0, \text{-}\frac{1}{2}, 0, \text{-}\frac{1}{2}, \frac{1}{2}) \quad&\boxed{\  0 \  \ 0 \ 0 \   $-1$  \ $1$ \ $-1$ \ $1$ }& \\
\varpi_{32}=\varpi_{19}-\alpha_1-\alpha_2-\alpha_3-\alpha_4-\alpha_5& \quad(0, 0, 0, \text{-}\frac{1}{2}, \text{-}1, \text{-}\frac{1}{2}, \frac{1}{2}) \quad&\boxed{\ 0 \   \ 0 \ 0 \ 0 \  $-1$\   0\   \ $1$ } &\\
\varpi_{30}=\varpi_{19}-\alpha_1-\alpha_2-\alpha_3-\alpha_7 & \quad(0, 0, 0, \frac{1}{2}, 0,\text{-}\frac{1}{2},\text{-}\frac{1}{2}) \quad&\boxed{\ 0\  \ 0 \ 0 \ $1$  \ 0 \ $-1$ \ $-1$ }&\end{array}\end{aligned}\label{Eq:translation}
	\end{equation}
including their expressions in both the basis of simple roots and the basis of fundamental weights
These two bases are used for different purposes in the analysis of the chambers.

\begin{table}[h!]
\begin{center}
\scalebox{.95}{

\begin{tabular}{| c|  l |l | l |}
\hline
Chambers &\quad Conditions  & \quad\quad\quad \quad\quad Splitting curves & Weights \\
\hline 
Ch$_1$ &   $
\begin{array}{l}
\phi_1-\phi_6<0
\end{array}
$& 
$
\begin{array}{l}
C_0\to C'_0+C_{06}\\
C_6 \to 2C_{06}+2C_1+2C_{2}+2C_3+C_4+C_7
\end{array}
$
 &$
 \begin{array}{l}
  C'_0\to -\varpi_1\\
 C_{06}\to -\varpi_{19}
 \end{array}$
 \\
\hline
Ch$_2$ &
$
\begin{array}{l}
\phi_1-\phi_6>0\\
-\phi_1+\phi_2-\phi_6<0
\end{array}
$

& 
$
\begin{array}{l}
C_1\to C'_{1}+C_{16} \\
C_6 \to 2C_{16}+2C_{2}+2C_3+C_4+C_7\\
\end{array}
$
& 
$
\begin{array}{l}
C'_1\   \to \varpi_{19}\\
C_{16}\to -\varpi_{20}\\
\end{array}
$

\\
\hline
Ch$_3$ &  $
\begin{array}{l}
-\phi_1+\phi_2-\phi_6>0\\
-\phi_2+\phi_3-\phi_6<0\\
\end{array}
$& 
$
\begin{array}{l}
C_2\to C'_{2}+C_{26} \\
C_6 \to 2C_{26}+2C_{3}+C_4+C_7\\
\end{array}
$
& 
$
\begin{array}{l}
C'_2\  \to \varpi_{20}\\
C_{26}\to -\varpi_{23}\\
\end{array}
$
\\
\hline
Ch$_4$ &   $
\begin{array}{l}
-\phi_2+\phi_3-\phi_6>0\\
-\phi_3+\phi_4-\phi_6+\phi_7<0\\
\end{array}
$&$
\begin{array}{l}
C_3\to C'_{3}+C_{36} \\
C_6 \to  2C_{36}+C_{4}+C_7\\
\end{array}
$ & 
$
\begin{array}{l}
C_{3}' \   \to   \varpi_{23}\\
C_{36} \to   -\varpi_{26}\\
\end{array}

$
\\
\hline
Ch$_5$ &
$
\begin{array}{l}
-\phi_3+\phi_4-\phi_6+\phi_7>0\\
\phi_4-\phi_6-\phi_7<0\\
-\phi_4+\phi_5-\phi_6+\phi_7<0
\end{array}
$
&
$
\begin{array}{l}
C_4\to C_{46}+C_{47} \\
C_6 \to C_{46}+C_{67}\\
C_{7}\to C_{47}+C_{67}
\end{array}
$
 & 
 $
\begin{array}{l}
 C_{47}\to  \varpi_{26} \\
 C_{67}\to  -\varpi_{30}\\
 C_{46}\to -\varpi_{29} \\
\end{array} 
 $
 
 \\
 \hline 
Ch$_6$ &
$
\begin{array}{l}
-\phi_4+\phi_5-\phi_6+\phi_7>0\\
-\phi_5+\phi_7<0
\end{array}
$
& $\begin{array}{l}
 C_5\to C_5'+ C_{57}\\
  C_7\to  C_{4}+ C_{6}+2C_{57}
\end{array}$&
$
\begin{array}{l}
C_{57}\to   \varpi_{29}\\
C_5' \ \to  - \varpi_{32}\\
 \end{array}
$
 \\
\hline
Ch$_7$  &  $
\begin{array}{l}
-\phi_4+\phi_5-\phi_6+\phi_7>0\\
-\phi_5+\phi_7>0
\end{array}
$
&
\   \   $C_7\to 2C_7'+ C_{4}+ 2C_{5}+C_{6}$
 &\  \   $C_7'\to -\varpi_{32}$ \\
\hline 
Ch$_8$ &\    $\phi_4-\phi_6-\phi_7>0$&\  \    $C_4  \to2 C'_{4}+   C_{6}+C_{7} $&   \ \   $C_4'\to \varpi_{30}$\\
\hline
\end{tabular}
}
\end{center}
\caption{
Chambers and fiber degenerations of  an E$_7$-model.
The chambers are defined with respect to the interior walls $\varpi^\bot_m$ for $m=19, 20,23,26,29,30$. 
These inequalities are imposed on the interior of the  dual fundamental Weyl chamber $\langle \alpha_i, \phi \rangle >0$ $i=1,2,3,4,5,6,7$.  
All the weights appearing in the right column are weights of the representation $\bf{56}$. 
In chambers Ch$_4$, Ch$_5$, Ch$_6$, and Ch$_8$, the weights are also  obtained geometrically by studying the splitting of curves after a resolution of singularities \cite{E7}. 
 \label{Table:E7.Split} }
\end{table}

~\\

 Our algorithm for determining the fiber degeneration in each chamber consists of the following steps:
 \begin{itemize}
 \item Identifying extremal weights by noting the interior walls of each chamber. 
 This will be some subset of the weights appearing in the sign vector ($\varpi_m$  for $m\in\{19, 20,23,26,29,30\}$), which can be read off of Figure \ref{Figure:IG}.
 
\item  Expressing these extremal weights in the basis of simple roots  (see equation \eqref{Eq:translation}) to identify the degeneration of the components of the generic fiber into rational effective curves.

\item Expressing extremal weights in the basis of fundamental weights (see equation \eqref{Eq:translation}) to get the intersection numbers of the corresponding curve with the fibral divisors. As explained in Remark \ref{rem11},
intersection with D$_0$ can be computed using linearity.  Explicitly,  
for any curve $C$, we have  
\begin{equation}\label{Eq.D0}
D_0 \cdot C=- (2D_1+3 D_2 + 4 D_3 + 3 D_4+2D_5+D_6+2D_7)\cdot C.
\end{equation}
\end{itemize}
The above method is applied to each chamber in Appendix~\ref{sec:fibdegen} and our results are summarized in Table~\ref{Table:E7.Split} and illustrated in Figure~\ref{Figure:DivShape}.  
{ We  confirm the analysis of \cite{Diaconescu:1998cn} and correct few inaccuracies in \cite{Box} such as the splitting rules for the curve C$_6$  in Chamber 1. }

\clearpage

\begin{figure}[!p]
\begin{center}
\scalebox{.7}{
\begin{tabular}{l l}

\scalebox{1.1}{
\begin{tikzpicture}
		\node at (-1,0)		  {\scalebox{2}{C$_0\to$\quad}};
				\node at (-4,0){\scalebox{2}{Ch$_1$\quad }};
				\node[draw,circle,thick,scale=1,label=below:{\scalebox{1.2}{$C'_{0}$}}] (1) at (1.2,0){$1$};
				\node[draw,circle,thick,scale=1,label=below:{\scalebox{1.2}{$C_{06}$}}] (2) at (2.4,0){$1$};
						\draw[thick] (1)--(2);
					\end{tikzpicture}}& 	\quad \quad\quad\quad
					
					\scalebox{1.1}{
\begin{tikzpicture}
\node at (-1,0)   {\scalebox{2}{C$_6\to$\quad}};
				\node[draw,circle,thick,scale=1,label=above:{\scalebox{1.2}{ $C_4$}}] (1b) at (1.2,1){$1$};
				\node[draw,circle,thick,scale=1,label=below:{\scalebox{1.2}{$C_{7}$}}] (1a) at (0,0){$1$};
				\node[draw,circle,thick,scale=1,label=below:{\scalebox{1.2}{$C_3$}}] (2) at (1.2,0){$2$};
				\node[draw,circle,thick,scale=1,label=below:{\scalebox{1.2}{$C_2$}}] (3) at (2.4,0){$2$};
				\node[draw,circle,thick,scale=1,label=below:{\scalebox{1.2}{$C_1$}}] (4) at (3.6,0){$2$};
				\node[draw,circle,thick,scale=1,label=below:{\scalebox{1.2}{$C_{06}$}}] (5) at (4.8,0){$2$};
				\draw[thick]  (1b)--(2);
				\draw[thick]  (1a)--(2)--(3)--(4)--(5);
					\end{tikzpicture}}
					\\
					&
					
					\\

\scalebox{1.1}{
\begin{tikzpicture}
		\node at (-.5,0)		  {\scalebox{2}{C$_1\to$\quad}};
				\node at (-4,0){\scalebox{2}{Ch$_2$\quad }};
				\node[draw,circle,thick,scale=1,label=below:{\scalebox{1.2}{$C'_{1}$}}] (1) at (1.2,0){$1$};
				\node[draw,circle,thick,scale=1,label=below:{\scalebox{1.2}{$C_{16}$}}] (2) at (2.4,0){$1$};
						\draw[thick] (1)--(2);
					\end{tikzpicture}}& 	\quad \quad\quad\quad
					
					\scalebox{1.1}{
\begin{tikzpicture}

\node at (-1,0)   {\scalebox{2}{C$_6\to$\quad}};
				\node[draw,circle,thick,scale=1,label=above:{\scalebox{1.2}{ $C_4$}}] (1b) at (1.2,1){$1$};
				\node[draw,circle,thick,scale=1,label=below:{\scalebox{1.2}{$C_{7}$}}] (1a) at (0,0){$1$};
				\node[draw,circle,thick,scale=1,label=below:{\scalebox{1.2}{$C_3$}}] (2) at (1.2,0){$2$};
				\node[draw,circle,thick,scale=1,label=below:{\scalebox{1.2}{$C_2$}}] (3) at (2.4,0){$2$};
				\node[draw,circle,thick,scale=1,label=below:{\scalebox{1.2}{$C_{16}$}}] (4) at (3.6,0){$2$};
				\draw[thick]  (1b)--(2);
				\draw[thick]  (1a)--(2)--(3)--(4);
					\end{tikzpicture}}
					\\
					&
					
					\\

\scalebox{1.1}{
\begin{tikzpicture}
		\node at (-.5,0)		  {\scalebox{2}{C$_2\to$\quad}};
				\node at (-4,0){\scalebox{2}{Ch$_3$\quad }};
				\node[draw,circle,thick,scale=1,label=below:{\scalebox{1.2}{$C'_{2}$}}] (1) at (1.2,0){$1$};
				\node[draw,circle,thick,scale=1,label=below:{\scalebox{1.2}{$C_{26}$}}] (2) at (2.4,0){$1$};
						\draw[thick] (1)--(2);
					\end{tikzpicture}}& 	\quad \quad\quad\quad
					
					\scalebox{1.1}{
\begin{tikzpicture}
\node at (-1,0)   {\scalebox{2}{C$_6\to$\quad}};
				\node[draw,circle,thick,scale=1,label=above:{\scalebox{1.2}{ $C_4$}}] (1b) at (1.2,1){$1$};
				\node[draw,circle,thick,scale=1,label=below:{\scalebox{1.2}{$C_{7}$}}] (1a) at (0,0){$1$};
				\node[draw,circle,thick,scale=1,label=below:{\scalebox{1.2}{$C_3$}}] (2) at (1.2,0){$2$};
				\node[draw,circle,thick,scale=1,label=below:{\scalebox{1.2}{$C_{26}$}}] (3) at (2.4,0){$2$};
				\draw[thick]  (1b)--(2);
				\draw[thick]  (1a)--(2)--(3);
					\end{tikzpicture}}
					\\
					&
					
					\\

\scalebox{1.1}{
\begin{tikzpicture}
		\node at (-.5,0)		  {\scalebox{2}{C$_3\to$\quad}};
				\node at (-4,0){\scalebox{2}{$\text{Ch}_4$\quad }};
				\node[draw,circle,thick,scale=1,label=below:{\scalebox{1.2}{$C'_{3}$}}] (1) at (1.2,0){$1$};
				\node[draw,circle,thick,scale=1,label=below:{\scalebox{1.2}{$C_{36}$}}] (2) at (2.4,0){$1$};
						\draw[thick] (1)--(2);
					\end{tikzpicture}}& 	\quad \quad\quad \quad
					
					\scalebox{1.1}{
\begin{tikzpicture}
\node at (-1,0)   {\scalebox{2}{C$_6\to$\quad}};
				\node[draw,circle,thick,scale=1,label=above:{\scalebox{1.2}{ $C_4$}}] (1b) at (1.2,1){$1$};
				\node[draw,circle,thick,scale=1,label=below:{\scalebox{1.2}{$C_{7}$}}] (1a) at (0,0){$1$};
				\node[draw,circle,thick,scale=1,label=below:{\scalebox{1.2}{$C_{36}$}}] (2) at (1.2,0){$2$};
				\draw[thick]  (1b)--(2);
				\draw[thick]  (1a)--(2);
					\end{tikzpicture}}
					
					\\
					&
					\\
					&
					\\

\scalebox{1.1}{
\begin{tikzpicture}
		\node at (-.5,0)		  {\scalebox{2}{C$_4\to$\quad}};
				\node at (-4,0){\scalebox{2}{$\text{Ch}_5$\quad }};
				\node[draw,circle,thick,scale=1,label=below:{\scalebox{1.2}{$C_{46}$}}] (1) at (1.2,0){$1$};
				\node[draw,circle,thick,scale=1,label=below:{\scalebox{1.2}{$C_{47}$}}] (2) at (2.4,0){$1$};
						\draw[thick] (1)--(2);
					\end{tikzpicture}}& 	\quad \quad\quad \quad
					
					\scalebox{1.1}{
\begin{tikzpicture}
\node at (-1,0)   {\scalebox{2}{C$_6\to$\quad}};
				\node[draw,circle,thick,scale=1,label=below:{\scalebox{1.2}{$C_{46}$}}] (1a) at (0,0){$1$};
				\node[draw,circle,thick,scale=1,label=below:{\scalebox{1.2}{$C_{67}$}}] (2) at (1.2,0){$1$};
				\draw[thick]  (1a)--(2);
					\end{tikzpicture}}
										\quad \quad\quad \quad
					
					\scalebox{1.1}{
\begin{tikzpicture}
\node at (-1,0)   {\scalebox{2}{C$_7\to$\quad}};
				\node[draw,circle,thick,scale=1,label=below:{\scalebox{1.2}{$C_{47}$}}] (1a) at (0,0){$1$};
				\node[draw,circle,thick,scale=1,label=below:{\scalebox{1.2}{$C_{67}$}}] (2) at (1.2,0){$1$};
				\draw[thick]  (1a)--(2);
					\end{tikzpicture}}
					
					\\					
					
					& \\
					& \\
					&\\

\scalebox{1.1}{
\begin{tikzpicture}
		\node at (-.5,0)		  {\scalebox{2}{C$_5\to$\quad}};
				\node at (-4,0){\scalebox{2}{$\text{Ch}_6$\quad }};
				\node[draw,circle,thick,scale=1,label=below:{\scalebox{1.2}{$C'_{5}$}}] (1) at (1.2,0){$1$};
				\node[draw,circle,thick,scale=1,label=below:{\scalebox{1.2}{$C_{57}$}}] (2) at (2.4,0){$1$};
						\draw[thick] (1)--(2);
					\end{tikzpicture}}& 	\quad \quad\quad \quad
					
					\scalebox{1.1}{
\begin{tikzpicture}
\node at (-1,0)   {\scalebox{2}{C$_7\to$\quad}};
				\node[draw,circle,thick,scale=1,label=below:{\scalebox{1.2}{$C_{4}$}}] (1a) at (0,0){$1$};
				\node[draw,circle,thick,scale=1,label=below:{\scalebox{1.2}{$C_{57}$}}] (2) at (1.2,0){$2$};
				\node[draw,circle,thick,scale=1,label=below:{\scalebox{1.2}{$C_{6}$}}] (3) at (2.4,0){$1$};
				\draw[thick]  (1a)--(2)--(3);
					\end{tikzpicture}}\\
					
						\\
					& \\

\scalebox{1.1}{
\begin{tikzpicture}
		\node at (-.5,0)		  {\scalebox{2}{C$_7\to$\quad}};
				\node at (-4,0){\scalebox{2}{$\text{Ch}_7$\quad }};
				\node[draw,circle,thick,scale=1,label=below:{\scalebox{1.2}{$C_{4}$}}] (1) at (1.2,0){$1$};
				\node[draw,circle,thick,scale=1,label=below:{\scalebox{1.2}{$C_{5}$}}] (2) at (2.4,0){$2$};
				\node[draw,circle,thick,scale=1,label=below:{\scalebox{1.2}{$C'_{7}$}}] (3) at (3.6,0){$2$};
				\node[draw,circle,thick,scale=1,label=above:{\scalebox{1.2}{$C_{6}$}}] (4) at (2.4,1){$1$};
						\draw[thick] (1)--(2)--(3)
						;
						\draw[thick] (2)--(4);
					\end{tikzpicture}}& 	
					\\
						
					& \\
						&\\

\scalebox{1.1}{
\begin{tikzpicture}
		\node at (-.5,0)		  {\scalebox{2}{C$_4\to$\quad}};
				\node at (-4,0){\scalebox{2}{$\text{Ch}_8$\quad }};
				\node[draw,circle,thick,scale=1,label=below:{\scalebox{1.2}{$C_{7}$}}] (1) at (1.2,0){$1$};
				\node[draw,circle,thick,scale=1,label=below:{\scalebox{1.2}{$C'_{4}$}}] (2) at (2.4,0){$2$};
								\node[draw,circle,thick,scale=1,label=below:{\scalebox{1.2}{$C_6$}}] (3) at (3.6,0){$1$};
						\draw[thick] (1)--(2)--(3);
					\end{tikzpicture}}& 		
					\end{tabular}}
					\end{center}
					\caption{
				In each chamber, the decomposition of the III$^*$ fiber is only possible if some of the nodes degenerate. 
				We give the decomposition for all fibers over V$(s,a)$. For more information, see  Appendix \ref{sec:fibdegen}. 								\label{Figure:DivShape}}

\end{figure}
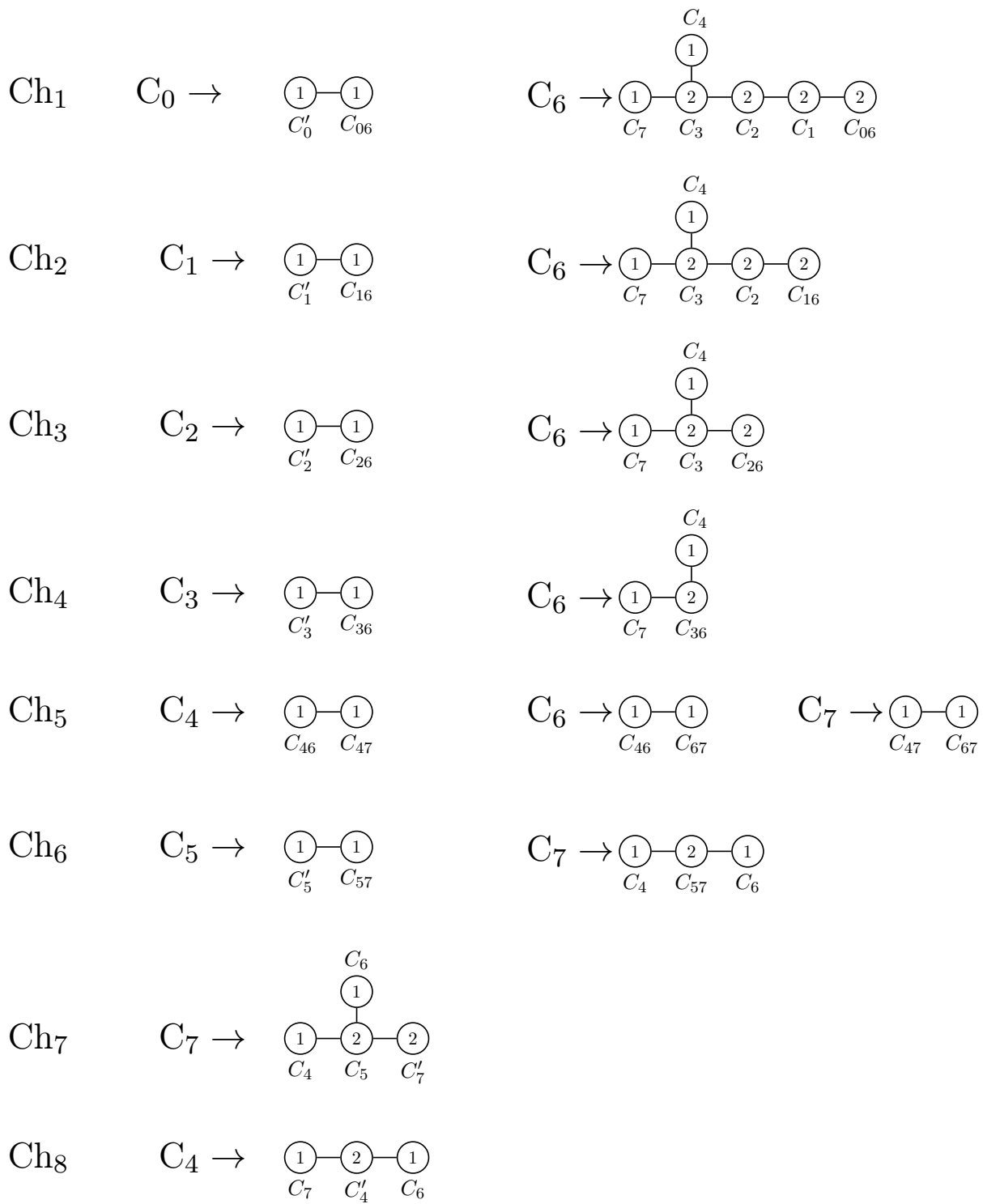

\clearpage

~\\
\begin{table}[H]
\begin{center}
\begin{tabular}{|c|c|c|c|c|c|c|c|}
\hline
Ch$_1$ &  Ch$_2$ & Ch$_3$ & Ch$_4$ & Ch$_5$ & Ch$_6$ & Ch$_7$ & Ch$_8$ \\
\hline
VIII             & VII                & VI               & V                 &    IV             & II               & III                  & I \\
\hline
\end{tabular}
\end{center}
\caption{\label{Table:Dictionary}
Dictionary between our conventions and those of reference \cite{Diaconescu:1998cn}.
 }
\end{table}

In  \cite[Section 4]{Diaconescu:1998cn},  Diaconescu and Entin  identified the Coulomb chambers of an E$_7$-gauge theory with matter transforming  in the representation $\bf{56}$.\footnote{See the inequalities  listed in equations (4.6), (4.10), (4.11), (4.16), (4.19), and (4.21) of \cite[Section 4]{Diaconescu:1998cn}, which are reproduced here in Table  \ref{Table:E7.Split}. }
The dictionary between our conventions (Ch$_i$) for the chambers of I(E$_7$, $\bf{56}$) and those of \cite{Diaconescu:1998cn} (I$-$VIII) is given in Table \ref{Table:Dictionary}.

Reference \cite{Diaconescu:1998cn}  also  gave the degeneration of the components of the fiber III$^*$, however, the description of chamber Ch$_1$ was incomplete as  the  splitting of the affine node was not specified. For Chamber 1 we find (see Appendix \ref{sec:fibdegen} for details) 
\begin{equation}
\text{Degeneration in Chamber 1:}\quad 
\begin{cases}
C_0\to C'_0 +C_{06}, \\
C_6\to 2C_{06}+2C_1+2C_2+2C_3+C_4+C_7.
\end{cases}
\end{equation}

Our results can be compared with other findings in the literature.  In \cite{Box}, the chambers of an E$_7$-model with matter in the representation $\mathbf{56}$ were re-analyzed. In particular, the splitting for the affine node in Chamber 1 was explicitly discussed. Unfortunately, the description of Chamber 1 in \cite{Box} has two important inaccuracies (see  Section \ref{Sec:Ch1}): the splitting of the fiber C$_6$  is incorrect as written as it misses the  component C$_4$,  and the weight of the node of appearing in the degeneration of C$_0$ is also inaccurate. The curve representing the zero node  has weights $[0,0,0,0,0,1,0]$ in the basis of fundamental weights and therefore corresponds to  the highest weight of the representation $\mathbf{56}$. In the notation of \cite{Box}, that should be $L_7$ and not the weight (7).

Reference \cite{DelZotto} examined the geometry of E$_7$-models for which the divisor $S$ supporting the fiber of type III$^*$ is assumed to be a smooth rational curve of self-intersection $-8$ or $-7$. For a ($-8$)-curve, the authors of \cite{DelZotto} conclude that the fibral divisors are Hirzebruch surfaces that intersect transversally. In particular, the fiber III$^*$ does not degenerate to a more singular fiber. However when the curve $S$ is of self-intersection $-7$, there are necessarily singular fibers that carry the weights of the representation $\mathbf{56}$. In our notation, the claim of \cite{DelZotto} is that only one fibral divisor is not Hirzebruch and this divisor  is $D_0$ or $D_6$ and that only one new extremal curve appears in the degeneration. 
However,  the analysis of \cite{DelZotto} does not agree with  any of the chambers of an  E$_7$-model with matter in the representation $\mathbf{56}$ and is in contradiction with both \cite{Diaconescu:1998cn} and \cite{Box}.

~\\

\pagebreak

\section{D$_4$-flops of the E$_7$-model as flops of the orbifold  $\mathbb{C}^3/(\mathbb{Z}_2\times \mathbb{Z}_2)$} \label{dirpf}

In this section, we prove that the flops between Y$_4$, Y$_5$, Y$_6$, and Y$_8$ can be understood as the D$_4$ flops of the crepant resolution of the binomial variety $t^2-u_1 u_2 u_3=0$.
Denote the blowup $X_{i+1}\to X_i$ along the ideal $(f_1,f_2,\ldots,f_n)$ with exceptional divisor $E$ as:
$$\begin{tikzcd}[column sep=2.4cm]X_i \arrow[leftarrow]{r} {\displaystyle (f_1,\ldots, f_n|E)}  & X_{i+1}\end{tikzcd},$$
where $X_0$ is the projective bundle in which the Weierstrass model is defined. 
 We consider the following tree of  blowups:
\vspace{-1.5em}
\begin{equation}\label{eq:finalbl}
\scalebox{.8}{
\begin{tikzpicture}[baseline= (a).base]
\node[scale=1] (a) at (0,0) {
\begin{tikzcd}[column sep=1.7cm, ampersand replacement=\&]
\& \& \& \& \&       \text{X}_5'' \arrow[leftarrow]{r}[above]{\displaystyle (y,e_4|e_6)} \&   \text{X}_6'' \arrow[leftarrow]{r}[above]{\displaystyle (e_4,e_6|e_7)}  \&  \text{X}_{7}'' 
\\
\& \& \& \& \& \& \&  \text{X}_{7}^+ 
\\
X_0 \arrow[leftarrow]{r} {\displaystyle (x,y,s|e_1)} \& X_1 \arrow[leftarrow]{r} {\displaystyle (x,y,e_1|e_2)} \&  X_2\arrow[leftarrow]{r} {\displaystyle (y,e_1|e_3)} \&  X_3  \arrow[leftarrow]{r} {\displaystyle (e_2,e_3|e_4)} \&  
X_4 
\arrow[leftarrow,sloped]{ruu}[above]{\displaystyle (y,e_2|e_5)}
\arrow[leftarrow]{r} {\displaystyle (y,e_4|e_5)} \&  X_5  \arrow[leftarrow]{r} {\displaystyle (y,e_2|e_6)}  \arrow[leftarrow, sloped]{ddr}[below]{\displaystyle (e_2,e_5|e_6)}\&  X_6 
\arrow[leftarrow,sloped]{ru}[above]{\displaystyle (e_2,e_5|e_7)}  
\arrow[leftarrow,sloped]{rd}[below]{\displaystyle (e_2,e_6|e_7)}
\& \\
\& \& \& \& \& \& \& \text{X}_7^-\\
\& \& \& \& \& \&\text{X}_6' \arrow[leftarrow]{r}[above]{\displaystyle (y,e_2|e_7)} \& \text{X}_7'
\end{tikzcd} 
}
;
\end{tikzpicture}}
\vspace{-.5em}
\end{equation}
where Y$_4$, Y$_5$, Y$_6$, and Y$_8$ correspond, respectively, to the proper transforms  X$_7''$, X$_7^+$, X$_7^-$, and  X$_7'$. 
They each stem from a crepant resolution of the partial resolution X$_4$: 
\begin{equation}
\widetilde{Y}:e_3  y^2 -e_1 e_2 e_4 (a e_1 e_3 s^3 x+b e_1^2 e_3^2 e_4  s^5+e_2 x^3)=0
\end{equation}
with projective coordinates: 
\begin{equation}
[e_1e_2^2e_3e_4^3x:e_1e_2^2e_3^2e_4^4y:z][e_2e_4x :e_2e_3e_4^2y :s] [x: e_3e_4y : e_1e_3e_4]  [y:e_1][e_2:e_3].
\end{equation}
The  singularities are at $e_1e_3\neq 0$. 
In that patch, $\widetilde{Y}$ has the singularities of the binomial variety 
\begin{equation}
\mathbb{C}[u_1,u_2,u_3,t]/(t^2- u_1 u_2 u_3 ),
\end{equation}
which is isomorphic to the orbifold $\mathbb{C}^3/(\mathbb{Z}_2\times \mathbb{Z}_2)$. Here the discrete group $\mathbb{Z}_2\times \mathbb{Z}_2$ is generated by $(u_1, u_2, u_3)\to (-u_1, u_2, -u_3)$ and $(u_1, u_2, u_3)\to (u_1,- u_2, -u_3)$. This variety is known to have four crepant resolutions whose flops form a D$_4$ Dynkin diagram as shown in Figure \ref{Fig:C322}. 

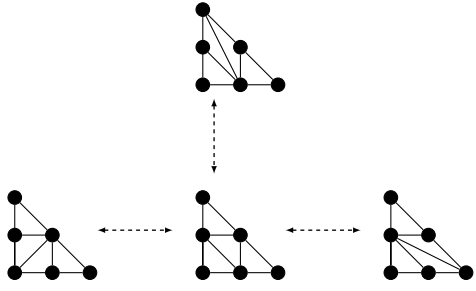
\begin{figure}[H]
\begin{center}
\vspace{-1em}
\scalebox{.5}{
\begin{tikzpicture}
			\draw[dashed, latex-latex, line width=1 pt] (1.2,0)--(3.2,0);  \draw[dashed, latex-latex, line width=1 pt]  (6.2,0)--(8.2,0); \draw[dashed, latex-latex, line width=1 pt] (4.3,3.5)--(4.3,1.5);
					
					\node (1) at (0,0)	{\begin{tikzpicture}
				\node[draw,circle,thick,scale=1,fill=black,label=above:{}] (6) at (0,2){};
				\node[draw,circle,thick,scale=1,fill=black,label=above:{}] (1) at (0,0){};
				\node[draw,circle,thick,scale=1,fill=black,label=above:{}] (2) at (1,0){};
				\node[draw,circle,thick,scale=1,fill=black,label=above:{}] (3) at (2,0){};
				\node[draw,circle,thick,scale=1,fill=black,label=above:{}] (5) at (0,1){};
				\node[draw,circle,thick,scale=1,fill=black,label=above:{}] (4) at (1,1){};
				\draw[thick] (1) to (2) to (3) to (4) to (5)  to (1) to (6) to (4);
				\draw[thick] (1) to (4) to (2);
			\end{tikzpicture}};
					\node (2) at (5,0)	{\begin{tikzpicture}
				\node[draw,circle,thick,scale=1,fill=black,label=above:{}] (6) at (0,2){};
				\node[draw,circle,thick,scale=1,fill=black,label=above:{}] (1) at (0,0){};
				\node[draw,circle,thick,scale=1,fill=black,label=above:{}] (2) at (1,0){};
				\node[draw,circle,thick,scale=1,fill=black,label=above:{}] (3) at (2,0){};
				\node[draw,circle,thick,scale=1,fill=black,label=above:{}] (5) at (0,1){};
				\node[draw,circle,thick,scale=1,fill=black,label=above:{}] (4) at (1,1){};
				\draw[thick] (1) to (2) to (3) to (4) to (5)  to (1) to (6) to (4);
				\draw[thick] (2) to (4);   \draw[thick] (5) to (2);
			\end{tikzpicture}};
			\node (3) at (10,0) 
			{
			\begin{tikzpicture}
				\node[draw,circle,thick,scale=1,fill=black,label=above:{}] (6) at (0,2){};
				\node[draw,circle,thick,scale=1,fill=black,label=above:{}] (1) at (0,0){};
				\node[draw,circle,thick,scale=1,fill=black,label=above:{}] (2) at (1,0){};
				\node[draw,circle,thick,scale=1,fill=black,label=above:{}] (3) at (2,0){};
				\node[draw,circle,thick,scale=1,fill=black,label=above:{}] (5) at (0,1){};
				\node[draw,circle,thick,scale=1,fill=black,label=above:{}] (4) at (1,1){};
				\draw[thick] (1) to (2) to (3) to (4) to (5)  to (1) to (6) to (4);
				\draw[thick] (5) to (3);
				\draw[thick] (5) to (2);
			\end{tikzpicture}};
			\node (4) at (5,5)	{\begin{tikzpicture}
				\node[draw,circle,thick,scale=1,fill=black,label=above:{}] (6) at (0,2){};
				\node[draw,circle,thick,scale=1,fill=black,label=above:{}] (1) at (0,0){};
				\node[draw,circle,thick,scale=1,fill=black,label=above:{}] (2) at (1,0){};
				\node[draw,circle,thick,scale=1,fill=black,label=above:{}] (3) at (2,0){};
				\node[draw,circle,thick,scale=1,fill=black,label=above:{}] (5) at (0,1){};
				\node[draw,circle,thick,scale=1,fill=black,label=above:{}] (4) at (1,1){};
				\draw[thick] (1) to (2);
				\draw[thick] (2) to (3);
				\draw[thick]  (3) to (4);
				\draw[thick]  (4) to (2);
				
							\draw[thick]  (1) to (6) to (4);
				\draw[thick] (2) to (5);   \draw[thick] (6) to (2);
			\end{tikzpicture}};
			\end{tikzpicture}}
			\end{center}
			\caption{Flops between the four crepant resolutions of the singularity $\mathbb{C}[u_1,u_2,u_3,t]/(t^2- u_1 u_2 u_3 )$. 
			\label{Fig:C322}
			 }
			\end{figure}
			
\vspace{-1em}
This provides a direct proof of the $D_4$ flop structure between the chambers examined in~\cite{E7}. There we considered a different set of blowups to resolve the four shaded chambers in Figure~\ref{Figure:IG}. To understand the flops between Y$_4$, Y$_5$, and Y$_8$, we considered
\begin{equation}\label{Ch458}
\begin{tikzpicture}[baseline= (a).base]
\node[scale=.75] (a) at (0,0) {
\begin{tikzcd}[column sep=1.7cm, ampersand replacement=\&]
\& \& \& \& \& \& \&  \text{X}_{7}^+ 
\\
X_0 \arrow[leftarrow]{r} {\displaystyle (x,y,s|e_1)} \& X_1 \arrow[leftarrow]{r} {\displaystyle (x,y,e_1|e_2)} \&  X_2\arrow[leftarrow]{r} {\displaystyle (y,e_1,e_2|e_3)} \&  X_3  \arrow[leftarrow]{r} {\displaystyle (y,e_2|e_4)} \&  
X_4 \arrow[leftarrow]{r} {\displaystyle (e_2,e_4|e_5)} \&  X_5  \arrow[leftarrow]{r} {\displaystyle (y,e_3|e_6)}  \arrow[leftarrow]{ddr}[left]{\displaystyle (e_3,e_4|e_6)}\&  X_6 
\arrow[leftarrow,sloped]{ru}[above]{\displaystyle (e_3,e_6|e_7)}  
\arrow[leftarrow,sloped]{rd}[below]{\displaystyle (e_3,e_4|e_7)}
\& \\
\& \& \& \& \& \& \& \text{X}_7^-\\
\& \& \& \& \& \&\text{X}_6' \arrow[leftarrow]{r}[above]{\displaystyle (y,e_3|e_7)} \& \text{X}_7'
\end{tikzcd} 
};
\end{tikzpicture}
\end{equation}
where Y$_4$, Y$_5$, and Y$_8$ are the proper transforms of $X_7^+$, $X_7^-$, and $X_7'$,  respectively.
To understand the flops between Y$_4$, Y$_5$, and Y$_6$, we considered 
\begin{equation}\label{Ch456}
\begin{tikzpicture}[baseline= (a).base]
\node[scale=.75] (a) at (0,0) {
\begin{tikzcd}[column sep=1.8cm, ampersand replacement=\&]
 \& \& \& \& \& \&  \text{X}_{6}^+ 
\\
X_0 \arrow[leftarrow]{r} {\displaystyle (x,y,s|e_1)} \& X_1
 \arrow[leftarrow]{r} {\displaystyle (y,e_1|e_2)} 
  \&  X_2
  \arrow[leftarrow]{r} {\displaystyle (x,y,e_2|e_3)} 
  \&  X_3 
   \arrow[leftarrow]{r} {\displaystyle (x,e_2,e_3|e_4)} \& X_4 \arrow[leftarrow]{r} {\displaystyle (e_2,e_4|e_5)}
   \arrow[leftarrow,sloped]{rdd}[below] {\displaystyle (e_2,e_3|e_5)}
    \&  X_5
\arrow[leftarrow,sloped]{ru}[above]{\displaystyle (e_2,e_5|e_6)}  
\arrow[leftarrow,sloped]{rd}[below]{\displaystyle (e_2,e_3|e_6)}
\& \\ 
\& \& \& \& \& \& \text{X}_6^-\\
 \& \& \& \& \&\text{X}_5' \arrow[leftarrow]{r}[above]{\displaystyle (e_2,e_4|e_6)} \& \text{X}_6'
\end{tikzcd} 
};
\end{tikzpicture}
\end{equation}
where 
Y$_4$, Y$_5$, and  Y$_6$  are the proper transforms of $X_6^+$ , $X_6^-$, and $X'_6$, respectively.  In addition to allowing an explicit proof of the D$_4$-flops of the E$_7$-model, the new sequence of blowups~(\ref{eq:finalbl}) will simplify our triple intersection computations in Section~\ref{tripint}.\footnote{\label{addblfootnt}
While blowups~(\ref{Ch456}) and~(\ref{Ch456}) were sufficient to resolve the respective varieties and study the flop structure in~\cite{E7}, if one were to continue to use them to compute triple intersection numbers via the techniques laid out in Section~\ref{Sec:Intersection}, one would need to append the following additional blowups
\begin{equation}\label{addbl}
\begin{tikzcd}[column sep=1.4cm, ampersand replacement=\&]
X^+_7 \arrow[leftarrow]{r} {\displaystyle (e_1,e_6|e_8)} \& X^+_8
\end{tikzcd} ,~~~
\begin{tikzcd}[column sep=1.4cm, ampersand replacement=\&]
X^-_7 \arrow[leftarrow]{r} {\displaystyle (y,e_1|e_8)} \& X^-_8
\end{tikzcd} ,~~~
\begin{tikzcd}[column sep=1.4cm, ampersand replacement=\&]
X'_6 \arrow[leftarrow]{r} {\displaystyle (x,e_2|e_7)} \& X''_7
\end{tikzcd} ,~~~
\begin{tikzcd}[column sep=1.4cm, ampersand replacement=\&]
X'_7 \arrow[leftarrow]{r} {\displaystyle (y,e_1|e_8)} \& X'_8
\end{tikzcd} 
\end{equation}
to compute triple intersections in $Ch_4$, $Ch_5$, $Ch_6$, and $Ch_8$, respectively.
This is safe to do since blowing-up a variety along a smooth locus is an isomorphism. Such additional blowups can be interpreted as auxiliary blowups requested by excess intersections.  Upon doing so we can use the exceptional divisor as a clean Cartier divisor.
}

\section{Triple intersection numbers}\label{tripint}

\label{Sec:Intersection}

We begin with some pushforward theorems needed to perform our triple intersection computations.

\begin{defn}[Resolution of singularities]
A resolution of singularities of a variety $Y$ is a proper birational morphism $\varphi:\widetilde{Y}\longrightarrow Y$  such that  
$\widetilde{Y}$ is nonsingular
and  $\varphi$ is an isomorphism away  from the singular  locus of $Y$. 
\end{defn}

\begin{defn}[Crepant birational map]
A  birational map $\varphi:\widetilde{Y}\to Y$ between two algebraic varieties with  $\mathbb{Q}$-Cartier canonical classes is said to be {\em crepant} if it preserves the canonical class. 
\end{defn}

\begin{thm}[Aluffi, {
{\cite[Lemma 1.3]{Aluffi_CBU}}}]
\label{Thm:AluffiCBU}
Let $Z\subset X$ be the  complete intersection  of $d$ nonsingular hypersurfaces $Z_1$, \ldots, $Z_d$ meeting transversally in $X$.  Let  $f: \widetilde{X}\longrightarrow X$ be the blowup of $X$ centered at $Z$. We denote the exceptional divisor of $f$  by $E$. The total Chern class of $\widetilde{X}$ is then:
$$
c( T{\widetilde{X}})=(1+E) \left(\prod_{i=1}^d  \frac{1+f^* Z_i-E}{1+ f^* Z_i}\right)  f^* c(TX).
$$
\end{thm}

\begin{thm}[Esole--Jefferson--Kang,  see  {\cite{Euler}}] \label{Thm:Push}
    Let the nonsingular variety $Z\subset X$ be a complete intersection of $d$ nonsingular hypersurfaces $Z_1$, \ldots, $Z_d$ meeting transversally in $X$. Let $E$ be the class of the exceptional divisor of the blowup $f:\widetilde{X}\longrightarrow X$ centered 
at $Z$.
 Let $\widetilde{Q}(t)=\sum_a f^* Q_a t^a$ be a formal power series with $Q_a\in A_*(X)$.
 We define the associated formal power series  ${Q}(t)=\sum_a Q_a t^a$, whose coefficients pullback to the coefficients of $\widetilde{Q}(t)$. 
 Then the pushforward $f_*\widetilde{Q}(E)$ is
 $$
  f_*  \widetilde{Q}(E) =  \sum_{\ell=1}^d {Q}(Z_\ell) M_\ell, \quad \text{where} \quad  M_\ell=\prod_{\substack{m=1\\
 m\neq \ell}}^d  \frac{Z_m}{ Z_m-Z_\ell }.
 $$ 
\end{thm}

\begin{thm}[{See  \cite{Euler} and  \cite{AE1,AE2,Fullwood:SVW,EKY}}]\label{Thm:PushH}
Let $\mathscr{L}$ be a line bundle over a variety $B$ and $\pi: X_0=\mathbb{P}[\mathscr{O}_B\oplus\mathscr{L}^{\otimes 2} \oplus \mathscr{L}^{\otimes 3}]\longrightarrow B$ a projective bundle over $B$. 
 Let $\widetilde{Q}(t)=\sum_a \pi^* Q_a t^a$ be a formal power series in  $t$ such that $Q_a\in A_*(B)$. Define the auxiliary power series $Q(t)=\sum_a Q_a t^a$. 
Then 
$$
\pi_* \widetilde{Q}(H)=-2\left. \frac{{Q}(H)}{H^2}\right|_{H=-2L}+3\left. \frac{{Q}(H)}{H^2}\right|_{H=-3L}  +\frac{Q(0)}{6 L^2},
$$
 where  $L=c_1(\mathscr{L})$ and $H=c_1(\mathscr{O}_{X_0}(1))$ is the first Chern class of the dual of the tautological line bundle of  $ \pi:X_0=\mathbb{P}(\mathscr{O}_B \oplus\mathscr{L}^{\otimes 2} \oplus\mathscr{L}^{\otimes 3})\rightarrow B$.
\end{thm}
The above theorems are enough for most applications of intersection theory to elliptic fibrations.  
Since our blowups involve regular sequences (c.f. Fulton) of length two or three, we can use:
\begin{thm}[{See \cite[Lemma 3.4]{Euler} }]\label{Thm:Lem34}
For a blowup $f$ with center $(Z_1, Z_2)$, exceptional divisor $E$
\begin{align}\nonumber
f_* E=0, \quad f_* E^2=-Z_1 Z_2, \quad f_*E^3= -(Z_1 +Z_2)Z_1 Z_2, \quad f_* E^4=-(Z_1^2+Z_2^2+ Z_1 Z_2) Z_1 Z_2.
\end{align}
For a blowup $f$ with the complete intersection $(Z_1, Z_2, Z_3)$ as its center and exceptional divisor $E$
\begin{align}\nonumber
f_* E=0, \quad f_* E^2=0, \quad f_*E^3= Z_1 Z_2 Z_3, \quad f_* E^4=(Z_1+Z_2) Z_1 Z_2 Z_3.
\end{align}
\end{thm}

We are now ready to compute the triple intersection numbers of the fibral divisors  for the cases for which an explicit crepant resolution is available.
The triple intersection polynomial is by definition 
\begin{equation}
F= \int (\sum_{a=0}^7D_a \phi_a)^3 [Y]=\int_B \pi_* f_{1*} f_{2*} f_{3*} f_{4*} f_{5*} f_{6*} f_{7*} f_{8*} \Big[\big({\sum_{a=0}^7D_a \phi_a}\big)^3 [Y]\Big],
\end{equation}
where $f_i$ is the $i$-th blowup and $\pi: X_0=\mathbb{P}[\mathscr{O}_B\oplus\mathscr{L}^{\otimes 2}\oplus \mathscr{L}^{\otimes 3}]\to B$ is the map defining the projective bundle. 
Using the pushforward theorems discussed above, the triple intersection polynomial can be expressed in terms of intersection numbers in the base $B$.

The intersection polynomials are computed in Appendix~\ref{intpoly} below.  We consider the tree of  blowups~(\ref{eq:finalbl}) and
each of the triple intersections is computed  by successively applying Theorem \ref{Thm:Lem34} or Theorem \ref{Thm:Push} to pushforward the intersection computation to X$_0$ and finally using Theorem \ref{Thm:PushH} to pushforward to the base.\footnote{ 
The only challenge arises when it is not obvious how to express the fibral divisor as a neat Cartier divisor. In some cases, when it is only defined as a complete intersection $g_1=g_2=0$, we can use an excess intersection formula. Equivalently, one can just perform another blowup with center $(g_1,g_2)$ (see footnote~\ref{addblfootnt}).}

 In the Calabi--Yau threefold case, we have 
\begin{equation}
L=-K, \quad S\cdot L= 2-2g +S^2.
\end{equation}
The matter representations are the adjoint and the $\mathbf{56}$.  The number of hypermultiplets charged under these are functions of the genus and self-intersection number of $S$: 
\begin{equation}
n_A=g, \quad n_F=  \frac{1}{2}  V(a)\cdot S  =4(1-g)+\frac{1}{2}S^2.
\end{equation}
The non-negativity of $n_A$ and $n_F$  implies a bound on the self-intersection of $S$: 
\begin{equation}
S^2\geq -8(1-g).
\end{equation}
\begin{rem}
There are no adjoint hypermultiplets when $g=0$ and no matter in the representation $\mathbf{56}$ when $n_F=0$, which means $S^2=-8(1-g)$. 
In particular, when $g=0$ and $S^2=-8$, we find that all four of the triple intersection polynomials (see Theorem~\ref{thm:TripleCY3}) reduce to
\end{rem}
\begin{equation}
\begin{aligned}
F(\phi)= &\phantom{+}   
8 (\phi _0^3 +  \phi _1^3+\phi _2^3+ \phi _3^3 +\phi _4^3+ \phi _5^3+\phi _6^3+ \phi _7^3) 
-6  (-2\phi _0^2 \phi _1+3\phi _0\phi _1^2 )
\\
&-6  (\phi _1^2 \phi _2+ \phi _3^2 \phi _2+\phi _3^2\phi _4+\phi _3^2\phi _7- \phi _4 \phi _5^2+2\phi_1 \phi_2^2+2 \phi _4^2 \phi _5-2\phi _5  \phi _6^2+3\phi _5^2 \phi _6 )\\
~\\
\end{aligned}
\end{equation}

\begin{rem}[Odd self-intersection and half-hypermultiplets]
When $S^2$ is odd, $n_F$ is a half-integer.  This is possible since the representation $\mathbf{56}$ is pseudo-real and thus allows half-hypermultiplets. 
In particular, if $g=0$ and  $S^2=-7$, we have one half-hypermultiplet at $V(a)\cap S$ since $n_F=\frac{1}{2}$. 
\end{rem}

\begin{thm}  \label{thm:TripleCY3}

For a Calabi--Yau threefold $Y$ defined as the crepant resolution of an E$_7$ Weierstrass model and corresponding to the chamber $i$ ($i=4,5,6,8$), the cubic intersection polynomials reduce to $F_i(\phi)$ and are as follows: 

\begin{equation}
\begin{aligned}
F_4(\phi)= &\phantom{+}   
8(1-g) (\phi _0^3 +  \phi _1^3+\phi _2^3+\phi _4^3+ \phi _5^3+ \phi _7^3)-S^2 \phi _3^3    -2 (4-4g+S^2)\phi _6^3 \\
&-3(4-4g+S^2) \phi _0^2 \phi _1+ 3 (2-2g+S^2)\phi _0\phi _1^2 +
3 (6 - 6 g + S^2) \phi _1^2 \phi _2+3(4 - 4 g + S^2)  \phi_1 \phi_2^2\\
& -3(8 - 8 g + S^2)( \phi _3 \phi _2^2 +\phi_3^2 \phi_6+\phi_3 \phi_6^2+2  \phi_4^2 \phi_6+2\phi_6\phi_7^2)\\
& +3(6 - 6 g + S^2)\phi _3^2 \phi _2-6(11 - 11 g + S^2) \phi _5^2 \phi _6 +6 (10 - 10 g + S^2) \phi _5  \phi _6^2\\
&+6 (-1 + g)  ( \phi _3^2\phi _4+\phi _3^2\phi _7+2 \phi _4^2 \phi _5- \phi _4 \phi _5^2)\\
& +6(8 - 8 g + S^2)( \phi_3 \phi_4 \phi_6+\phi_4 \phi_5 \phi_6 +\phi_3 \phi_6 \phi_7)
\end{aligned}
\end{equation}

\begin{align}
\begin{aligned}
F_5(\phi)=&8(1-g)( \phi _0^3 + \phi _1^3 +\phi_2^3+ \phi _3^3+ \phi _5^3)-S^2 (\phi _4^3+ \phi _6^3 +\phi _7^3)\\
&
-3(8 - 8 g + S^2)( \phi _3 \phi _2^2   -\phi _3 \phi _4^2+ \phi _4 \phi _6^2+ \phi _4 \phi _7^2+ \phi _6 \phi _7^2+ \phi _4^2 \phi _6+ \phi _6^2 \phi _7-\phi _3 \phi _7^2 +\phi _7 \phi _4^2)\\
&-3 (4 - 4 g + S^2) \phi _1 \phi _0^2 +3(2 - 2 g + S^2)\phi _1^2 \phi _0+3(6 - 6 g + S^2) \phi _2 \phi _3^2- \phi _1^2 \phi _2)
\\
&+ 3(10 - 10 g + S^2)( 2 \phi _5 \phi _6^2 - \phi _3^2 \phi _4 - \phi _3^2 \phi _7 )+3
(4 - 4 g + S^2) \phi _1 \phi _2^2 -6 (11 - 11 g + S^2)\phi _5^2 \phi _6
\\
&  +6 (-1 + g)(2 \phi _4^2 \phi _5 - \phi _4 \phi _5^2)
 + 6 (8 - 8 g + S^2) (\phi _4 \phi _5 \phi _6+\phi _3 \phi _4 \phi _7+\phi _4 \phi _6 \phi _7)
\end{aligned}
\end{align}

\begin{equation}
\begin{aligned}
F_6(\phi) = &\phantom{+}     
8(1-g)( \phi _0^3+ \phi _1^3 +\phi _2^3 + \phi _3^3 + \phi _4^3 + \phi _6^3)-S^2 \phi _5^3-2 (4 - 4 g + S^2)\phi _7^3\\
& +6(8 - 8 g + S^2) ( \phi _3 \phi _4 \phi _7+ \phi _4 \phi _5 \phi _7 + \phi _5 \phi _6 \phi _7)+ 
3 (2 - 2 g + S^2) \phi _0 \phi _1^2\\
& -3 (8 - 8 g + S^2) (\phi _3 \phi _2^2-\phi _3 \phi _4^2-\phi _3 \phi _7^2+\phi _5 \phi _7^2+2 \phi _4^2 \phi _7+\phi _5^2 \phi _7+2 \phi _6^2 \phi _7) \\
&- 3(4 - 4 g + S^2) (\phi _1 \phi _0^2- \phi _1 \phi _2^2)
+3(6 - 6 g + S^2)(\phi _2 \phi _3^2- \phi _1^2 \phi _2)-3(14 - 14 g + S^2)\phi _5^2 \phi _6\\
& +3(10 - 10 g + S^2)(\phi _4 \phi _5^2- \phi _3^2 \phi _7 - \phi _3^2 \phi _4)
+3(12 - 12 g + S^2) (\phi _5 \phi _6^2- \phi _4^2 \phi _5)
\end{aligned}
\end{equation}

\begin{align}
\begin{aligned}
F_8(\phi)=&\    8(1-g)(\phi _0^3+\phi _1^3+\phi _2^3+ \phi _3^3 +\phi _5^3 + \phi _6^3+ \phi _7^3) -2 (4 - 4 g + S^2)\phi _4^3  \\
&+3(2 - 2 g + S^2) \phi _0 \phi _1^2 -3 (4 - 4 g + S^2) \phi _0^2  \phi _1 \\
&+3 (4 - 4 g + S^2) \phi _1 \phi _2^2-3 (6 - 6 g + S^2)\phi _1^2 \phi _2  +3 (6 - 6 g + S^2) \phi _2 \phi _3^2-3 (8 - 8 g + S^2) \phi _2^2 \phi _3\\
&+6(10 - 10 g + S^2)\phi _5 \phi _6^2 -6(11 - 11 g + S^2)  \phi _5^2 \phi _6 -3 (10 - 10 g + S^2)  \phi _3^2 \phi _7\\
&+6(1-g)\phi _4 \phi _5^2 -3(10 - 10 g + S^2)\phi _3^2 \phi _4 +12(-1+g)\phi _4^2 \phi _5\\
&-6(8 - 8 g + S^2))\phi _4 \left( \phi _6^2 -\phi _5 \phi _6+\phi _7^2\right)+ 3(8 - 8 g + S^2) \phi _3 ( \phi _4+\phi _7)^2.
\end{aligned}
\end{align}

\end{thm}

 \section{Isomorphism classes of fibral divisors}\label{isomcl}

In this section, we determine the isomorphism classes of the fibral divisors that are projective bundles in the relative minimal models Y$_4$, Y$_5$, Y$_6$, and Y$_8$, using their known crepant resolutions. The results are listed in Table \ref{Table:Div}.

When the fiber $C_a$ does not degenerate, the fibral divisor D$_a$ $(a=0,1,\ldots, 7)$ is a ruled surface $D_a\to S$ isomorphic to a $\mathbb{P}^1$-bundle over the divisor $S$ supporting the E$_7$ fiber.  Since we only have two line bundles available, namely $\mathscr{L}$ and $\mathscr{S}$, we expect to have projective bundles of the form\footnote{As in Section~\ref{sec:defe7}, we use the symbol $\mathscr{S}$ to denote the line bundle for which the divisor $S$ is the zero locus of a smooth section. } 
$$\mathbb{P}_S(\mathscr{S}^{\otimes p} \oplus\mathscr{L}^{\otimes q}), $$
where $p$ and $q$ are integer numbers. 
There are two methods of finding $p$ and $q$.  The first method is akin to that used in \cite{F4}, where one keeps track of the rescaling freedom after each blowup in order to identify the class of the relative projective coordinates.  The second method uses intersection theory results from the previous section. We will take the intersection theoretic approach here and include an example scaling computation in Appendix~\ref{fibdivpf}.

\begin{table}[htb]
\begin{center}
{
\begin{tabular}{|c|c|c|c|c|}
\hline 
& Y$_4$& Y$_5$&  Y$_6$&  Y$_8$\\
\hline 
D$_0$ &  $\mathbb{P}_S(\mathscr{O}_S \oplus \mathscr{L})$ & $\mathbb{P}_S(\mathscr{O}_S \oplus \mathscr{L})$ & $\mathbb{P}_S(\mathscr{O}_S \oplus \mathscr{L})$ &  $\mathbb{P}_S(\mathscr{O}_S \oplus \mathscr{L})$ \\
\hline 
D$_1$ & $\mathbb{P}_S(\mathscr{S} \oplus \mathscr{L}^{\otimes 2})$& $\mathbb{P}_S(\mathscr{S} \oplus \mathscr{L}^{\otimes 2})$& $\mathbb{P}_S(\mathscr{S} \oplus \mathscr{L}^{\otimes 2})$&  $\mathbb{P}_S(\mathscr{S} \oplus \mathscr{L}^{\otimes 2})$ \\
\hline 
D$_2$ &$\mathbb{P}_S(\mathscr{S}^{\otimes 2} \oplus \mathscr{L}^{\otimes 3})$ & $\mathbb{P}_S(\mathscr{S}^{\otimes 2} \oplus \mathscr{L}^{\otimes 3})$ & $\mathbb{P}_S(\mathscr{S}^{\otimes 2} \oplus \mathscr{L}^{\otimes 3})$ &  $\mathbb{P}_S(\mathscr{S}^{\otimes 2} \oplus \mathscr{L}^{\otimes 3})$ \\
\hline 
D$_3$ & N/A &$\mathbb{P}_S(\mathscr{S}^{\otimes 3} \oplus \mathscr{L}^{\otimes 4})$ & $\mathbb{P}_S(\mathscr{S}^{\otimes 3} \oplus \mathscr{L}^{\otimes 4})$& $\mathbb{P}_S(\mathscr{S}^{\otimes 3} \oplus \mathscr{L}^{\otimes 4})$  \\
\hline 
D$_4$ & 
$\mathbb{P}_S(\mathscr{S} \oplus \mathscr{L})$
&
N/A
& $\mathbb{P}_S(\mathscr{S}^{\otimes 4} \oplus \mathscr{L}^{\otimes 5})$ &  N/A \\
\hline 
D$_5$ & $\mathbb{P}_S(\mathscr{S}^{\otimes 2} \oplus \mathscr{L}^{\otimes 2})$ & $\mathbb{P}_S(\mathscr{S}^{\otimes 2} \oplus \mathscr{L}^{\otimes 2})$ &  N/A&  $\mathbb{P}_S(\mathscr{S}^{\otimes 2} \oplus \mathscr{L}^{\otimes 2})$ \\
\hline 
D$_6$ &N/A & N/A &   $\mathbb{P}_S(\mathscr{S}^{\otimes 6} \oplus \mathscr{L}^{\otimes 7})$ &  $\mathbb{P}_S(\mathscr{S}^{\otimes 9} \oplus \mathscr{L}^{\otimes 11})$  \\
\hline 
D$_7$ & $\mathbb{P}_S(\mathscr{S}^{} \oplus \mathscr{L}^{})$& N/A  &N/A &  $\mathbb{P}_S(\mathscr{S}^{\otimes4} \oplus \mathscr{L}^{\otimes 5})$ \\
\hline 
\end{tabular}
}
\end{center}
\caption{Fibral divisors that are projective bundles. 
We write N/A when a divisor is not a projective bundle. 
If the base is of dimension three or higher, D$_6$ is not a projective bundle unless  $V(a,b)\cap S$ is empty. Otherwise it contains a full rational surface over the locus $V(a,b)\cap S$. We examine this possibility in Section~\ref{Sec:Flatness}. 
}
 \label{Table:Div}
\end{table}

We will now describe the steps to derive Table~\ref{Table:Div}.  Let $X$ be a  $\mathbb{P}^1$-bundle over a smooth variety $S$ of the type 
\begin{equation}\label{eq:ourX}
\pi: X=\mathbb{P}_S(\mathscr{O}_S\oplus \mathscr{D})\to S,
\end{equation}
 where $\mathscr{D}$ is a line bundle over $S$. Let $[u_0:u_1]$ be projective coordinates along the fiber of $X$ with $u_0$ a section of $\mathscr{O}_X(1)$   and $u_1$ a section of $\mathscr{O}_X(1)\otimes \pi^* \mathscr{D}$. Let $J$ denote the first Chern class of the line bundle $\mathscr{O}_X(1)$, and $D$ denote the first Chern class of $\mathscr{D}$. 
The divisors $V(u_0)$ and $V(u_1)$ define sections of $\pi$ corresponding to the classes  $J$ and  $J+\pi^* D$ in the Chow ring. 

The total Chern class of $X$ is 
\begin{equation}
c(TX)= (1+ J)(1+J+ \pi^* D)\pi^* c(TS). 
\end{equation}
In particular, we have 
\begin{equation}
c_1(TX)=2J+\pi^* D   + \pi^*c_1(TS).
\end{equation}
We can compute the pushforwards $\pi_*  J^k$ using the functorial properties of the Segre map.
The key formula is: 
\begin{equation}
\pi_* \frac{1}{1-J} \cap [X]=\frac{1}{1+D}\cap [S],
\end{equation}
or equivalently
\begin{equation}
\pi_*(  [X]+J \cap[X] + J^2 \cap [X]+\cdots) = [S]- D \cap [S] + D^2\cap [S] - D^3 \cap [S]+\cdots.
\end{equation}
By matching terms of the same dimensionality, we get: 
\begin{equation}
\pi_* 1=0, \quad \pi_* J = 1, \quad \pi_* J^2 = -D, \quad \pi_* J^{k+1}= (-1)^{k} D^k \quad (k>1).
\end{equation}
We now assume that $S$ is a smooth curve of genus $g$. Then,  $X$ is a geometrically ruled surface.   Before proceeding further, let us recall some facts about ruled surfaces.

\begin{defn} 
A  {\em smooth compact projective curve} is a curve isomorphic to the  projective line $\mathbb{P}^1$. 
A {\em ruled surface} is a morphism $\pi:X\to S$  such that  the generic fiber is a smooth compact rational curve. 
A smooth morphism $\pi: X\to S$ is called a {\em geometric ruled surface} if all its fibers are isomorphic to a smooth projective rational curve.   \end{defn}

Let $S$ be a curve of genus $g$. If we denote the class of  a fiber by $f$, then there is an irreducible curve of class $h_-$ and self-intersection $-n$ ($n\geq 0$) defining a section such that the  canonical class of $X$ satisfies 
\begin{equation}
-K_X= 2 h_-+(n+ 2-2g) f.
\end{equation}
There is also an irreducible curve of class $h_+=h_-+nf$ with self-intersection $n$. 
The curves of class $h_\pm$ both define sections of $X$ and they don't intersect
\begin{equation}
f^2 =1, \quad h_\pm^2=\pm n, \quad h_+ \cdot h_-=0, \quad h_\pm \cdot f =1. 
\end{equation} 
The integer $n$ is called the invariant of the ruled surface. 

We now apply this to the fibral divisor $X$ defined in equation~(\ref{eq:ourX}).
Since a projective bundle is  a flat fibration, all fibers have the same class $f$. 
At the level of the Chow group, the generators of A($X$) are $J$ and $f$. 
The degree of $D$ in $S$ is $n$ and we have: 
\begin{equation}
\int_X J \cdot \pi^* D = \int_S \pi_*  J\cdot D =\int_S D =\int D \cdot S= n, \quad  
[\pi^*D]=n [f],
\end{equation}
so that
\begin{equation}
\int_X f\cdot J=\int_{\mathbb{P}^1}c_1( \mathscr{O}_{\mathbb{P}^1}(1)) =1, \quad 
\int_X f^2 =\int_{\mathbb{P}^1} c_1( \mathscr{O}_{\mathbb{P}^1})=0. 
\end{equation}
\begin{equation}
\int_X J^2=-n, \quad
 \int_X J \cdot (J+\pi^* D)=0,  \quad \int_X (J+\pi^* D)^2=n. 
\end{equation}

We are now ready to derive Table~\ref{Table:Div}. For a fibral divisor which is a $\mathbb{P}^1$-projective bundle, we can determine its type by identifying two non-trivial classes r$_1$ and r$_2$ which correspond to two irreducible curves, as well as a line bundle $\mathscr{D}$   over $S$ with first Chern class $D$  such that 
$r_1^2+r_2^2=r_1 r_2 =0$ and  $\int_X r_1^2 = D S$.   In that situation, we deduce that $X\cong \mathbb{P}_S (\mathscr{O}_S \oplus \mathscr{D})$:
\begin{equation}
\left\{
{
\begin{matrix}
r_1^2+r_2^2=0\\
r_1 r_2 =0 \\ r_1^2 = D S
\end{matrix}}
\right\}
\Longrightarrow X\cong \mathbb{P}_S (\mathscr{O}_S \oplus \mathscr{D}).
\end{equation}
The results are listed in Table \ref{Table:Ind}, which requires the pushforward formulas of Section~\ref{tripint} and divisor class computations in Appendix~\ref{intpoly}. 
We deduce 
Table \ref{Table:Div} directly from Table  \ref{Table:Ind}. \\

\par The following theorem explains how different projective bundles are related to each other. 

\begin{thm}
Let $\rho:Y\to B$ be an elliptic fibration defined by a  crepant resolution of a Weierstrass model over a base $B$.  
Let $D_a$ and $D_b$ be two divisors  of $Y$ such that $\rho_* (D_a D_b)=S$. 
Then 
\begin{equation}
\rho_* (D^2_a D_{b}+D_a D^2_{b}) =(S-L)S.
\end{equation}
\end{thm}
\begin{proof}
If we denote by $\eta_{ab}$ the intersection of two adjacent divisors $D_a$ and $D_{b}$, then 
\begin{equation}
K_{\eta_{ab}}= K_Y+D_a+D_b,
\end{equation}
\begin{equation}
\chi(\eta_{ab})=-K_{\eta_{ab} }D_a D_b= -(K_Y+D_a +D_{b})D_ a D_{b}.
\end{equation}
But since $\eta_{ab}$ is isomorphic to $S$, we also have 
\begin{equation}
\chi(\eta_{ab})=-(K_B+S)S. 
\end{equation}
Since $K_Y=K_B+L$ and $\eta_{ab}=D_a D_b$ pushes forward to $S$  for any ADE model, we get 
\begin{align}
\rho_* (D^2_a D_{b}+D_a D^2_{b})& =  -(K_B+L) S -\chi(S)
 =-(K_B+L)S+(K_B+S) S=(S-L)S.
\end{align}
\end{proof}

\begin{table}[htb]
\begin{center}
{
\begin{tabular}{|c|c|c|c|c|c|c|}
\hline 
 & $R_1$ & $R_2$ & Y$_4$& Y$_5$&  Y$_6$&  Y$_8$\\
\hline 
D$_0$ & D$_1$ & $\frac{1}{3}H$  &  $LS$ &$LS$ &$LS$ & $LS$ \\
\hline 
D$_1$ & D$_0$ &  D$_2$ & $-(2L-S)S$ & $-(2L-S)S$ & $-(2L-S)S$ & $-(2L-S)S$ \\
\hline 
D$_2$  & D$_1$ & D$_3$ & $-(3L-2S)S$ & $-(3L-2S)S$ & $-(3L-2S)S$ & $-(3L-2S)S$ \\
\hline 
D$_3$ & D$_2$ & D$_4$ & N/A & $-(4L-3S)S$ & $-(4L-3S)S$ & $-(4L-3S)S$ \\
\hline 
D$_4$ &  D$_3$ & D$_5$&  $-(L-S)S$ & N/A & $-(5L-4S)S$&  N/A \\
\hline 
D$_5$ &  D$_4$ & $x$ & $-(2L-2S)S$ & $-(2L-2S)S$ &   N/A & $-(2L-2S) S$ \\
\hline 
D$_6$ & D$_5$ & $y$  & N/A & N/A &  $-(7L-6S)S$  &   $-(11L-9S) S$  \\
\hline 
D$_7$ & D$_3$ & $y$ & $-(L-S)S$& N/A & N/A &$-(5L-4S)S$ \\
\hline 
\end{tabular}
}
\end{center}
\caption{
For each fibral divisor D$_a$, we present a divisor $R_1$ and a divisor $R_2$ such that the divisor $R_1$ defines, by intersection, a section of  $D_a\to S$. 
The divisor $R_2$ is such that $D_a R_1^2=-D_a R_2^2$ and $R_1 R_2 D_a=0$.  An N/A indicates when these conditions do not hold.
When D$_a$ is a $\mathbb{P}^1$-bundle and has  $D_a R_1^2=\pm (p L -q S) S$, we deduce that D$_a$ is isomorphic to  
$\mathbb{P}_S(\mathscr{S}^{\otimes  q}\oplus\mathscr{L}^{\otimes p})$ as listed in Table \ref{Table:Div}. 
}
 \label{Table:Ind}
\end{table}

\section{Characteristic numbers of fibral divisors} 
In this section, we give the linear functions induced on H$^2(Y, \mathbb{Z})$ for the second Chern class of the minimal models Y=\{Y$_4$, Y$_5$, Y$_6$, Y$_8$\}, as well as  characteristic numbers of their fibral divisiors.  In particular, we consider the signature $\tau(D$)  and also the ordinary $\chi(D)$ and holomorphic $\chi_0(D)$ Euler characteristics. This data provides information about the structure of the fibral divisors.
 We assume that the minimal models are threefolds, thus, these characteristic numbers are all functions of the Chern numbers $c_1^2(D_a)$ and $c_2(D_a)$.  Characteristic numbers for elliptic fibrations are computed in \cite{Esole:2018bmf,EK.Invariant}.

Given a threefold $Y$, the second Chern class defines a linear form on  $H^2(Y,\mathbb{Z})$
$$\mu:\quad H^2(Y,\mathbb{Z})\to \mathbb{Z}\quad \quad D\mapsto\int_Y D\cdot c_2 (TY).$$
Knowing the properties of this linear form is important for several reasons. 
For one, Wilson showed that the  linear form $\mu$ plays a central role in the classification of Calabi--Yau varieties \cite{Wilson1}.
The second Chern class also appears  in the Hirzebruch--Riemann--Roch theorem and is used in the computation of  the microscopic entropy attached to a very ample divisor $D$ in a Calabi--Yau threefold \cite{Maldacena:1997de}. 
~\\

 \begin{thm}
For each of the minimal models Y$_4$, Y$_5$, Y$_6$, and Y$_8$ of an E$_7$-model, the second Chern class induces the following linear action on the divisors \label{thm:Chern2} 
$$
\begin{aligned}
& c_2(TY)\cdot \varphi^*H = 3 \Big(c_2(TB) - c_1(TB) L\Big),\\
& c_2(TY)\cdot \varphi^*\pi^* \alpha = 12 L \cdot \alpha,\\
& c_2(TY)\cdot D_a =2(L-S) S,
\end{aligned}
$$
where 
$H=c_1\Big(\mathscr{O}_{X_0}(1)\Big)$, 
$\varphi$ is the crepant resolution, $\pi$ is the projection of the projective bundle X$_0$ over the base $B$, and $\alpha$ is a class of the Chow ring of the base.  We note the following exceptions for each minimal model: 
$$
\begin{aligned}
\text{Y}_4  :\quad  & c_2(TY)\cdot D_3=2 (3 L - 2 S) S, \quad c_2(TY)\cdot D_6= 2 (7 L - 5 S) S, \\
\text{Y}_5  :\quad &c_2(TY)\cdot D_4=2 (3 L - 2 S) S, \quad c_2(TY)\cdot D_5=2 (3 L - 2 S) S, \quad c_2(TY)\cdot D_7=2 (3 L - 2 S) S, \\
\text{Y}_6  :\quad  & c_2(TY)\cdot D_5=2 (3 L - 2 S) S, \quad c_2(TY)\cdot D_7= 2 (7 L - 5 S) S ,\\
\text{Y}_8  :\quad  &  c_2 (TY)\cdot D_4= 2 (7 L - 5 S) S.
\end{aligned}
$$
\end{thm}
\begin{proof}
The Chern class of the variety $Y$ is computed using Theorem \ref{Thm:AluffiCBU} and the rest follows from the pushforward results in Theorem 
\ref{Thm:Push}  (or Theorem \ref{Thm:Lem34}) and Theorem \ref{Thm:PushH}. 
\end{proof}

The characteristic numbers that we are interested in are 
\begin{equation}
\tau(D)=\frac{1}{3}\int_D (-2c_2 +c_1^2)\, \quad \chi(D)=\int_D c_2, \quad \chi_0(D)=\int_D\frac{c_1^2+c_2}{12},
\end{equation}
where $\tau(D$)  is the signature, $\chi(D)$ is its Euler number, and $\chi_0(D)$ is its holomorphic Euler characteristic. 

\begin{lem}
Let  $D$ be  a ruled surface over a smooth curve of genus $g$. Then 
\begin{equation}
\chi(D)=4(1-g), \quad \chi_0(D)=(1-g), \quad \tau(D)=0.
\end{equation}
\end{lem}
\noindent While the holomorphic Euler characteristic of a fibral divisor is the same for any crepant resolution,
 both the signature and the ordinary Euler characteristic depend on the choice of minimal model. 
We can use the Euler characteristic and the signature to identify fibral divisors that have  reducible singular fibers. For instance, the signature is zero when the fibral divisor is a ruled variety.

\begin{thm}  The characteristic numbers of the fibral divisors of the crepant resolution of an E$_7$-model are as follows for the relative minimal model Y$_4$, Y$_5$, Y$_6$, and Y$_8$: 
\label{thm:char.fibral}
\begin{align}
& Y_4
\begin{cases}
\begin{array}{lll}
\chi_0(D_a)=(1-g), &  & a =0, 1, 2, 3, 4, 5, 6, 7\\
 \tau(D_a)=0,  & \
    \chi(D_a)=  4(1-g), &  a=0,1,2,4,5,7 \\
 \tau(D_3)=  (4 L - 3 S) S,  &  \ 
 \chi(D_3)= 4(1-g) +  (4 L - 3 S) S,  \\  
 \tau(D_6)=  2(4 L - 3 S) S,    & \   \   
 \chi(D_6)= 4(1-g) + 2 (4 L -3 S) S,
 \end{array}
\end{cases}
\\
& Y_5
\begin{cases}
\begin{array}{lll}
\chi_0(D_a)=(1-g), & & \quad a =0, 1, 2, 3, 4, 5, 6, 7 \\
 \tau(D_a)=0,  & 
 \chi(D_a)=4(1-g), &\quad  a=0,1,2,3,5\\
   \tau(D_a)=  (4 L - 3 S) S, &  
 \chi(D_a)= 4(1-g) + (4 L - 3 S) S,  & \quad a=4,6,7
 \end{array}
\end{cases}
\\
&
Y_6
\begin{cases}
\begin{array}{lll}
\chi_0(D_a)=(1-g), & &  a =0, 1, 2, 3, 4, 5, 6, 7 \\
 \tau(D_a)=0,  & 
    \chi(D_a)=  4(1-g), &  a=0,1,2, 3, 4, 6 \\
 \tau(D_5)=  (4 L - 3 S) S,  &   
 \chi(D_5)= 4(1-g) + (4 L - 3 S) S,  \\  
 \tau(D_7)=  2(4 L - 3 S) S,    & 
 \chi(D_7)= 4(1-g) +2 (4 L - 3 S) S, 
 \end{array}
\end{cases}
\\
&Y_8
\begin{cases}
\begin{array}{lll}
\chi_0(D_a)=(1-g), &  & a =0, 1, 2, 3, 4, 5, 6, 7 \\
 \tau(D_a)=0,  & 
    \chi(D_a)=  4(1-g), &  a=0,1,2, 3, 5, 6,  7 \\
 \tau(D_4)= 2 (4 L - 3 S) S,  &  
 \chi(D_4)= 4(1-g) +  2 (4 L - 3 S) S.  &
 \end{array}
\end{cases}
\end{align}
\end{thm}

\begin{proof}

To ease the notation, we will not write all of the pushforwards and pullbacks.
The total Chern class of a fibral divisor $D$ is computed by adjoint from the total Chern class of the variety $Y$ which, in turn, can be deduced from Theorem \ref{Thm:AluffiCBU}. Namely, 

\begin{equation}
c(TD)=\frac{c(TY)}{1+D}=1+(c_1(TY)-D)+ (c_2(TY)-c_1(TY)D +D^2) + \cdots
\end{equation}
By adjoint for $Y$, we have 
\begin{equation}
c_1(TY)=c_1(TB)-c_1(\mathscr{L}),
\end{equation}
so we deduce that 
\begin{equation}
c_1(TD)=c_1(TB)-L-D, \quad 
c_2(TD)= c_2(TY)-(c_1(TB)-L)D +D^2.
\end{equation}
The second Chern class of $Y$ will appear multiplied by the class of a fibral divisor and we can therefore use Theorem 
\ref{thm:Chern2} to express all the results as functions of the base, once we pushforward to the base using Theorem 
\ref{Thm:Push}  (or Theorem \ref{Thm:Lem34}) and Theorem \ref{Thm:PushH}. 
\end{proof}

\section{Fat fibers and loss of flatness}\label{Sec:Flatness}

For a base of high enough dimension, there can be points over which the fiber is not a collection of rational curves, but rather contains an entire rational surface as a component. 
This phenomena has been studied in the case of an E$_6$-model by analyzing a partial resolution of its Weierstrass model in~\cite{Candelas}. 
In M-theory compactifications, the presence of a complex surface $Q$ in the fiber results in new light degrees of freedom in the low energy spectrum, as M5-branes wrapping the surface can produce massless stringy modes and a tower of particle states arise by wrapping membranes on holomorphic curves in $Q$.

When the base of an E$_7$-model has dimension three or higher, the 
fibration is no longer flat. 
When the locus $s=a=b=0$ is non-empty,  the divisor D$_6$ has  a fiber over the divisor $S$ that jumps in dimension to become a rational surface over this locus.   
In the minimal model Y$_a$, we call this rational surface Q$_a$. For the minimal models that are constructed explicitly by a  crepant resolution, namely  Y$_a$, ($a=4,5, 6, 8$), we can   identify Q$_a$ up to isomorphism, explicitly. 

\begin{thm}\label{thm:Q}
Let Y$_a$ ($a=1,\ldots, 8$) be a crepant resolution of an E$_7$-model with fibral divisors D$_n$ with generic fibers C$_n$ ($n=0,\ldots, 7$). 
Let Q$_a$ be the surface that the fiber C$_6$ degenerates into over the locus $V(a,b)\cap S$. 
Then for Y$_4$, Y$_5$, Y$_6$, and Y$_8$, which can be defined by the crepant resolutions given in equation \eqref{eq:finalbl}, we have: 
\begin{equation}
Q_4\cong\mathbb{F}_2^{(2)}, \quad 
Q_5\cong\mathbb{F}_1^{(1)}, \quad
Q_6\cong\mathbb{F}_1,\quad
Q_8\cong\mathbb{F}_2.
\end{equation}
\end{thm}
\begin{tabular}{lp{14.5cm}}
Y$_8$ :&  The surface $Q_8$ is isomorphic to the Hirzebruch surface $\mathbb{F}_2$. \\
Y$_6$ :&  The surface $Q_6$ is isomorphic to an  $\mathbb{F}_1$ (which is also a del Pezzo surface of degree $8$). \\
Y$_5$  :&  The  surface $Q_5$ is a Hirzebruch surface $\mathbb{F}_1$ blown-up at a point of its curve of self-intersection $-1$. \\
Y$_4$   :&  The surface $Q_4$ is the blowup of a  Hirzebruch surface $\mathbb{F}_2$ over a point $p$ of the curve of self-intersection $2$ followed by a blowup of the intersection point between the proper transform of that curve and the proper transform of the fiber over the point $p$. The structures of these rational surfaces are summarized in Figure \ref{Fig:Q}.
\end{tabular}
~\\

\begin{figure}[htb]
\begin{center}
\scalebox{1.2}{
\begin{tabular}{c}
\begin{tikzpicture}
				\draw[thick] (-.2,0) to (2.2,0);
				\draw[thick] (0,-.2) to (0,2.2);
				\draw[thick] (-.2,2) to (2.2,2);
				\draw[thick] (2,-.2) to (2,2.2);
				\node  at (1,1){Q$_8\cong\mathbb{F}_2$};
				\node  at (-.3,1){$0$};
				\node  at (2.3,1){$0$};
				\node  at (1,-.3){$-2$};
				\node  at (1,2.3){$2$};
			\end{tikzpicture}

		\\

					\begin{tikzpicture}
				\draw[thick] (-.2,0) to (2.2,0);
				\draw[thick] (0,-.2) to (0,2.2);
				\draw[thick] (-.2,2) to (2.2,2);
				\draw[thick] (1.4,-.2) to +(30:1.6);
				\draw[thick] (1.4,2.2) to +(-30:1.6);
			         \draw[thick] (2.5,2) to +(0,-2);
			         				\node  at (1,1){Q$_4$};
				\node  at (-.3,1){$0$};
				\node  at (1.8,.3){$-2$};
			        \node  at (2.2,1){$-1$};
				\node  at (1.8,1.7){$-2$};
				\node  at (1,-.3){$-2$};
				\node  at (1,2.3){$1$};
			\end{tikzpicture}
			\quad\quad

			\begin{tikzpicture}
				\draw[thick] (-.2,0) to (2.2,0);
				\draw[thick] (0,-.2) to (0,2.2);
				\draw[thick] (-.2,2) to (2.2,2);
				\draw[thick] (1.8,-.2) to +(45:2);
				\draw[thick] (1.8,2.2) to +(-45:2);
								\node  at (1,1){Q$_5$};
				\node  at (-.3,1){$0$};
				\node  at (2,.4){$-1$};
				\node  at (2,1.6){$-1$};
				\node  at (1,-.3){$-2$};
				\node  at (1,2.3){$1$};
				\node at (3.01,1) {$\bullet$};
			\end{tikzpicture}
			
		\quad \quad

						\begin{tikzpicture}
				\draw[thick] (-.2,0) to (2.2,0);
				\draw[thick] (0,-.2) to (0,2.2);
				\draw[thick] (-.2,2) to (2.2,2);
				\draw[thick] (2,-.2) to (2,2.2);
								\node  at (1,1){Q$_6\cong\mathbb{F}_1$};
				\node  at (-.3,1){$0$};
				\node  at (2.3,1){$0$};
				\node  at (1,-.3){$-1$};
				\node  at (1,2.3){$1$};
				\node at (2.01,0) {$\bullet$};
			\end{tikzpicture}

			\end{tabular}}
			\end{center}

\caption{ Isomorphism classes of the rational surfaces Q$_a$ ($a=4,5,6,8$). 
For the crepant resolution Y$_a$, the fibral divisor D$_6$ has a generic fiber C$_6$ that degenerates into a full rational surface $Q_a$ over the codimension-three locus $V(a,b)\cap S$. The four rational surfaces Q$_a$ are connected to each other by blowing-up points or blowing-down $(-1)$-curves. 
The rational surfaces Q$_8$ and Q$_6$ are, respectively, isomorphic to the Hirzebruch surfaces $\mathbb{F}_2$ and $\mathbb{F}_1$. 
The rational surface Q$_5$ is obtained by blowing-up a point of the ($-1$)-curve of $\mathbb{F}_1$ or by blowing-up a point of the curve of self-intersection $2$ in $\mathbb{F}_2$. 
Alternatively, Q$_8$ is obtained from Q$_5$ by contracting the ($-1$)-curve that is the proper transform of the fiber over the point $p$ that was blown-up to go from Q$_6$ to Q$_5$. 
The rational surface Q$_4$ is obtained by blowing-up the intersection of the two ($-1$)-curves of Q$_5$.  
\label{Fig:Q}
}			\end{figure}
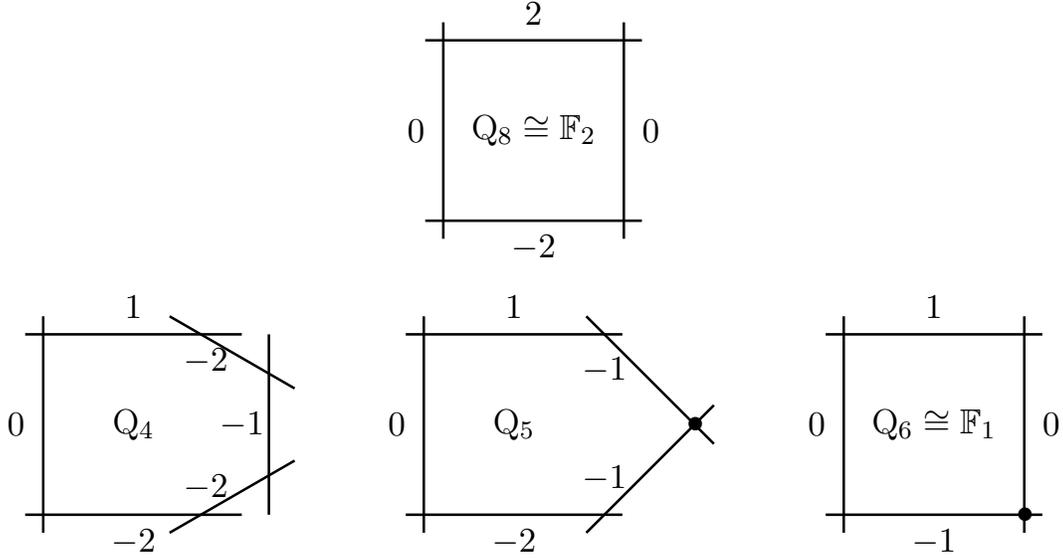

The relevant computations can be found in Appendix~\ref{highercodim}.
Using adjunction, we compute the total Chern class of the surface $Q$  and its characteristic numbers. We can then use this data to identify the surfaces.
\begin{lem}
The Euler characteristic, the degree, the holomorphic Euler characteristic, and signature of the surfaces $Q_a$ ($a=4,5,6,8$) are:
$$
\begin{array}{l l l l l }
Y_4 :\quad & \chi(Q_4)=6 \hspace{1cm} &  K_{Q_4}^2 = 6 ,\qquad  & \chi_0(Q_4)=1,\quad & \tau(Q_4)=-2,\\
Y_5 :\quad & \chi(Q_5)=5  &  K_{Q_5}^2 = 7 ,\qquad  &\chi_0(Q_5)=1,\quad & \tau(Q_5)=-1,\\
Y_6 :\quad & \chi(Q_6)=4  &  K_{Q_6}^2 = 8 ,\qquad  &\chi_0(Q_6)=1,\quad & \tau(Q_6)=\ 0,\\
Y_8 :\quad & \chi(Q_8)=4  &  K_{Q_8}^2 = 8 ,\qquad & \chi_0(Q_8)=1, \quad & \tau(Q_8)=\  0.
\end{array}
$$
\end{lem}
~\\

While flops do not change the fibral divisor (namely D$_6$) whose generic fiber  degenerates into the surface $Q$ over $V(a,b,s)$, flops do change the topology of $Q$ by blowing-up/down  certain points. Such blowups will change $\chi(Q)$, $K_Q^2$, and $\tau(Q)$. Since we know the degeneration of the fiber C$_6$, we make the following conjectures.

~\\~\\

\begin{conjec} The Euler characteristic, the degree, the holomorphic Euler characteristic, and signature of the surfaces $Q_a$ ($a=1,2,3,7$) are expected to be:
$$
\begin{array}{l l l l l }
Y_1 :\quad & \chi(Q_1)=9 \hspace{1cm} &  K_{Q_1}^2 = 3 ,\qquad  & \chi_0(Q_1)=1,\quad & \tau(Q_1)=-5,\\
Y_2 :\quad & \chi(Q_2)=8 \hspace{1cm} &  K_{Q_2}^2 = 4 ,\qquad  & \chi_0(Q_2)=1,\quad & \tau(Q_2)=-4,\\
Y_3 :\quad & \chi(Q_3)=7 \hspace{1cm} &  K_{Q_3}^2 = 5 ,\qquad  & \chi_0(Q_3)=1,\quad & \tau(Q_3)=-3,\\
Y_7 :\quad & \chi(Q_7)=4  &  K_{Q_7}^2 = 8 ,\qquad & \chi_0(Q_7)=1, \quad & \tau(Q_7)=\  0,
\end{array}
$$
\end{conjec}
We expect $Q_7$ to be a Hirzebruch surface $\mathbb{F}_1$, and 
$Q_3$, 
$Q_2$, 
and $Q_1$ to be
the blowup of a Hirzebruch surface $\mathbb{F}_2$  at three, four, and five points, respectively. 
These points are on the fiber of self-intersection $-2$ and then on successive intersections of the proper transform of this fiber with the exceptional divisors.

\pagebreak
\appendix

\section{Fiber degenerations}\label{sec:fibdegen}
Here we derive in detail  the splitting of curves in each of the eight chambers using the hyperplane arrangement I(E$_7$, $\mathbf{56}$).

\subsection{Ch$_1$}\label{Sec:Ch1}
\vspace{-2em}
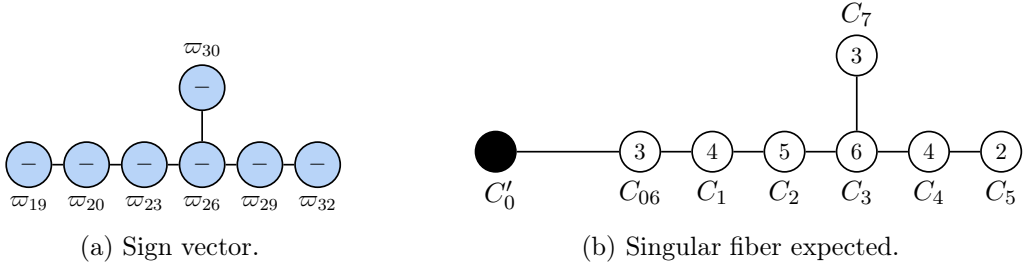
\begin{figure}[H]
\centering 
\begin{subfigure}[t]{0.45\textwidth}
\centering
\scalebox{.8}{
\begin{tikzpicture}[scale=.8]
				\node[fill=myblue,draw,circle,thick,scale=1,label=below:{\scalebox{1}{$\varpi_{19}$}}] (1) at (1.2,0){$-$};
				\node[fill=myblue,draw,circle,thick,scale=1,label=below:{\scalebox{1}{$\varpi_{20}$}}] (2) at (2.4,0){$-$};
				\node[fill=myblue,draw,circle,thick,scale=1,label=below:{\scalebox{1}{$\varpi_{23}$}}] (3) at (3.6,0){$-$};
				\node[fill=myblue,draw,circle,thick,scale=1,label=below:{\scalebox{1}{$\varpi_{26}$}}] (4) at (4.8,0){$-$};
				\node[fill=myblue,draw,circle,thick,scale=1,label=below:{\scalebox{1}{$\varpi_{29}$}}] (5) at (6,0){$-$};
				\node[fill=myblue,draw,circle,thick,scale=1,label=below:{\scalebox{1}{$\varpi_{32}$}}] (6) at (7.2,0){$-$};
				\node[fill=myblue,draw,circle,thick,scale=1, label=above:{\scalebox{1}{$\varpi_{30}$}}] (8) at (4.8,1.6){$-$};
				\draw[thick] (1)--(2)--(3)--(4)--(5)--(6);
				\draw[thick]  (4)--(8);
	\end{tikzpicture}}
\caption{Sign vector.}
\label{Fig:Ch1Sign}
\end{subfigure}
\begin{subfigure}[t]{0.45\textwidth}
\centering
\scalebox{.8}{
\begin{tikzpicture}
				\node[draw,circle,thick,scale=1,fill=black,label=below:{\scalebox{1.2}{ $C'_0$}}] (0) at (0,0){$1$};
				\node[draw,circle,thick,scale=1,label=below:{\scalebox{1.2}{$C_{06}$}}] (2) at (2.4,0){$3$};
				\node[draw,circle,thick,scale=1,label=below:{\scalebox{1.2}{$C_1$}}] (3) at (3.6,0){$4$};
				\node[draw,circle,thick,scale=1,label=below:{\scalebox{1.2}{$C_2$}}] (4) at (4.8,0){$5$};
				\node[draw,circle,thick,scale=1,label=below:{\scalebox{1.2}{$C_3$}}] (5) at (6,0){$6$};
				\node[draw,circle,thick,scale=1,label=below:{\scalebox{1.2}{$C_4$}}] (6) at (7.2,0){$4$};
				\node[draw,circle,thick,scale=1, label=below:{\scalebox{1.2}{$C_5$}}] (7) at (8.4,0){$2$};
				\node[draw,circle,thick,scale=1, label=above:{\scalebox{1.2}{$C_7$}}] (8) at (6,1.6){$3$};
				\draw[thick] (0)--(2)--(3)--(4)--(5)--(6)--(7);
				\draw[thick]  (5)--(8);
					\end{tikzpicture}}
					\caption{Singular fiber expected. }
					\label{Fig:Ch1Fib}
					\end{subfigure}
					
\caption{Chamber 1.}
\end{figure}

As we can see from Figure \ref{Figure:IG}, the unique wall of  chamber 1 that intersects the interior of the fundamental open Weyl chamber is the hyperplane $\varpi_{19}^\bot$ . This hyperplane  separates chamber 1 and chamber 2:  the form $\langle \varpi_{19},\phi\rangle$ is negative in the interior of the chamber  1, vanishes on the wall $\varpi_{19}^\bot$ and is positive in chamber 2. 
The condition 
\begin{equation}
\langle \varpi_{19}, \phi\rangle=\phi_1-\phi_6 <0,
\end{equation}
completely characterizes chamber 1 since $\varpi_{19}$ is higher than all the other weights appearing in the sign vector. 
It follows that   $-\varpi_{19}$ (resp. $\varpi_{19}$)  is an effective curve in chamber 1 (resp. in chamber 2) which we call $C'_6$ .  

 Geometrically, any simple root is an effective curve in a given chamber unless it connects two weights of different signs. 
{
In the basis of simple roots, we have   
$\varpi_{19}=(1, 1, 1, \frac{1}{2}, 0, \text{-}\frac{1}{2}, \frac{1}{2})$ as listed in  equation \eqref{Eq:translation}, or equivalently }
\begin{equation}
\varpi_{19}=\alpha_1 + \alpha_2 +\alpha_3+\frac{1}{2}\alpha_4 -\frac{1}{2}\alpha_6 +\frac{1}{2}\alpha_7.
\end{equation}
We rewrite this equation in the following suggestive form 
\begin{equation}
\alpha_6=2(-\varpi_{19})+2\alpha_1+2\alpha_2 + 2 \alpha_3 +\alpha_4 + \alpha_7
\end{equation}
and deduce that the curve $C_6$ splits as follows: 
\begin{equation}
C_6\to 2 C'_{6}+2C_1+2C_2+2C_3 +C_4+C_7.
\end{equation}
{This matches the description in \cite{Diaconescu:1998cn}.}
The intersection number of $C'_6$ with the fibral divisor D$_i$ is given by minus the coefficient of $\varpi_{19}$ written in the basis of fundamental weights.  Since in  the basis of fundamental weights, we have $-(-\varpi_{19})= 
\boxed{\ 1\ \ 0\ 0 \ 0 \ 0  \ $-1$ \ \ 0 }$, we find
\begin{equation}
D_1\cdot C'_6=1, \quad D_6 \cdot C'_6=-1, \quad D_i \cdot C'_6=0, \quad i=2,3,4,5,7.
\end{equation}
{
We also note that by linearity (see equation \eqref{Eq.D0}), we have} 
\begin{equation}
D_0\cdot C'_6=-1,
\end{equation}
which is negative and therefore implies that  $D_0$ contains $C'_6$.  Since both $D_0$ and $D_6$ contain $C'_6$, we rename   $C'_6$ as $C_{06}$ and we have the splitting rule 
\begin{equation}
\begin{cases}
C_0 \to C_{06}+C'_0\\
C_6\to 2 C_{06}+2C_1+2C_2+2C_3 +C_4+C_7
\end{cases}
\end{equation}
where $C'_0$ is the left-over curve in $C_0$. Since $C_0$ has weight $\varpi_0$ and $C_{06}$ has weight $-\varpi_{19}$, and  $\varpi_0 = -\varpi_{19}-\varpi_1$, we see that $-\varpi_1$ is the weight of $C'_0$:
\begin{equation}
\begin{cases}
C_{06}\to -\varpi_{19}, \\
 C'_0\to \boxed{\ 0\ \ 0\ 0 \ 0 \ 0  \ $-1$ \ \ 0 }=-\varpi_1.
\end{cases}
\end{equation}
From the Dynkin indices of $\varpi_1$, we deduce the following intersection numbers
\begin{equation}
D_0\cdot C'_0= -1, \quad D_r \cdot C'_0=0 \quad r=1,2,3,4,5,7, \quad D_6\cdot C'_0=1,
\end{equation}
and the degeneration 
\begin{equation}
C_0+2 C_1 + 3 C_2 + 4C_3 + 3 C_4 + 2C_5+C_6 + 2  C_7\to 
C'_0+ 3C_{06}+4C_1+5C_2+6C_3+4C_4+ 2C_5+3C_7
\end{equation}
We get a fiber whose dual graph is the affine Dynkin diagram $\tilde{\text{E}}_8$ with the node corresponding to $\alpha_1$ contracted to a point and the identification: 
\begin{equation}
(C'_0, C_{06}, C_1,C_2,C_3, C_4, C_5, C_7)\to (\alpha_0, \alpha_2, \alpha_3, \alpha_4, \alpha_5, \alpha_6, \alpha_7, \alpha_8), 
\end{equation}
with the respective multiplicities $(1,3,4,5,6,4,2,3)$.

\subsection{Ch$_2$}
\vspace{-2em}

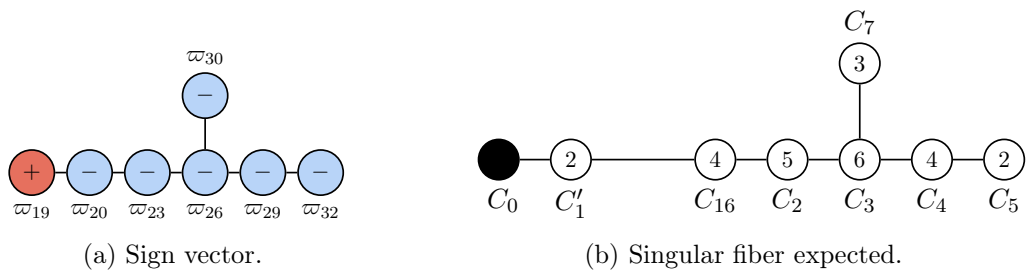
\begin{figure}[H]
\centering 
\begin{subfigure}[t]{0.45\textwidth}
\centering
\scalebox{.8}{
\begin{tikzpicture}[scale=.8]
		\node[fill=myred,draw,circle,thick,scale=1,label=below:{\scalebox{1}{$\varpi_{19}$}}] (1) at (1.2,0){$+$};
				\node[fill=myblue,draw,circle,thick,scale=1,label=below:{\scalebox{1}{$\varpi_{20}$}}] (2) at (2.4,0){$-$};
				\node[fill=myblue,draw,circle,thick,scale=1,label=below:{\scalebox{1}{$\varpi_{23}$}}] (3) at (3.6,0){$-$};
				\node[fill=myblue,draw,circle,thick,scale=1,label=below:{\scalebox{1}{$\varpi_{26}$}}] (4) at (4.8,0){$-$};
				\node[fill=myblue,draw,circle,thick,scale=1,label=below:{\scalebox{1}{$\varpi_{29}$}}] (5) at (6,0){$-$};
				\node[fill=myblue,draw,circle,thick,scale=1,label=below:{\scalebox{1}{$\varpi_{32}$}}] (6) at (7.2,0){$-$};
				\node[fill=myblue,draw,circle,thick,scale=1, label=above:{\scalebox{1}{$\varpi_{30}$}}] (8) at (4.8,1.6){$-$};
				\draw[thick] (1)--(2)--(3)--(4)--(5)--(6);
				\draw[thick]  (4)--(8);
				\end{tikzpicture}}
\caption{Sign vector.}
\label{Fig:Ch2Sign}
\end{subfigure}
\begin{subfigure}[t]{0.45\textwidth}
\centering
\scalebox{.8}{
\begin{tikzpicture}
				\node[draw,circle,thick,scale=1,fill=black,label=below:{\scalebox{1.2}{ $C_0$}}] (0) at (0,0){$1$};
				\node[draw, circle,thick,scale=1, label=below:{\scalebox{1.2}{$C'_{1}$}}] (1) at (1.2,0){$2$};
			\node[draw,circle,thick,scale=1,label=below:{\scalebox{1.2}{$C_{16}$}}] (3) at (3.6,0){$4$};
				\node[draw,circle,thick,scale=1,label=below:{\scalebox{1.2}{$C_2$}}] (4) at (4.8,0){$5$};
				\node[draw,circle,thick,scale=1,label=below:{\scalebox{1.2}{$C_3$}}] (5) at (6,0){$6$};
				\node[draw,circle,thick,scale=1,label=below:{\scalebox{1.2}{$C_4$}}] (6) at (7.2,0){$4$};
				\node[draw,circle,thick,scale=1, label=below:{\scalebox{1.2}{$C_5$}}] (7) at (8.4,0){$2$};
				\node[draw,circle,thick,scale=1, label=above:{\scalebox{1.2}{$C_7$}}] (8) at (6,1.6){$3$};
				\draw[thick] (0)--(1)--(3)--(4)--(5)--(6)--(7);
				\draw[thick]  (5)--(8);					\end{tikzpicture}}
					\caption{Singular fiber expected. }
					\label{Fig:Ch2Fib}
					\end{subfigure}
					
\caption{Chamber 2.}
\end{figure}
{In chamber 2, we see from  Figure \ref{Figure:IG} that the extremal faces are  $\varpi_{19}^\bot$ and $\varpi_{20}^\bot$ with}
\begin{equation}
\varpi_{19}-\alpha_1=\varpi_{20}, \quad \varpi_{19}\cdot \phi>0 , \quad \varpi_{20}\cdot \phi<0.
\end{equation}
We conclude that $\varpi_{19}$ and $-\varpi_{20}$ will correspond to effective extremal curves in this chamber. 
We can also use the expression of $\varpi_{20}$ in terms of simple roots: 
\begin{equation}
\varpi_{19}=\alpha_1 + \alpha_2 +\alpha_3+\frac{1}{2}\alpha_4 -\frac{1}{2}\alpha_6 +\frac{1}{2}
\alpha_7,\quad 
\varpi_{20}=\  \alpha_2+\alpha_3+\frac{1}{2}\alpha_4-\frac{1}{2}\alpha_6+\frac{1}{2}\alpha_7
\end{equation}
{and solve for  $\alpha_1$ and $\alpha_6$ as sums of weights that define a positive form in the interior of the chamber:}
\begin{equation}
\alpha_1= \varpi_{19}+(-\varpi_{20}),\quad \alpha_6=2(-\varpi_{20})+ 2\alpha_2+2\alpha_3 +\alpha_4 +\alpha_7.
\end{equation}
There is an effective curve $C_{16}$ corresponding to $-\varpi_{20}$ and an effective cure $C'_1$ corresponding to $\varpi_{19}$.
Our choice of notation is because $-\varpi_{20}$ shows up for both $\alpha_1$ and $\alpha_6$ while $\varpi_{19}$ only appears in $\alpha_1$, which we see as follows.

 Since in  the basis of fundamental weights, we have $-(-\varpi_{20})= 
\boxed{\ $-1$\ \ 1\ 0 \ 0 \ 0  \ $-1$ \ \ 0 }$, we deduce  
\begin{equation}
D_1\cdot C_{16}=-1, \quad D_2\cdot C_{16}=1, \quad  D_6 \cdot C_{16}=-1, \quad D_i \cdot C_{16}=0, \quad i=3,4,5,7.
\end{equation}
We also note that by linearity
\begin{equation}
D_0\cdot C_{16}=0.
\end{equation}
The negative intersection numbers imply both $D_1$ and $D_6$ contain $C_{16}$. Meanwhile from 
$-(\varpi_{19})= 
\boxed{\ $-1$\ \ 0\ 0 \ 0 \ 0  \ $1$ \ \ 0 }$,  we deduce the intersections of $C'_1$ which are the negative of those computed for $C_{06}$ above, namely  
\begin{equation}
D_0\cdot C'_6=1,\quad  D_1\cdot C'_6=-1, \quad D_6 \cdot C'_6=1, \quad D_i \cdot C'_6=0, \quad i=2,3,4,5,7.
\end{equation}
We then have 
\begin{equation}
\begin{cases}
C_1\to C_{16}+C'_1\\
C_6\to 2C_{16}+2C_2+2C_3 +C_4+C_7. 
\end{cases}
\end{equation}
Using these linear equations, we find: 
\begin{equation}
C_0+2 C_1 + 3 C_2 + 4C_3 + 3 C_5 + C_6 + 2  C_7\to 2 C'_1 + 4 C_{16} + 5 C_2 + 6 C_3 + 4 C_4 + 2 C_5 + 3 C_7.
\end{equation}
We get a fiber whose dual graph is the affine Dynkin diagram $\tilde{\text{E}}_8$ with the node corresponding to $\alpha_2$ contracted to a point and the identification: 
\begin{equation}
(C_0, C'_1, C_{16}, C_1,C_2,C_3, C_4, C_5, C_7)\to (\alpha_0, \alpha_1, \alpha_3, \alpha_4, \alpha_5, \alpha_6, \alpha_7, \alpha_8), 
\end{equation}
with the respective multiplicities $(1,2,4,5,6,4,2,3)$.

\subsection{Ch$_3$}

\vspace{-2em}
\begin{figure}[H]
\centering 
\begin{subfigure}[t]{0.45\textwidth}
\centering
\scalebox{.8}{
\begin{tikzpicture}[scale=.8]
		\node[fill=myred,draw,circle,thick,scale=1,label=below:{\scalebox{1.2}{$\varpi_{19}$}}] (1) at (1.2,0){$+$};
				\node[fill=myred,draw,circle,thick,scale=1,label=below:{\scalebox{1.2}{$\varpi_{20}$}}] (2) at (2.4,0){$+$};
				\node[fill=myblue,draw,circle,thick,scale=1,label=below:{\scalebox{1.2}{$\varpi_{23}$}}] (3) at (3.6,0){$-$};
				\node[fill=myblue,draw,circle,thick,scale=1,label=below:{\scalebox{1.2}{$\varpi_{26}$}}] (4) at (4.8,0){$-$};
				\node[fill=myblue,draw,circle,thick,scale=1,label=below:{\scalebox{1.2}{$\varpi_{29}$}}] (5) at (6,0){$-$};
				\node[fill=myblue,draw,circle,thick,scale=1,label=below:{\scalebox{1.2}{$\varpi_{32}$}}] (6) at (7.2,0){$-$};
				\node[fill=myblue,draw,circle,thick,scale=1, label=above:{\scalebox{1.2}{$\varpi_{30}$}}] (8) at (4.8,1.6){$-$};
				\draw[thick] (1)--(2)--(3)--(4)--(5)--(6);
				\draw[thick]  (4)--(8);					\end{tikzpicture}}
\caption{Sign vector.}
\label{Fig:Ch3Sign}
\end{subfigure}
\begin{subfigure}[t]{0.45\textwidth}
\centering
\scalebox{.8}{
\begin{tikzpicture}
				\node[draw,circle,thick,scale=1,fill=black,label=below:{\scalebox{1.2}{ $C_0$}}] (0) at (0,0){$1$};
				\node[draw,circle,thick,scale=1,label=below:{\scalebox{1.2}{$C_{1}$}}] (2) at (1.2,0){$2$};
				\node[draw,circle,thick,scale=1,label=below:{\scalebox{1.2}{$C'_{2}$}}] (3) at (2.4,0){$3$};
				\node[draw,circle,thick,scale=1,label=below:{\scalebox{1.2}{$C_{26}$}}] (4) at (4.8,0){$5$};
				\node[draw,circle,thick,scale=1,label=below:{\scalebox{1.2}{$C_3$}}] (5) at (6,0){$6$};
				\node[draw,circle,thick,scale=1,label=below:{\scalebox{1.2}{$C_4$}}] (6) at (7.2,0){$4$};
				\node[draw,circle,thick,scale=1, label=below:{\scalebox{1.2}{$C_5$}}] (7) at (8.4,0){$2$};
				\node[draw,circle,thick,scale=1, label=above:{\scalebox{1.2}{$C_7$}}] (8) at (6,1.6){$3$};
				\draw[thick] (0)--(2)--(3)--(4)--(5)--(6)--(7);
				\draw[thick]  (5)--(8);
					\end{tikzpicture}}
					\caption{Singular fiber expected. }
					\label{Fig:Ch3Fib}
					\end{subfigure}
					
\caption{Chamber 3.}
\end{figure}
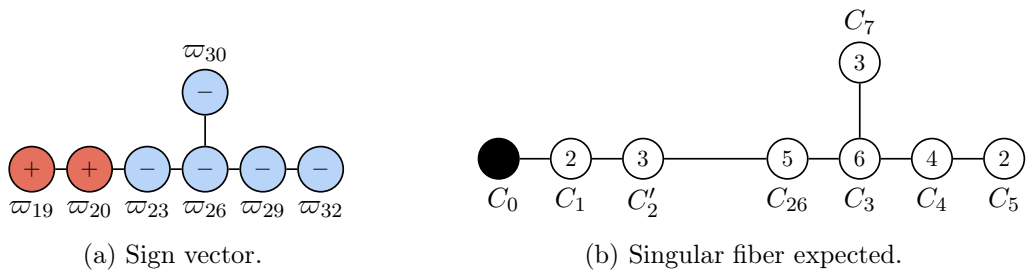

The interior walls are now $\varpi_{20}$ and $\varpi_{23}$ with
\begin{equation}
\varpi_{20}-\alpha_2=\varpi_{23}, \quad \varpi_{20}\cdot \phi>0 , \quad \varpi_{23}\cdot \phi<0.
\end{equation}
We conclude that $\varpi_{20}$ and $-\varpi_{23}$ will correspond to effective extremal curves in this chamber. 
We can also use the expressions for these weights in terms of simple roots: 

\begin{equation}
\begin{cases}
\varpi_{20}=\alpha_2+\alpha_3+\frac{1}{2}\alpha_4-\frac{1}{2}\alpha_6+\frac{1}{2}\alpha_7,\\
\varpi_{23}=\alpha_3+\frac{1}{2}\alpha_4-\frac{1}{2}\alpha_6+\frac{1}{2}\alpha_7.
\end{cases}
\end{equation}
Solving for $\alpha_2$ and $\alpha_6$, we have
\begin{equation}
\alpha_2= \varpi_{20}+(-\varpi_{23}),\quad \alpha_6=2(-\varpi_{20})+2\alpha_3 +\alpha_4 +\alpha_7.
\end{equation}
There is an effective curve $C_{26}$ corresponding to $-\varpi_{23}$ and an effective cure $C'_2$ corresponding to $\varpi_{20}$. 
 We thus have 
\begin{equation}
\begin{cases}
C_2\to C_{26}+C'_2\\
C_6\to 2C_{26}+2C_3 +C_4+C_7. 
\end{cases}
\end{equation}
Our choice of notation is because $C_{26}$ shows up for both $\alpha_2$ and $\alpha_6$ while $C'_2$ only appears in $\alpha_2$, which we see as follows.

Since in  the basis of fundamental weights, we have $-(-\varpi_{23})= 
\boxed{\ $0$\ \ $-1$\ 1 \ 0 \ 0  \ $-1$ \ \ 0 }$, we deduce  
\begin{equation}
D_2\cdot C_{26}=-1, \quad D_3\cdot C_{26}=1, \quad  D_6 \cdot C_{26}=-1, \quad D_i \cdot C_{26}=0, \quad i=1,4,5,7.
\end{equation}
We also note that by linearity $D_0 \cong - (2D_1+3 D_2 + 4 D_3 + 3 D_4+2D_5+D_6+2D_7)$, hence, 
\begin{equation}
D_0\cdot C_{26}=0.
\end{equation}
The negative intersection numbers imply both $D_2$ and $D_6$ contain $C_{26}$.  Meanwhile from $
-(\varpi_{20})= 
\boxed{\ $1$\ \ -1\ 0 \ 0 \ 0  \ $1$ \ \ 0 }$, we deduce the intersections of $C'_2$ which are the negative of those computed for $C_{16}$ above, namely
\begin{equation}
D_1\cdot C'_2=1, \quad D_2\cdot C'_2=-1, \quad  D_6 \cdot C'_2=1, \quad D_i \cdot C'_6=0, \quad i=0,3,4,5,7.
\end{equation}
Finally we find the degeneration
\begin{equation}
C_0+2 C_1 + 3 C_2 + 4C_3 +3 C_4+ 2 C_5 + C_6 + 2  C_7\to 
C_0+2 C_1 + 3C'_2 +5 C_{26} + 6C_3 +4 C_4+ 2 C_5  +3C_7.
\end{equation}
We get a fiber whose dual graph is the affine Dynkin diagram $\tilde{\text{E}}_8$ with the node corresponding to $\alpha_3$ contracted to a point and the identification: 
\begin{equation}
(C_0, C_{1}, C'_2, C_{26},C_3, C_4, C_5, C_7)\to (\alpha_0, \alpha_1, \alpha_2, \alpha_4, \alpha_5, \alpha_6, \alpha_7, \alpha_8), 
\end{equation}
with the respective multiplicities $(1,2,3,5,6,4,2,3)$.

\subsection{Ch$_4$}
\vspace{-2em}

\begin{figure}[H]
\centering 
\begin{subfigure}[t]{0.45\textwidth}
\centering
\scalebox{.8}{
\begin{tikzpicture}[scale=.8]
		\node[fill=myred,draw,circle,thick,scale=1,label=below:{\scalebox{1.2}{$\varpi_{19}$}}] (1) at (1.2,0){$+$};
				\node[fill=myred,draw,circle,thick,scale=1,label=below:{\scalebox{1.2}{$\varpi_{20}$}}] (2) at (2.4,0){$+$};
				\node[fill=myred,draw,circle,thick,scale=1,label=below:{\scalebox{1.2}{$\varpi_{23}$}}] (3) at (3.6,0){$+$};
				\node[fill=myblue,draw,circle,thick,scale=1,label=below:{\scalebox{1.2}{$\varpi_{26}$}}] (4) at (4.8,0){$-$};
				\node[fill=myblue,draw,circle,thick,scale=1,label=below:{\scalebox{1.2}{$\varpi_{29}$}}] (5) at (6,0){$-$};
				\node[fill=myblue,draw,circle,thick,scale=1,label=below:{\scalebox{1.2}{$\varpi_{32}$}}] (6) at (7.2,0){$-$};
				\node[fill=myblue,draw,circle,thick,scale=1, label=above:{\scalebox{1.2}{$\varpi_{30}$}}] (8) at (4.8,1.6){$-$};
				\draw[thick] (1)--(2)--(3)--(4)--(5)--(6);
				\draw[thick]  (4)--(8);
					\end{tikzpicture}}
\caption{Sign vector.}
\label{Fig:Ch4Sign}
\end{subfigure}
\begin{subfigure}[t]{0.45\textwidth}
\centering
\scalebox{.8}{
\begin{tikzpicture}
				\node[draw,circle,thick,scale=1,fill=black,label=below:{\scalebox{1.2}{ $C_0$}}] (0) at (0,0){$1$};
				\node[draw,circle,thick,scale=1,label=below:{\scalebox{1.2}{$C_{1}$}}] (2) at (1.2,0){$2$};
				\node[draw,circle,thick,scale=1,label=below:{\scalebox{1.2}{$C_{2}$}}] (3) at (2.4,0){$3$};
				\node[draw,circle,thick,scale=1,label=below:{\scalebox{1.2}{$C'_{3}$}}] (4) at (3.6,0){$4$};
				\node[draw,circle,thick,scale=1,label=below:{\scalebox{1.2}{$C_{36}$}}] (5) at (6,0){$6$};
				\node[draw,circle,thick,scale=1,label=below:{\scalebox{1.2}{$C_4$}}] (6) at (7.2,0){$4$};
				\node[draw,circle,thick,scale=1, label=below:{\scalebox{1.2}{$C_5$}}] (7) at (8.4,0){$2$};
				\node[draw,circle,thick,scale=1, label=above:{\scalebox{1.2}{$C_7$}}] (8) at (6,1.6){$3$};
				\draw[thick] (0)--(2)--(3)--(4)--(5)--(6)--(7);
				\draw[thick]  (5)--(8);
					\end{tikzpicture}}
					\caption{Singular fiber observed. }
					\label{Fig:Ch4Fib}
					\end{subfigure}
					
\caption{Chamber 4.}
\end{figure}
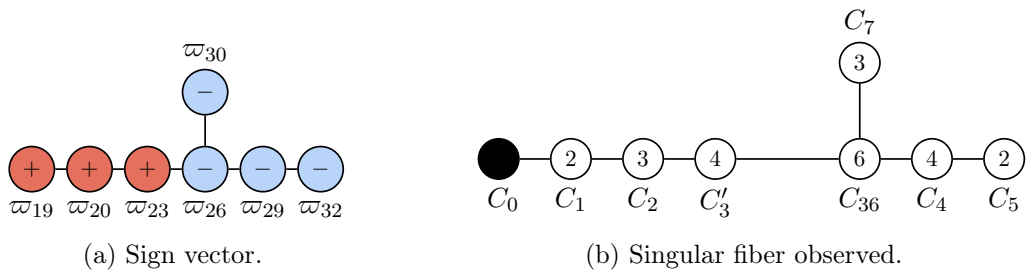

The interior walls are now $\varpi_{23}$ and $\varpi_{26}$ with
\begin{equation}
\varpi_{23}-\alpha_3=\varpi_{26}, \quad \varpi_{23}\cdot \phi>0 , \quad \varpi_{26}\cdot \phi<0.
\end{equation}
We conclude that $\varpi_{23}$ and $-\varpi_{26}$ will correspond to effective extremal curves in this chamber. 
We can also use the expressions for these weights in terms of simple roots: 
\begin{equation}
\begin{cases}
\varpi_{23}=\alpha_3+\frac{1}{2}\alpha_4-\frac{1}{2}\alpha_6+\frac{1}{2}\alpha_7,\\
\varpi_{26}=\frac{1}{2}\alpha_4-\frac{1}{2}\alpha_6+\frac{1}{2}\alpha_7.\\
\end{cases}
\end{equation}
Solving for $\alpha_3$ and $\alpha_6$, we have
\begin{equation}
\alpha_3= \varpi_{23}+(-\varpi_{26}),\quad \alpha_6=2(-\varpi_{20})+\alpha_4 +\alpha_7.
\end{equation}
There is an effective curve $C_{36}$ corresponding to $-\varpi_{26}$ and an effective cure $C'_3$ corresponding to $\varpi_{23}$. 
 We thus have 
\begin{equation}
\begin{cases}
C_3\to C_{36}+C'_3\\
C_6\to 2C_{36}+C_4+C_7. 
\end{cases}
\end{equation}
Our choice of notation is because $C_{36}$ shows up for both $\alpha_3$ and $\alpha_6$ while $C'_3$ only appears in $\alpha_3$, which we see as follows.
Since in  the basis of fundamental weights, we have $-(-\varpi_{26})= 
\boxed{\ $0$\ \ $0$\ $-1$ \ 1 \ 0  \ $-1$ \ \ 1 }$, we deduce  
\begin{equation}
D_3\cdot C_{36}=-1, \quad D_4\cdot C_{36}=1, \quad  D_6 \cdot C_{36}=-1, \quad  D_7 \cdot C_{36}=1, \quad D_i \cdot C_{26}=0, \quad i=1,2,5.
\end{equation}
We also note that by linearity $D_0 \cong - (2D_1+3 D_2 + 4 D_3 + 3 D_4+2D_5+D_6+2D_7)$, hence, 
\begin{equation}
D_0\cdot C_{36}=0.
\end{equation}
The negative intersection numbers imply both $D_3$ and $D_6$ contain $C_{36}$.  Meanwhile from $
-(\varpi_{23})= 
\boxed{\ $0$\ \ 1\ $-1$ \ 0 \ 0  \ $1$ \ \ 0 }$, we deduce the intersections of $C'_3$ which are the negative of those computed for $C_{26}$ above.  Namely,
\begin{equation}
D_2\cdot C'_{3}=-1, \quad D_3\cdot C'_{3}=1, \quad  D_6 \cdot C'_{3}=-1, \quad D_i \cdot C'_{3}=0, \quad i=0,1,4,5,7.
\end{equation}
Finally we find the degeneration
\begin{equation}
C_0+2 C_1 + 3 C_2 + 4C_3 +3 C_4+ 2 C_5 + C_6 + 2  C_7\to 
C_0+2 C_1 + 3 C_2 +4C'_3+ 6C_{36}+4 C_4+ 2 C_5 +3  C_7.
\end{equation}
We get a fiber whose dual graph is the affine Dynkin diagram $\tilde{\text{E}}_8$ with the node corresponding to $\alpha_4$ contracted to a point and the identification: 
\begin{equation}
(C_0, C_{1}, C_2, C'_3,C_{36}, C_4, C_5, C_7)\to (\alpha_0, \alpha_1, \alpha_2, \alpha_4, \alpha_5, \alpha_6, \alpha_7, \alpha_8), 
\end{equation}
with the respective multiplicities $(1,2,3,4,6,4,2,3)$.

\subsection{Ch$_5$}

\vspace{-2em}

\begin{figure}[H]
\centering 
\begin{subfigure}[t]{0.45\textwidth}
\centering
\scalebox{.8}{
\begin{tikzpicture}[scale=.8]
		\node[fill=myred,draw,circle,thick,scale=1,label=below:{\scalebox{1.2}{$\varpi_{19}$}}] (1) at (1.2,0){$+$};
				\node[fill=myred,draw,circle,thick,scale=1,label=below:{\scalebox{1.2}{$\varpi_{20}$}}] (2) at (2.4,0){$+$};
				\node[fill=myred,draw,circle,thick,scale=1,label=below:{\scalebox{1.2}{$\varpi_{23}$}}] (3) at (3.6,0){$+$};
				\node[fill=myred,draw,circle,thick,scale=1,label=below:{\scalebox{1.2}{$\varpi_{26}$}}] (4) at (4.8,0){$+$};
				\node[fill=myblue,draw,circle,thick,scale=1,label=below:{\scalebox{1.2}{$\varpi_{29}$}}] (5) at (6,0){$-$};
				\node[fill=myblue,draw,circle,thick,scale=1,label=below:{\scalebox{1.2}{$\varpi_{32}$}}] (6) at (7.2,0){$-$};
				\node[fill=myblue,draw,circle,thick,scale=1, label=above:{\scalebox{1.2}{$\varpi_{30}$}}] (8) at (4.8,1.6){$-$};
				\draw[thick] (1)--(2)--(3)--(4)--(5)--(6);
				\draw[thick]  (4)--(8);
					\end{tikzpicture}}
\caption{Sign vector.}
\label{Fig:Ch5Sign}
\end{subfigure}
\begin{subfigure}[t]{0.45\textwidth}
\centering
\scalebox{.8}{
\begin{tikzpicture}
				\node[draw,circle,thick,scale=1,fill=black,label=below:{\scalebox{1.2}{ $C_0$}}] (0) at (0,0){$1$};
				\node[draw,circle,thick,scale=1,label=below:{\scalebox{1.2}{$C_{1}$}}] (2) at (1.2,0){$2$};
				\node[draw,circle,thick,scale=1,label=below:{\scalebox{1.2}{$C_{2}$}}] (3) at (2.4,0){$3$};
				\node[draw,circle,thick,scale=1,label=below:{\scalebox{1.2}{$C_{3}$}}] (4) at (3.6,0){$4$};
				\node[draw,circle,thick,scale=1,label=below:{\scalebox{1.2}{$C_{47}$}}] (5) at (4.8,0){$5$};
				\node[draw,circle,thick,scale=1,label=below:{\scalebox{1.2}{$C_{46}$}}] (6) at (7.2,0){$4$};
				\node[draw,circle,thick,scale=1, label=below:{\scalebox{1.2}{$C_5$}}] (7) at (8.4,0){$2$};
				\node[draw,circle,thick,scale=1, label=above:{\scalebox{1.2}{$C_{67}$}}] (8) at (6,1.6){$3$};
				\draw[thick] (0)--(2)--(3)--(4)--(5)--(6)--(7);
				\draw[thick]  (6,0)--(8);
					\end{tikzpicture}}
					\caption{Singular fiber observed. }
					\label{Fig:Ch5Fib}
					\end{subfigure}
					
\caption{Chamber 5.}
\end{figure}
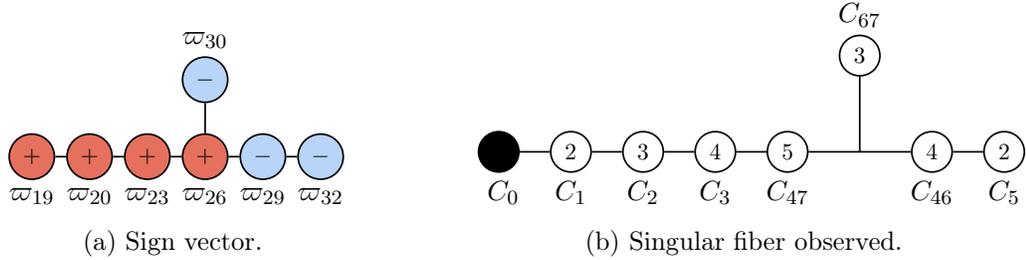

The interior walls are now $\varpi_{26}$, $\varpi_{29}$, and  $\varpi_{30}$  with
\begin{equation}
\varpi_{26}-\alpha_4=\varpi_{29},\quad \varpi_{26}-\alpha_7=\varpi_{30}, \quad \varpi_{26}\cdot \phi>0 , \quad \varpi_{29}\cdot \phi<0, \quad \varpi_{30}\cdot \phi<0.
\end{equation}
We conclude that $\varpi_{26}$, $-\varpi_{29}$, and $-\varpi_{30}$ will correspond to effective extremal curves in this chamber. 
We can also use the expressions for these weights in terms of simple roots: 
\begin{equation}
\begin{cases}
\varpi_{26}=\frac{1}{2}\alpha_4-\frac{1}{2}\alpha_6+\frac{1}{2}\alpha_7,\\
\varpi_{29}=-\frac{1}{2}\alpha_4-\frac{1}{2}\alpha_6+\frac{1}{2}\alpha_7,\\
\varpi_{30}=\frac{1}{2}\alpha_4-\frac{1}{2}\alpha_6-\frac{1}{2}\alpha_7.\\
\end{cases}
\end{equation}
Solving for $\alpha_4$, $\alpha_6$, and $\alpha_7$ we have
\begin{equation}
\alpha_4= \varpi_{26}+(-\varpi_{29}),\quad \alpha_7= \varpi_{26}+(-\varpi_{30}),\quad \alpha_6=(-\varpi_{29})+(-\varpi_{30}).
\end{equation}
There is an effective curve $C_{47}$ corresponding to $\varpi_{26}$, an effective curve $C_{46}$ corresponding to $-\varpi_{29}$, and an effective cure $C_{67}$ corresponding to $-\varpi_{30}$. 
 We thus have 
\begin{equation}
\begin{cases}
C_4\to C_{46}+C_{47}, \\
C_6 \to C_{46}+C_{67},\\
C_{7}\to C_{47}+C_{67}.
\end{cases}
\end{equation}
Our choice of notation is because $C_{46}$ shows up for both $\alpha_4$ and $\alpha_6$, $C_{47}$ shows up for both $\alpha_4$ and $\alpha_7$, and $C_{67}$ shows up for both $\alpha_6$ and $\alpha_7$, which we see as follows.
Since in  the basis of fundamental weights, we have $-(\varpi_{26})= 
\boxed{\ $0$\ \ $0$\ $1$ \ $-1$ \ 0  \ $1$ \ \ $-1$ }$, we deduce  
$$D_3\cdot C_{47}=1, \quad D_4\cdot C_{47}=-1, \quad  D_6 \cdot C_{47}=1, \quad  D_7 \cdot C_{47}=-1, \quad D_i \cdot C_{26}=0, \quad i=0,1,2,5.$$
The negative intersection numbers imply both $D_4$ and $D_7$ contain $C_{47}$.  Meanwhile, from 
$-(-\varpi_{29})= 
\boxed{\ $0$\ \ $0$\ $0$ \ $-1$ \ $1$  \ $-1$ \ \ $1$ }$, we deduce  
$$D_4\cdot C_{46}=-1, \quad D_5\cdot C_{46}=1, \quad  D_6 \cdot C_{46}=-1, \quad  D_7 \cdot C_{46}=1, \quad D_i \cdot C_{57}=0, \quad i=0,1,2,3.$$
The negative intersection numbers imply both $D_4$ and $D_6$ contain $C_{46}$.  Finally, from 
 $
-(-\varpi_{30})= 
\boxed{\ $0$\ \ 0\ $0$ \ $1$ \ $0$  \ $-1$ \ \ $-1$ }$, we deduce 
$$D_4\cdot C_{67}=1,  \quad  D_6 \cdot C_{67}=-1,\quad  D_7 \cdot C_{67}=-1, \quad D_i \cdot C'_{5}=0, \quad i=0,1,2,3,5$$
and the negative intersection numbers imply both $D_6$ and $D_7$ contain $C_{67}$.

In the end, we find the degeneration
$$
C_0+2 C_1 + 3 C_2 + 4C_3 +3 C_4+ 2 C_5 + C_6 + 2  C_7\to 
C_0+2 C_1 + 3 C_2 + 4C_3 +5C_{47}+4 C_{46}+ 2 C_5 +3C_{67}.
$$
We get a fiber whose dual graph is the affine Dynkin diagram $\tilde{\text{E}}_8$ with the node corresponding to $\alpha_5$ contracted to a point and the identification: 
$$
(C_0, C_{1}, C_2, C_3,C_{47}, C_{46}, C_5, C_{67})\to (\alpha_0, \alpha_1, \alpha_2, \alpha_3, \alpha_4, \alpha_6, \alpha_7, \alpha_8), 
$$
with the respective multiplicities $(1,2,3,4,6,4,2,3)$.

\subsection{Ch$_6$}

\vspace{-2em}

\begin{figure}[H]
\centering 
\begin{subfigure}[t]{0.45\textwidth}
\centering
\scalebox{.8}{
\begin{tikzpicture}[scale=.8]
		\node[fill=myred,draw,circle,thick,scale=1,label=below:{\scalebox{1.2}{$\varpi_{19}$}}] (1) at (1.2,0){$+$};
				\node[fill=myred,draw,circle,thick,scale=1,label=below:{\scalebox{1.2}{$\varpi_{20}$}}] (2) at (2.4,0){$+$};
				\node[fill=myred,draw,circle,thick,scale=1,label=below:{\scalebox{1.2}{$\varpi_{23}$}}] (3) at (3.6,0){$+$};
				\node[fill=myred,draw,circle,thick,scale=1,label=below:{\scalebox{1.2}{$\varpi_{26}$}}] (4) at (4.8,0){$+$};
				\node[fill=myred,draw,circle,thick,scale=1,label=below:{\scalebox{1.2}{$\varpi_{29}$}}] (5) at (6,0){$+$};
				\node[fill=myblue,draw,circle,thick,scale=1,label=below:{\scalebox{1.2}{$\varpi_{32}$}}] (6) at (7.2,0){$-$};
				\node[fill=myblue,draw,circle,thick,scale=1, label=above:{\scalebox{1.2}{$\varpi_{30}$}}] (8) at (4.8,1.6){$-$};
				\draw[thick] (1)--(2)--(3)--(4)--(5)--(6);
				\draw[thick]  (4)--(8);
					\end{tikzpicture}}
\caption{Sign vector.}
\label{Fig:Ch6Sign}
\end{subfigure}
\begin{subfigure}[t]{0.45\textwidth}
\centering
\scalebox{.8}{
\begin{tikzpicture}
				\node[draw,circle,thick,scale=1,fill=black,label=below:{\scalebox{1.2}{ $C_0$}}] (0) at (0,0){$1$};
				\node[draw,circle,thick,scale=1,label=below:{\scalebox{1.2}{$C_{1}$}}] (2) at (1.2,0){$2$};
				\node[draw,circle,thick,scale=1,label=below:{\scalebox{1.2}{$C_{2}$}}] (3) at (2.4,0){$3$};
				\node[draw,circle,thick,scale=1,label=below:{\scalebox{1.2}{$C_{3}$}}] (4) at (3.6,0){$4$};
				\node[draw,circle,thick,scale=1,label=below:{\scalebox{1.2}{$C_{4}$}}] (5) at (4.8,0){$5$};
				\node[draw,circle,thick,scale=1,label=below:{\scalebox{1.2}{$C_{57}$}}] (6) at (6,0){$6$};
				\node[draw,circle,thick,scale=1, label=below:{\scalebox{1.2}{$C'_5$}}] (7) at (8.4,0){$2$};
				\node[draw,circle,thick,scale=1, label=above:{\scalebox{1.2}{$C_6$}}] (8) at (6,1.6){$3$};
				\draw[thick] (0)--(2)--(3)--(4)--(5)--(6)--(7);
				\draw[thick]  (6)--(8);
					\end{tikzpicture}}
					\caption{Singular fiber observed.}
					\label{Fig:Ch6Fib}
					\end{subfigure}
					
\caption{Chamber 6.}
\end{figure}
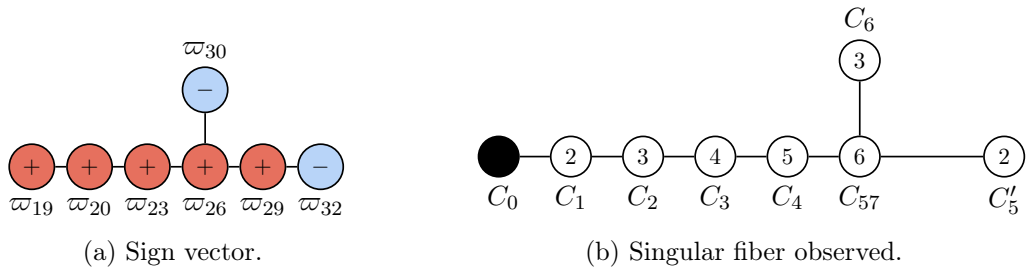

The interior walls are now $\varpi_{29}$ and $\varpi_{32}$ with
\begin{equation}
\varpi_{29}-\alpha_5=\varpi_{32}, \quad \varpi_{29}\cdot \phi>0 , \quad \varpi_{32}\cdot \phi<0.
\end{equation}
We conclude that $\varpi_{29}$ and $-\varpi_{32}$ will correspond to effective extremal curves in this chamber. 
We can also use the expressions for these in terms of simple roots: 
\begin{equation}
\begin{cases}
\varpi_{29}=-\frac{1}{2}\alpha_4-\frac{1}{2}\alpha_6+\frac{1}{2}\alpha_7,\\
\varpi_{32}=-\frac{1}{2}\alpha_4-\alpha_5-\frac{1}{2}\alpha_6+\frac{1}{2}\alpha_7.\\
\end{cases}
\end{equation}
Solving for $\alpha_5$ and $\alpha_7$, we have
\begin{equation}
\alpha_5= \varpi_{29}+(-\varpi_{32}),\quad \alpha_7=2\varpi_{29}+\alpha_4 +\alpha_6.
\end{equation}
There is an effective curve $C_{57}$ corresponding to $\varpi_{29}$ and an effective cure $C'_5$ corresponding to $-\varpi_{32}$. 
 We thus have 
\begin{equation}
\begin{cases}
C_5\to C_{57}+C'_5\\
C_7\to 2C_{57}+C_4+C_6. 
\end{cases}
\end{equation}
Since in  the basis of fundamental weights, we have $-(\varpi_{29})= 
\boxed{\ $0$\ \ $0$\ $0$ \ 1 \ $-1$  \ $1$ \ \ $-1$ }$, we deduce  
\begin{equation}
D_4\cdot C_{57}=1, \quad D_5\cdot C_{57}=-1, \quad  D_6 \cdot C_{57}=1, \quad  D_7 \cdot C_{57}=-1, \quad D_i \cdot C_{57}=0, \quad i=1,2,3.
\end{equation}
We also note that by linearity $D_0 \cong - (2D_1+3 D_2 + 4 D_3 + 3 D_4+2D_5+D_6+2D_7)$, hence, 
\begin{equation}
D_0\cdot C_{57}=0.
\end{equation}
The negative intersection numbers imply both $D_5$ and $D_7$ contain $C_{57}$.  Meanwhile from $
-(-\varpi_{32})= 
\boxed{\ $0$\ \ 0\ $0$ \ 0 \ $-1$  \ $0$ \ \ 1 }$, we deduce the intersection of $C'_5$, namely
\begin{equation}
D_5\cdot C'_{5}=-1,  \quad  D_7 \cdot C'_{5}=1, \quad D_i \cdot C'_{5}=0, \quad i=1,2,3,4,6,
\end{equation}
and, by linearity,
\begin{equation}
D_0\cdot C_{57}=0.
\end{equation}
Finally we find the degeneration
\begin{equation}
C_0+2 C_1 + 3 C_2 + 4C_3 +3 C_4+ 2 C_5 + C_6 + 2  C_7\to 
C_0+2 C_1 + 3 C_2 + 4C_3 +5 C_4+ 6C_{57}+2C'_5 + 3C_6. 
\end{equation}
We get a fiber whose dual graph is the affine Dynkin diagram $\tilde{\text{E}}_8$ with the node corresponding to $\alpha_6$ contracted to a point and the identification: 
\begin{equation}
(C_0, C_{1}, C_2, C_3,C_{4}, C_{57}, C'_5, C_6)\to (\alpha_0, \alpha_1, \alpha_2,\alpha_3, \alpha_4, \alpha_5, \alpha_7, \alpha_8), 
\end{equation}
with the respective multiplicities $(1,2,3,4,5,6,2,3)$.

\subsection{Ch$_7$}

\vspace{-2em}

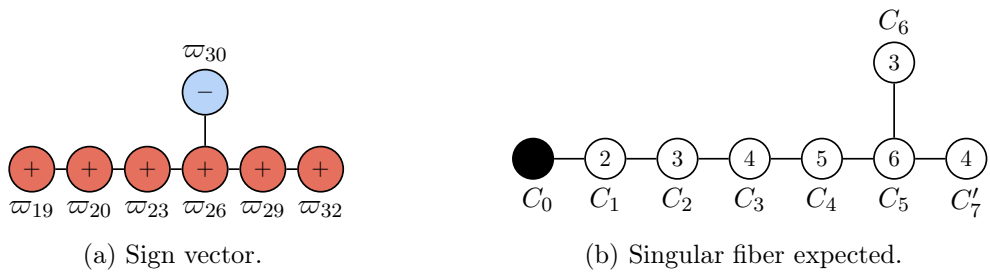
\begin{figure}[H]
\centering 
\begin{subfigure}[t]{0.45\textwidth}
\centering
\scalebox{.8}{
\begin{tikzpicture}[scale=.8]
		\node[fill=myred,draw,circle,thick,scale=1,label=below:{\scalebox{1.2}{$\varpi_{19}$}}] (1) at (1.2,0){$+$};
				\node[fill=myred,draw,circle,thick,scale=1,label=below:{\scalebox{1.2}{$\varpi_{20}$}}] (2) at (2.4,0){$+$};
				\node[fill=myred,draw,circle,thick,scale=1,label=below:{\scalebox{1.2}{$\varpi_{23}$}}] (3) at (3.6,0){$+$};
				\node[fill=myred,draw,circle,thick,scale=1,label=below:{\scalebox{1.2}{$\varpi_{26}$}}] (4) at (4.8,0){$+$};
				\node[fill=myred,draw,circle,thick,scale=1,label=below:{\scalebox{1.2}{$\varpi_{29}$}}] (5) at (6,0){$+$};
				\node[fill=myred,draw,circle,thick,scale=1,label=below:{\scalebox{1.2}{$\varpi_{32}$}}] (6) at (7.2,0){$+$};
				\node[fill=myblue,draw,circle,thick,scale=1, label=above:{\scalebox{1.2}{$\varpi_{30}$}}] (8) at (4.8,1.6){$-$};
				\draw[thick] (1)--(2)--(3)--(4)--(5)--(6);
				\draw[thick]  (4)--(8);
					\end{tikzpicture}}
\caption{Sign vector.}
\label{Fig:Ch7Sign}
\end{subfigure}
\begin{subfigure}[t]{0.45\textwidth}
\centering
\scalebox{.8}{
\begin{tikzpicture}
				\node[draw,circle,thick,scale=1,fill=black,label=below:{\scalebox{1.2}{ $C_0$}}] (0) at (0,0){$1$};
				\node[draw,circle,thick,scale=1,label=below:{\scalebox{1.2}{$C_{1}$}}] (2) at (1.2,0){$2$};
				\node[draw,circle,thick,scale=1,label=below:{\scalebox{1.2}{$C_{2}$}}] (3) at (2.4,0){$3$};
				\node[draw,circle,thick,scale=1,label=below:{\scalebox{1.2}{$C_{3}$}}] (4) at (3.6,0){$4$};
				\node[draw,circle,thick,scale=1,label=below:{\scalebox{1.2}{$C_{4}$}}] (5) at (4.8,0){$5$};
				\node[draw,circle,thick,scale=1,label=below:{\scalebox{1.2}{$C_{5}$}}] (6) at (6,0){$6$};
				\node[draw,circle,thick,scale=1,label=below:{\scalebox{1.2}{$C'_7$}}] (7) at (7.2,0){$4$};
				\node[draw,circle,thick,scale=1, label=above:{\scalebox{1.2}{$C_6$}}] (8) at (6,1.6){$3$};
				\draw[thick] (0)--(2)--(3)--(4)--(5)--(6)--(7);
				\draw[thick]  (6)--(8);
					\end{tikzpicture}}
					\caption{Singular fiber expected.}
					\label{Fig:Ch7Fib}
					\end{subfigure}
					
\caption{Chamber 7.}
\end{figure}

The interior wall is now $\varpi_{32}$ with
\begin{equation}
\quad \varpi_{32}\cdot \phi>0. 
\end{equation}
We conclude that $\varpi_{32}$ will correspond to an effective extremal curve in this chamber. 
Recalling its expression in terms of simple roots: 
\begin{equation}
\varpi_{32}=-\frac{1}{2}\alpha_4-\alpha_5-\frac{1}{2}\alpha_6+\frac{1}{2}\alpha_7,\\
\end{equation}
we have
\begin{equation}
 \alpha_7=2\varpi_{32}+\alpha_4+2\alpha_5 +\alpha_6.
\end{equation}
There is an effective curve $C'_{7}$ corresponding to $\varpi_{32}$, and we have
\begin{equation}
C_7\to C_4+2C_5+C_6+2C'_{7}. 
\end{equation}
Since in  the basis of fundamental weights, we have $
-(\varpi_{32})= 
\boxed{\ $0$\ \ 0\ $0$ \ 0 \ $1$  \ $0$ \ \ $-1$ }$, we deduce the intersections of $C'_7$, namely
\begin{equation}
{ D_5\cdot C'_{7}=1,  \quad  D_7 \cdot C'_{7}=-1, \quad D_i \cdot C'_{7}=0, \quad i=0,1,2,3,4,6,}
\end{equation}
which are the negative of those found for $C'_{5}$ in the previous chamber.

Finally, we find the degeneration
\begin{equation}
C_0+2 C_1 + 3 C_2 + 4C_3 +3 C_4+ 2 C_5 + C_6 + 2  C_7\to 
C_0+2 C_1 + 3 C_2 + 4C_3 +5 C_4+ 6 C_5 +4C'_{7}+ 3C_6.
\end{equation}
We get a fiber whose dual graph is the affine Dynkin diagram $\tilde{\text{E}}_8$ with the node corresponding to $\alpha_7$ contracted to a point and the identification: 
\begin{equation}
(C_0, C_{1}, C_2, C_3,C_{4}, C_{5}, C'_7, C_6)\to (\alpha_0, \alpha_1, \alpha_2,\alpha_3, \alpha_4, \alpha_5, \alpha_6, \alpha_8), 
\end{equation}
with the respective multiplicities $(1,2,3,4,5,6,4,3)$.

\subsection{Ch$_8$}

\vspace{-2em}

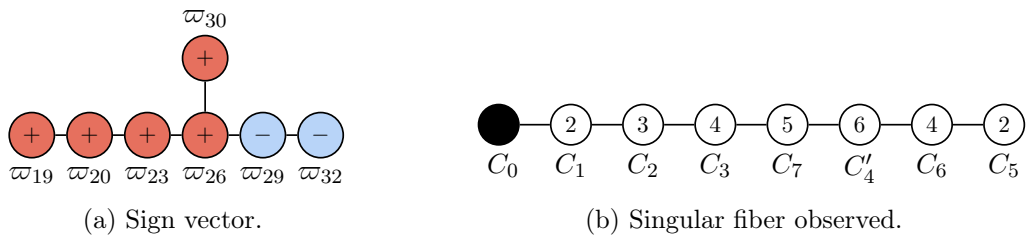
\begin{figure}[H]
\centering 
\begin{subfigure}[t]{0.45\textwidth}
\centering
\scalebox{.8}{
\begin{tikzpicture}[scale=.8]
		\node[fill=myred,draw,circle,thick,scale=1,label=below:{\scalebox{1.2}{$\varpi_{19}$}}] (1) at (1.2,0){$+$};
				\node[fill=myred,draw,circle,thick,scale=1,label=below:{\scalebox{1.2}{$\varpi_{20}$}}] (2) at (2.4,0){$+$};
				\node[fill=myred,draw,circle,thick,scale=1,label=below:{\scalebox{1.2}{$\varpi_{23}$}}] (3) at (3.6,0){$+$};
				\node[fill=myred,draw,circle,thick,scale=1,label=below:{\scalebox{1.2}{$\varpi_{26}$}}] (4) at (4.8,0){$+$};
				\node[fill=myblue,draw,circle,thick,scale=1,label=below:{\scalebox{1.2}{$\varpi_{29}$}}] (5) at (6,0){$-$};
				\node[fill=myblue,draw,circle,thick,scale=1,label=below:{\scalebox{1.2}{$\varpi_{32}$}}] (6) at (7.2,0){$-$};
				\node[fill=myred,draw,circle,thick,scale=1, label=above:{\scalebox{1.2}{$\varpi_{30}$}}] (8) at (4.8,1.6){$+$};
				\draw[thick] (1)--(2)--(3)--(4)--(5)--(6);
				\draw[thick]  (4)--(8);
					\end{tikzpicture}}
\caption{Sign vector.}
\label{Fig:Ch8Sign}
\end{subfigure}
\begin{subfigure}[t]{0.45\textwidth}
\centering
\scalebox{.8}{
\begin{tikzpicture}
				\node[draw,circle,thick,scale=1,fill=black,label=below:{\scalebox{1.2}{ $C_0$}}] (0) at (0,0){$1$};
				\node[draw,circle,thick,scale=1,label=below:{\scalebox{1.2}{$C_{1}$}}] (2) at (1.2,0){$2$};
				\node[draw,circle,thick,scale=1,label=below:{\scalebox{1.2}{$C_{2}$}}] (3) at (2.4,0){$3$};
				\node[draw,circle,thick,scale=1,label=below:{\scalebox{1.2}{$C_{3}$}}] (4) at (3.6,0){$4$};
				\node[draw,circle,thick,scale=1,label=below:{\scalebox{1.2}{$C_{7}$}}] (5) at (4.8,0){$5$};
				\node[draw,circle,thick,scale=1,label=below:{\scalebox{1.2}{$C'_{4}$}}] (6) at (6,0){$6$};
				\node[draw,circle,thick,scale=1,label=below:{\scalebox{1.2}{$C_6$}}] (7) at (7.2,0){$4$};
				\node[draw,circle,thick,scale=1, label=below:{\scalebox{1.2}{$C_5$}}] (8) at (8.4,0){$2$};
				\draw[thick] (0)--(2)--(3)--(4)--(5)--(6)--(7)--(8);
					\end{tikzpicture}}
					\caption{Singular fiber observed.}
					\label{Fig:Ch8Fib}
					\end{subfigure}
					
\caption{Chamber 8.}
\end{figure}

Chamber 8 is only adjacent to chamber 5 and the two chambers are separated by the hyperplane  $\varpi_{30}^\bot$ (as seen in Figure \ref{Figure:IG}). In particular, chamber 8 is characterized by: 
\begin{equation}
\quad \varpi_{30}\cdot \phi=\phi_4-\phi_6-\phi_7>0. 
\end{equation}
{We conclude that in this chamber, the weight $\varpi_{30}$}  will correspond to an effective extremal curve. 
Recalling its expression in terms of simple roots: 
\begin{equation}
\varpi_{30}=\frac{1}{2}\alpha_4-\frac{1}{2}\alpha_6-\frac{1}{2}\alpha_7,
\end{equation}
we have
\begin{equation}
 \alpha_4=2\varpi_{30} +\alpha_6+\alpha_7.
\end{equation}
There is an effective curve $C'_{4}$ corresponding to $\varpi_{30}$, and we have
\begin{equation}
C_4\to C_6+C_7+2C'_{4}. 
\end{equation}
Since in  the basis of fundamental weights, we have $
-(\varpi_{30})= 
\boxed{\ $0$\ \ 0\ $0$ \ $-1$ \ $0$  \ $1$ \ \ $1$ }$, we deduce the intersections of $C'_4$, namely
\begin{equation}
D_4\cdot C'_{4}=-1,  \quad  D_6 \cdot C'_{4}=1,\quad  D_7 \cdot C'_{4}=1, \quad D_i \cdot C'_{4}=0, \quad i=1,2,3,5,
\end{equation}
and by linearity
$D_0 \cong - (2D_1+3 D_2 + 4 D_3 + 3 D_4+2D_5+D_6+2D_7)$
hence, 
\begin{equation}
D_0\cdot C'_{4}=0.
\end{equation}
Finally we find the degeneration
 \begin{equation}
C_0+2 C_1 + 3 C_2 + 4C_3 +3 C_4+ 2 C_5 + C_6 + 2  C_7\to 
C_0+2 C_1 + 3 C_2 + 4C_3 +5C_7+6C'_{4} +4 C_6+ 2 C_5. 
\end{equation}
We get a fiber whose dual graph is the affine Dynkin diagram $\tilde{\text{E}}_8$ with the node corresponding to $\alpha_8$ contracted to a point and the identification: 
\begin{equation}
(C_0, C_{1}, C_2, C_3,C_{7}, C'_{4}, C_6, C_5)\to (\alpha_0, \alpha_1, \alpha_2,\alpha_3, \alpha_4, \alpha_5, \alpha_6, \alpha_7), 
\end{equation}
with the respective multiplicities $(1,2,3,4,5,6,4,2)$.

\section{Triple intersection computations}\label{intpoly}
Here we compute the triple intersection polynomials in each chamber for which we have an explicit resolution of singularities.

\subsection{Y$_4$}

$Y_4$ is  the proper transform of the Weierstrass model of equation 
\eqref{eq:E7} after the blowups leading to $X_7^{''}$ in~(\ref{eq:finalbl}).  The result is
\begin{equation}\label{eq:Y4}
Y_4:\quad\quad
e_3 e_5 e_6 y^2-e_1 e_2 e_4 (b e_1^2 e_3^2 e_4 e_6 e_7^2 s^5 + a e_1 e_3 s^3 x + e_2 e_5 x^3) =0,
\end{equation}
where the relative projective coordinates are
\begin{align}\label{eq:Y4projcoord}
\begin{aligned}
[e_2 e_4 e_5 e_6 e_7^2 x :e_2 e_3 e_4^2 e_5^2 e_6^3 e_7^5 y : s]\   [x: e_3 e_4 e_5 e_6^2 e_7^3 y:e_1 e_3 e_4 e_6 e_7^2 ]\\
   [e_5 e_6 e_7 y :e_1 ]\  
   [e_2 e_5 :e_3]\   [e_6 e_7 y :e_2]\  [y :e_4 e_7] \  [e_4:e_6 ].
  \end{aligned}
\end{align}
 The total transform of $s$ is $se_1 e_2 e_3 e_4^2 e_5 e_6^2 e_7^4 $ and we have the following fibral divisors 
\begin{equation}
\begin{cases}\label{div4}
1\   D_0: &\quad s=  e_3 e_6 y^2-e_1 e_2^2 e_4  x^3 =0,\\
2\  D_{1}: &\quad  e_1=e_3=0,\\
3\  D_{2}: &\quad e_3=e_4=0,\\
4\  D_{3}: &\quad e_7=e_3 e_5 e_6 y^2-e_1 e_2 e_4 (a e_1 e_3 s^3 x + e_2 e_5 x^3) =0,\\
3\  D_4: &\quad   e_4=e_5 =0,\\
2\  D_{5}: & \quad e_2=e_5=0, \\
    1\    D_6 : &\quad e_5=b e_1 e_3 e_4 e_6 e_7^2 s^2 + a x =0,\\
       2\   D_7 : &\quad   e_6=a e_1 e_3 s^3  + e_2 e_5 x^2=0.
\end{cases}
\end{equation}
The classes of the Cartier divisors defined by the zero loci of $s,x,y,e_i$ ($i=1,\cdots, 7$)  are 
\begin{equation}
\begin{cases}
[s]=S-E_1, \quad 
[x]=H+2L-E_1-E_2,\quad 
[y]=H+3L-E_1-E_2-E_3-E_5-E_6,\\
[e_1]=E_1-E_2-E_3, \quad  
[e_2]=E_2-E_4-E_5, \quad
[e_3]=E_3-E_4,\\
[e_4]=E_4-E_6-E_7, \quad 
[e_5]=E_5, \quad 
[e_6]=E_6-E_7, \quad 
[e_7]=E_7,
\end{cases}
\end{equation}
where E$_i$ is the total transform of the $i^{th}$ exceptional divisor and $[e_i]$ is the proper transform of the $i^{th}$ exceptional divisor. 

We have the linear relations
\be\begin{array}{c}
~[e_1]=D_1,~~[e_2]=D_5,~~[e_3]=D_1+D_2,~~[e_4]=D_2+D_4,\\
~[e_5]=D_4+D_5+D_6,~~[e_6]=D_7,~~[e_7]=D_3,
\end{array}\ee
and can thus solve for the $D_i$ in terms of the $E_i$ to get\footnote{Were we to only use the sequence of blowups described in~(\ref{Ch458}) and~(\ref{Ch456}) without the additional~(\ref{addbl}) one would not be able to invert the equations for $E_i$ in terms of the $D_i$.}
\begin{align}\label{eq:divclass}
\begin{cases}
D_0\quad  & =S-E_1, \\
D_1 \quad &  =  E_1-E_2-E_3, \\ 
D_2 \quad & = -E_1+E_2+2E_3-E_4, \\ 
D_3 \quad & =E_7, \\ 
D_4 \quad &  =E_1-E_2-2E_3+2E_4-E_6-E_7,  \\ 
D_5 \quad  & = E_2-E_4-E_5,\\
D_6 \quad & =-E_1+2E_3-E_4+2E_5+E_6+E_7,   \\ 
D_7 \quad & =E_6-E_7.
\end{cases}
\end{align}
Now that we have the  classes of the fibral divisors, the sequence of blowups and the pushforward theorems will be enough to compute the triple intersection numbers.

The triple intersection polynomial is by definition 
\begin{equation}
F= \int (\sum_{a=0}^7D_a \phi_a)^3 [Y]=\int_B \pi_* f_{1*} f_{2*} f_{3*} f_{4*} f_{5*} f_{6*} f_{7*} f_{8*} \Big[\big({\sum_{a=0}^7D_a \phi_a}\big)^3 [Y]\Big],
\end{equation}
where $f_i$ is the $i^{th}$ blowup and $\pi: X_0=\mathbb{P}[\mathscr{O}_B\oplus\mathscr{L}^{\otimes 2}\oplus \mathscr{L}^{\otimes 3}]\to B$ is the map defining the projective bundle. 
Noting that
\be
[Y_4]=3H+6L-2E_1-2E_2-E_3-E_4-E_5-E_6-E_7,
\ee
we can use the pushforward theorems from Section \ref{Sec:Intersection} to get:
\begin{equation}
\begin{aligned}
F_4(\phi)= &\phantom{+}   
 4 S(L-S)  (\phi _0^3 +  \phi _1^3+\phi _2^3+\phi _4^3+ \phi _5^3+ \phi _7^3)-S^2 \phi _3^3    +2 S(S-2 L)  \phi _6^3 \\
&+3S (-2 L + S) \phi _0^2 \phi _1+ 3 LS\phi _0\phi _1^2 +
3  S (-3 L + 2 S) \phi _1^2 \phi _2+3 S(2 L - S)  \phi_1 \phi_2^2\\
& +3S (3 S-4 L)( \phi _3 \phi _2^2 +\phi_3^2 \phi_6+\phi_3 \phi_6^2+2  \phi_4^2 \phi_6+2\phi_6\phi_7^2)\\
& +3S (3 L-2 S)\phi _3^2 \phi _2+3S(9 S-11 L)  \phi _5^2 \phi _6 +6 S(5 L-4 S) \phi _5  \phi _6^2\\
&+3 S (S-L) ( \phi _3^2\phi _4+\phi _3^2\phi _7+2 \phi _4^2 \phi _5- \phi _4 \phi _5^2)\\
& -6S (3 S-4 L)( \phi_3 \phi_4 \phi_6+\phi_4 \phi_5 \phi_6 +\phi_3 \phi_6 \phi_7).
\end{aligned}
\end{equation}

\subsection{Y$_5$}

$Y_5$ is  the proper transform of the Weierstrass model of equation 
\eqref{eq:E7} after the blowups leading to  $X_7^{+}$  in~(\ref{eq:finalbl}).  The result is
\begin{equation}\label{eq:Y5}
Y_5:\quad\quad
e_3 e_5 e_6 y^2-e_1 e_2 e_4 (b e_1^2 e_3^2 e_4 e_5 e_7 s^5 + a e_1 e_3 s^3 x + e_2 e_6 e_7 x^3) =0,
\end{equation}
where the relative projective coordinates are
\begin{align}\label{eq:Y5projcoord}
\begin{aligned}
[e_2 e_4 e_5 e_6 e_7^2 x :e_2 e_3 e_4^2 e_5^3 e_6^2 e_7^4 y : s]\   [x: e_3 e_4 e_5^2 e_6 e_7^2 y:e_1 e_3 e_4 e_5 e_7 ]\\
   [e_5 e_6 e_7 y :e_1 ]\  
   [e_2 e_6 e_7 :e_3]\   [e_6  y :e_4]\  [y :e_2 e_7] \  [e_2:e_5 ].
  \end{aligned}
\end{align}
 The total transform of $s$ is $se_1 e_2 e_3 e_4^2 e_5^2 e_6 e_7^3 $ and we have the following fibral divisors 
\begin{equation}
\begin{cases}\label{div5}
1\   D_0: &\quad s=  e_3 e_5 e_6 y^2-e_1 e_2^2 e_4 e_6 e_7 x^3 =0,\\
2\  D_{1}: &\quad  e_1=e_3=0,\\
3\  D_{2}: &\quad e_3=e_4=0,\\
4\  D_{3}: &\quad e_4=e_5 =0, \\
3\  D_4: &\quad  e_7=e_5 e_6 y^2-a e_1^2 e_2  e_4 s^3 x =0,  \\
2\  D_{5}: & \quad e_2=e_6=0, \\
    1\    D_6 : &\quad e_6=b e_1 e_3 e_4 e_5 e_7 s^2 + a  x =0,\\
       2\   D_7 : &\quad   e_5=a e_1 e_3 s^3 + e_2 e_6 e_7 x^2=0.

\end{cases}
\end{equation}
The classes of the Cartier divisors defined by the zero loci of the variables $s,x,y,e_i$  are 
\begin{equation}
\begin{cases}
[s]=S-E_1, \quad 
[x]=H+2L-E_1-E_2,\quad 
[y]=H+3L-E_1-E_2-E_3-E_5-E_6,\\
[e_1]=E_1-E_2-E_3, \quad  
[e_2]=E_2-E_4-E_6-E_7, \quad
[e_3]=E_3-E_4,\\
[e_4]=E_4-E_5, \quad 
[e_5]=E_5-E_7, \quad 
[e_6]=E_6, \quad 
[e_7]=E_7,
\end{cases}
\end{equation}
where E$_i$ is the total transform of the $i^{th}$ exceptional divisor and $[e_i]$ is the proper transform of the $i^{th}$ exceptional divisor.

We have the linear relations
\be\begin{array}{c}
~[e_1]=D_1,~~[e_2]=D_5,~~[e_3]=D_1+D_2,~~[e_4]=D_2+D_3,\\
~[e_5]=D_3+D_7,~~[e_6]=D_5+D_6,~~[e_7]=D_4,
\end{array}\ee
and can thus solve for the $D_i$ in terms of the $E_i$ to get
\begin{align}\label{eq:divclass5}
\begin{cases}
D_0\quad  & =S-E_1, \\
D_1 \quad &  =  E_1-E_2-E_3, \\ 
D_2 \quad & =-E_1+E_2+2E_3-E_4, \\ 
D_3 \quad & =E_1-E_2-2E_3+2E_4 -E_5,\\ 
D_4 \quad &  =E_7,  \\ 
D_5 \quad  & = E_2-E_4-E_6-E_7,\\
D_6 \quad & =-E_2+E_4+2E_6+E_7,   \\ 
D_7 \quad & =-E_1+E_2+2E_3-2E_4 +2E_5-E_7.
\end{cases}
\end{align}
Now that we have the  classes of the fibral divisors, the sequence of blowups and the pushforward theorems will be enough to compute the triple intersection numbers. 
Noting that
\be
[Y_5]=3H+6L-2E_1-2E_2-E_3-E_4-E_5-E_6-E_7,
\ee
we can use the pushforward theorems from Section \ref{Sec:Intersection} to get: 
\begin{align}
\begin{aligned}
F_5(\phi)=& 4 S(L-S)( \phi _0^3 + \phi _1^3 +\phi_2^3+ \phi _3^3+ \phi _5^3)-S^2 (\phi _4^3+ \phi _6^3 +\phi _7^3)\\
&
+3S(3 S-4 L)( \phi _3 \phi _2^2   -\phi _3 \phi _4^2+ \phi _4 \phi _6^2+ \phi _4 \phi _7^2+ \phi _6 \phi _7^2+ \phi _4^2 \phi _6+ \phi _6^2 \phi _7-\phi _3 \phi _7^2 +\phi _7 \phi _4^2)\\
&+
3 S \phi _1 \phi _0^2 (S-2 L)+3L S \phi _1^2 \phi _0+3S (3 L-2 S)( \phi _2 \phi _3^2- \phi _1^2 \phi _2)
\\
&+ 3S(5 L-4 S)( 2 \phi _5 \phi _6^2 - \phi _3^2 \phi _4 - \phi _3^2 \phi _7 )+3
S(2 L-S) \phi _1 \phi _2^2 +3  S  (9 S-11 L)\phi _5^2 \phi _6
\\
&  +3S(S-L)(2 \phi _4^2 \phi _5 - \phi _4 \phi _5^2)
 + 6 S(4L-3S) (\phi _4 \phi _5 \phi _6+\phi _3 \phi _4 \phi _7+\phi _4 \phi _6 \phi _7).
\end{aligned}
\end{align}

\subsection{Y$_6$}

$Y_6$ is  the proper transform of the Weierstrass model of equation 
\eqref{eq:E7} after the blowups leading to  $X_7^{-}$   in~(\ref{eq:finalbl}).  The result is
\begin{equation}\label{eq:Y6}
Y_6:\quad\quad
e_3 e_5 e_6 y^2-e_1 e_2 e_4 (b e_1^2 e_3^2 e_4 e_5 s^5 + a e_1 e_3 s^3 x + e_2 e_6 e_7^2 x^3) =0,
\end{equation}
where the relative projective coordinates are
\begin{align}\label{eq:Y4projcoord}
\begin{aligned}
[e_2 e_4 e_5 e_6 e_7^2 x :e_2 e_3 e_4^2 e_5^3 e_6^2 e_7^3 y : s]\   [x: e_3 e_4 e_5^2 e_6 e_7 y:e_1 e_3 e_4 e_5  ]\\
   [e_5 e_6 e_7 y :e_1 ]\  
   [e_2 e_6 e_7^2 :e_3]\   [e_6 e_7 y :e_4]\  [y :e_2 e_7] \  [e_2:e_6 ].
  \end{aligned}
\end{align}
 The total transform of $s$ is $se_1 e_2 e_3 e_4^2 e_5^2 e_6 e_7^2 $ and we have the following fibral divisors 
\begin{equation}
\begin{cases}\label{div4}
1\   D_0: &\quad s=  e_3 e_5 y^2-e_1 e_2^2 e_4 e_7^2 x^3 =0,\\
2\  D_{1}: &\quad  e_1=e_3=0,\\
3\  D_{2}: &\quad e_3=e_4=0,\\
4\  D_{3}: &\quad e_4=e_5 =0, \\
3\  D_4: &\quad  e_2=e_5=0,   \\
2\  D_{5}: & \quad  e_7=e_5 e_6 y^2-a e_1^2 e_2  e_4 s^3 x-be_1^3e_2e_3e_4^2e_5s^5 =0,\\
    1\    D_6 : &\quad e_6=b e_1 e_3 e_4 e_5  s^2 + a  x =0,\\
       2\   D_7 : &\quad   e_5=a e_1 e_3 s^3 + e_2 e_6 e_7^2 x^2=0.
\end{cases}
\end{equation}
The classes of the Cartier divisors defined by the zero loci of the variables $s,x,y,e_i$  are 
\begin{equation}
\begin{cases}
[s]=S-E_1, \quad 
[x]=H+2L-E_1-E_2,\quad 
[y]=H+3L-E_1-E_2-E_3-E_5-E_6,\\
[e_1]=E_1-E_2-E_3, \quad  
[e_2]=E_2-E_4-E_6-E_7, \quad
[e_3]=E_3-E_4,\\
[e_4]=E_4-E_5, \quad 
[e_5]=E_5, \quad 
[e_6]=E_6-E_7, \quad 
[e_7]=E_7,
\end{cases}
\end{equation}
where E$_i$ is the total transform of the $i^{th}$ exceptional divisor and $[e_i]$  is the proper transform of the $i^{th}$ exceptional divisor. 

We have the linear relations
\be\begin{array}{c}
~[e_1]=D_1,~~[e_2]=D_4,~~[e_3]=D_1+D_2,~~[e_4]=D_2+D_3,\\
~[e_5]=D_3+D_4+D_7,~~[e_6]=D_6,~~[e_7]=D_5,
\end{array}\ee
and can thus solve for the $D_i$ in terms of the $E_i$ to get
\begin{align}\label{eq:divclass}
\begin{cases}
[D_0]\quad  & =S-E_1, \\
[D_1] \quad &  =  E_1-E_2-E_3, \\ 
[D_2] \quad & =-E_1+E_2+2E_3-E_4, \\ 
[D_3] \quad & =E_1-E_2-2E_3+2E_4-E_5,\\ 
[D_4] \quad &  =E_2-E_4-E_6-E_7, \\ 
[D_5] \quad  & = E_7,\\
[D_6] \quad & =E_6-E_7,   \\ 
[D_7] \quad & =-E_1+2E_3-E_4+2E_5+E_6+E_7.
\end{cases}
\end{align}
Now that we have the  classes of the fibral divisors, the sequence of blowups and the pushforward theorems will be enough to compute the triple intersection numbers. 
Noting that
\be
[Y_6]=3H+6L-2E_1-2E_2-E_3-E_4-E_5-E_6-E_7,
\ee
we can use the pushforward theorems from Section \ref{Sec:Intersection} to get: 
\begin{equation}
\begin{aligned}
F_6(\phi) = &\phantom{+}     
4 S  (L-S)( \phi _0^3+ \phi _1^3 +\phi _2^3 + \phi _3^3 + \phi _4^3 + \phi _6^3)-S^2 \phi _5^3-2 S  (2 L-S)\phi _7^3\\
& +6S(4L-3S) ( \phi _3 \phi _4 \phi _7+ \phi _4 \phi _5 \phi _7 + \phi _5 \phi _6 \phi _7)+ 
3 L S \phi _0 \phi _1^2\\
& +3 S  (3 S-4 L) (\phi _3 \phi _2^2-\phi _3 \phi _4^2-\phi _3 \phi _7^2+\phi _5 \phi _7^2+2 \phi _4^2 \phi _7+\phi _5^2 \phi _7+2 \phi _6^2 \phi _7) \\
& +
3S  (S-2 L) (\phi _1 \phi _0^2- \phi _1 \phi _2^2)
+3S(3 L-2 S) (\phi _2 \phi _3^2- \phi _1^2 \phi _2)+3S  (6 S-7 L)\phi _5^2 \phi _6\\
& +3S(5 L-4 S)(\phi _4 \phi _5^2- \phi _3^2 \phi _7 - \phi _3^2 \phi _4)
+3S(6L-5S) (\phi _5 \phi _6^2- \phi _4^2 \phi _5). 
\end{aligned}
\end{equation}

\subsection{Y$_8$}

$Y_8$ is  the proper transform of the Weierstrass model of equation 
\eqref{eq:E7} after the blowups leading to  $X_7^{'}$   in~(\ref{eq:finalbl}).  The result is
\begin{equation}\label{eq:Y8}
Y_8:\quad\quad
e_3 e_5 e_7 y^2-e_1 e_2 e_4 (b e_1^2 e_3^2 e_4 e_5 e_6 s^5 + a e_1 e_3 s^3 x + e_2 e_6 e_7 x^3) =0,
\end{equation}
where the relative projective coordinates are
\begin{align}\label{eq:Y4projcoord}
\begin{aligned}
[e_2 e_4 e_5 e_6^2 e_7 x :e_2 e_3 e_4^2 e_5^3 e_6^4 e_7^2 y : s]\   [x: e_3 e_4 e_5^2 e_6^2 e_7 y:e_1 e_3 e_4 e_5 e_6 ]\\
   [e_5 e_6 e_7 y :e_1 ]\  
   [e_2 e_6 e_7 :e_3]\   [e_7  y :e_4]\ [e_2e_7:e_5]\ [y :e_2 ].
  \end{aligned}
\end{align}
 The total transform of $s$ is $se_1 e_2 e_3 e_4^2 e_5^2 e_6^3 e_7 $ and we have the following fibral divisors 
\begin{equation}
\begin{cases}\label{div4}
1\   D_0: &\quad s=  e_3 e_5  y^2-e_1 e_2^2 e_4 e_6  x^3 =0,\\
2\  D_{1}: &\quad  e_1=e_3=0,\\
3\  D_{2}: &\quad e_3=e_4=0,\\
4\  D_{3}: &\quad e_4=e_5 =0, \\
3\  D_4: &\quad  e_6=e_5 e_7 y^2-a e_1^2 e_2  e_4 s^3 x =0,  \\
2\  D_{5}: & \quad e_2=e_7=0, \\
    1\    D_6 : &\quad e_7=b e_1 e_3 e_4 e_5 e_6 s^2 + a  x =0,\\
       2\   D_7 : &\quad   e_5=a e_1 e_3 s^3 + e_2 e_6 e_7 x^2=0.

\end{cases}
\end{equation}
The classes of the Cartier divisors defined by the zero loci of the variables $s,x,y,e_i$  are 
\begin{equation}
\begin{cases}
[s]=S-E_1, \quad 
[x]=H+2L-E_1-E_2,\quad 
[y]=H+3L-E_1-E_2-E_3-E_5-E_7,\\
[e_1]=E_1-E_2-E_3, \quad  
[e_2]=E_2-E_4-E_6-E_7, \quad
[e_3]=E_3-E_4,\\
[e_4]=E_4-E_5, \quad 
[e_5]=E_5-E_6, \quad 
[e_6]=E_6, \quad 
[e_7]=E_7,
\end{cases}
\end{equation}
where E$_i$ is the total transform of the $i^{th}$exceptional divisor and $[e_i]$ is the proper transform of the $i^{th}$ exceptional divisor.

We have the linear relations
\be\begin{array}{c}
~[e_1]=D_1,~~[e_2]=D_5,~~[e_3]=D_1+D_2,~~[e_4]=D_2+D_3,\\
~[e_5]=D_3+D_7,~~[e_6]=D_4,~~[e_7]=D_5+D_6,
\end{array}\ee
and can thus solve for the $D_i$ in terms of the $E_i$ to get
\begin{align}\label{eq:divclass}
\begin{cases}
[D_0]\quad  & =S-E_1, \\
[D_1] \quad &  =  E_1-E_2-E_3, \\ 
[D_2] \quad & =-E_1+E_2+2E_3-E_4, \\ 
[D_3] \quad & =E_1-E_2-2E_3+2E_4 -E_5,\\ 
[D_4] \quad &  =E_6,  \\ 
[D_5] \quad  & = E_2-E_4-E_6-E_7,\\
[D_6] \quad & =-E_2+E_4+E_6+2E_7,  \\ 
[D_7] \quad & =-E_1+E_2+2E_3-2E_4 +2E_5 -E_6. 
\end{cases}
\end{align}
Now that we have the  classes of the fibral divisors, the sequence of blowups and the pushforward theorems will be enough to compute the triple intersection numbers. 
Noting that
\be
[Y_8]=3H+6L-2E_1-2E_2-E_3-E_4-E_5-E_6-E_7,
\ee
we can use the pushforward theorems from Section \ref{Sec:Intersection} to get: 
\begin{align}
\begin{aligned}
F_8(\phi)=&\    4 S (L-S)\phi _0^3 
+3 L S \phi _0 \phi _1^2 +3 S (S-2 L) \phi _0^2  \phi _1 \\
&+4 S  (L-S)(\phi _1^3+\phi _2^3+ \phi _3^3 +\phi _5^3 + \phi _6^3+ \phi _7^3) +2 S (S-2 L) \phi _4^3 \\
&+3 S  (2 L-S) \phi _1 \phi _2^2+3 S  (2 S-3 L) \phi _1^2 \phi _2  +3 S (3 L-2 S) \phi _2 \phi _3^2+3 S (3 S-4 L) \phi _2^2 \phi _3\\
&+6 S (5 L-4 S) \phi _5 \phi _6^2 +3 S (9 S-11 L)  \phi _5^2 \phi _6 +3 S (4 S-5 L)  \phi _3^2 \phi _7\\
&+3 S (L-S) \phi _4 \phi _5^2 +3 S (4 S-5 L)\phi _3^2 \phi _4 +6 S  (S-L) \phi _4^2 \phi _5\\
&+6S(3S-4L)\phi _4 \left( \phi _6^2 -\phi _5 \phi _6+\phi _7^2\right)+ 3S(4 L-3 S) \phi _3 ( \phi _4+\phi _7)^2.
\end{aligned}
\end{align}

Now that we have the triple intersection polynomial for each chamber, we conclude this appendix with a brief discussion of how to use this data to learn about the geometry of the fibral divisors. In particular, a necessary condition for a divisor D$_i$ to be a $\mathbb{P}^1$-bundle without singular fibers is that
\begin{equation}
D_i^3 =4S(L-S). 
\end{equation}
By looking at the  Fermat terms of  the triple intersection polynomials F$_8$, F$_6$, F$_5$, and F$_4$, we see that we recover the following information: 
\begin{enumerate}
\item In Ch$_4$, D$_3$ and D$_6$ are not  $\mathbb{P}^1$-bundles.
\item In Ch$_5$, D$_4$, D$_6$, and D$_7$ are not $\mathbb{P}^1$-bundles.
\item In Ch$_6$, D$_5$ and D$_7$ are not $\mathbb{P}^1$-bundles.
\item In Ch$_8$, D$_4$ is not a $\mathbb{P}^1$-bundle. 
 \end{enumerate} 
We note that these conclusions are consistent with the analysis in \cite{E7}. 

For example, in chamber 8, since the divisors D$_a$ for $a=0,1,2,3,5,6,7$ have fibers that do not degenerate, they are projective bundles. We can check that their triple intersection numbers are as expected: 
\begin{align}
D_0^3 =D_1^3=D_2^3=D_3^3=D_5^3=D_6^3=D_7^3=4 (L - S) S.
\end{align}
The divisor D$_4$ has a fiber that degenerates with the appearance of two new curves over $V(a,s)$. This is also reflected in its triple intersection: 
\begin{equation}
D_4^3 = 2S(-2L+S),
\end{equation}
which differs from that of a projective bundle over $S$ by 2 for each point of $V(a,s)$:
\begin{equation}
D_4^3=4(L-S) S-2(4L-3S)S=4(L-S) S - 2 [a]. S,
\end{equation}
where we used 
\begin{equation}
(4L-3S)S= [a]\cdot [s].
\end{equation} 
We see D$_4$ has the same self-triple intersection as a projective bundle with $2[a].S$ points blown-up.

\section{Fibral divisors from scaling}\label{fibdivpf}
In this appendix, we demonstrate an alternate route to Table~\ref{Table:Div} via scaling methods, using $D_3$ as an example.
The fibral divisor D$_3$ is not a projective bundle for Y$_4$ since there the curve C$_3$ can degenerate in codimension-two. For Y$_5$, Y$_6$, and Y$_8$, the fibral divisor D$_3$ is the same $\mathbb{P}^1$-bundle up to isomorphism and we will now determine its isomorphism class. 
It is enough to focus on the first 5 blowups these varieties have in common. Our divisor is defined by
\begin{equation}
D_3: e_4=e_5=0,
\end{equation}
and on this locus, we have the coordinates
\begin{equation}
[0:0:\ell_1 s]\  [\ell_1 \ell_2 x:0:0] [0:\ell_1^{-1} \ell_2 \ell_3 e_1] [\ell_2^{-1} \ell_4 e_2:\ell_4 \ell_3^{-1} e_3] [\ell_1 \ell_2 \ell_3 \ell_5 y:0],
\end{equation}
where we have included the relevant rescaling factors.
\begin{equation}
\begin{array}{|c|c|c|c|c|c|c|c|c|c|c|c|}
\hline
 & s & x& y & e_1 & e_2 & e_3 & e_4 & e_5 \\
\hline
\ell_1& 1 & 1 & 1 & -1 & 0 & 0 & 0 & 0  \\
 \hline
\ell_2 & 0 & 1 & 1 & 1 & -1 & 0 & 0 & 0  \\
 \hline
\ell_3 & 0 & 0 & 1 & 1 & 0 & -1 & 0 & 0  \\
 \hline
\ell_4 & 0 & 0 & 0 & 0 & 1 & 1 & -1 & 0  \\
 \hline
\ell_5 & 0 & 0 & 1 & 0 & 0 & 0 & 1 & -1  \\
 \hline
\end{array}
\end{equation}
We recall that the components of a given set of projective coordinates cannot be simultaneously zero.  
Thus, the fibral divisor  D$_3$ is defined in  the patch 
\begin{equation}
s x y e_1 \neq 0.
\end{equation}
To normalize the coordinates $[0:0: s] [x:0:0] [0:e_1] [y:0]$ to $[0:0:1] [1:0:0] [0:1][1:0]$, we take 
\begin{equation}
\ell_1= s^{-1},\quad  \ell_2=s x^{-1},\quad  \ell_3 = e_1^{-1}  s^{-2}  x,\quad  \ell_5=   e_1 s^2 y^{-1}.
\end{equation} 
This implies that the fiber is 
\begin{equation}
[e_2 \frac{x}{s}: e_1e_3 \frac{s^2}{x}]\cong[e_2 x^2: e_1e_3 s^3],
\end{equation}
and we deduce that 
\begin{equation}
D_3\cong \mathbb{P}_S(\mathscr{S}^{\otimes 3}\oplus \mathscr{L}^{\otimes 4}),
\end{equation}   
which agrees with the corresponding entries in Table~\ref{Table:Div}.

\section{Vertical rational surfaces  
}\label{highercodim}
In this appendix, we prove Theorem \ref{thm:Q} by analyzing the isomorphism class of the curve C$_6$ over the locus $V(a,b)\cap S$. 
This requires  a careful analysis of the projective space defined from $X_0$ by the sequence of blowups. 
\subsection{The vertical surface Q$_8$}

In the case of $Y_8$, the defining equation for C$_6$ is 
\begin{equation}
C_6: \quad e_7=b e_1 e_3e_4e_5e_6 s^2 + a x=0.
\end{equation}
The projective coordinates of the fiber of $X_7$ over $X_0$ are: 
\begin{equation}
[0:0:s][x:0: e_1 e_3 e_4e_5 e_6][0:e_1][0:e_3][0:e_4][0:e_5][y:e_2], 
\end{equation}
which shows C$_6$ is defined in the open patch 
\begin{equation}se_1 e_3 e_4 e_5 \neq 0.\end{equation}
The scaling symmetries due to the respective blowups from $X_0$ to $X'_7$ are:
\begin{equation}
\begin{array}{|c|c|c|c|c|c|c|c|c|c|c|}
\hline
X'_7& s & x& y& e_1 & e_2 & e_3 & e_4 & e_5 & e_6 & e_7 \\
\hline 
\ell_1 & 1 & 1 & 1 & -1 & 0 & 0 & 0 & 0 & 0 & 0 \\
\hline
\ell_2 & 0 & 1 & 1 & 1 & -1 & 0 & 0 & 0 & 0 & 0 \\
\hline
\ell_3 &0 &  0 &  1 & 1& 0& -1& 0& 0& 0& 0\\
\hline
\ell_4 &0 &  0 &  0 & 0& 1& 1& -1& 0& 0& 0\\
\hline
\ell_5 &0 &  0 &  1 & 0& 0& 0& 1& -1& 0& 0\\
\hline 
\ell_6 &0 &  0 &  0 & 0& 1& 0& 0& 1& -1& 0\\
\hline 
\ell_7 &0 &  0 &  1 & 0& 	1& 0& 0& 0& 0& -1\\
\hline 
\end{array}
\end{equation}
We introduce the following linear redefinitions:
\begin{equation}
\begin{array}{|l|c|c|c|c|c|c|c|c|c|c|}
\hline 
X'_7 & s & x & y &e_1& e_2 & e_3& e_4& e_5&e_6& e_7 \\
 \hline 
 \ell'_1=\ell_1+\ell_3+\ell_4+\ell_5+\ell_6
  & 1 & 1 & 3 & 0 & 2 & 0 & 0 & 0 & -1 & 0 \\
  \hline 
 \ell'_2=\ell_2-\ell_3-\ell_4-\ell_5-\ell_6+3\ell_7
 & 0 & 1 & 2 & 0 & 0 & 0 & 0 & 0 & 1 & -3 \\
 \hline 
\ell'_3= \ell_3+\ell_4+\ell_5+\ell_6
 & 0 & 0 & 2 & 1 & 2 & 0 & 0 & 0 & -1 & 0 \\
 \hline 
\ell'_4= \ell_4+\ell_5+\ell_6& 0 & 0 & 1 & 0 & 2 & 1 & 0 & 0 & -1 & 0 \\
 \hline 
\ell'_5=\ell_5+\ell_6& 0 & 0 & 1 & 0 & 1 & 0 & 1 & 0 & -1 & 0 \\
\hline 
\ell'_6= \ell_6& 0 & 0 & 0 & 0 & 1 & 0 & 0 & 1 & -1 & 0 \\
 \hline 
\ell'_7=\ell_7 & 0 & 0 & 1 & 0 & 1 & 0 & 0 & 0 & 0 & -1 \\
\hline 
\end{array}
\end{equation}
We fix $(s, e_1, e_3, e_4, e_5)$ by using $(\ell'_1, \ell'_3,\ell'_4, \ell'_5, \ell'_6)$, respectively.
   Then, after imposing $e_7=0$,  we are left with: 
\begin{equation}
\begin{array}{|l|c|c|c|c|}
\hline
Q_8 & x & y& e_2 & e_6 \\
\hline 
\ell'_2=\ell_2-\ell_3-\ell_4-\ell_5-\ell_6+3\ell_7& 1 &2 & 0 & 1 \\
\hline
\ell'_7=\ell_7 & 0 & 1 & 1 & 0\\
\hline
\end{array}
\end{equation}
which is the toric description of the Hirzebruch surface $\mathbb{F}_2$.

\subsection{The vertical surface Q$_6$}
In the case of $Y_6$, the defining equation for C$_6$ is 
\begin{equation}
 C_6: \quad e_6=b e_1 e_3 e_4 e_5  s^2 + a  x =0.
 \end{equation}
Imposing $e_6=0$ gives the following projective coordinates
\begin{align}
\begin{aligned}
[0 :0 : s]\   [x:  0:e_1 e_3 e_4 e_5  ]\ 
   [0 :e_1 ]\  
   [0 :e_3]\   [0 :e_4]\  [y :e_2 e_7] \  [e_2:0 ],
  \end{aligned}
\end{align}
which implies that 
\begin{equation}
s e_1e_2 e_3 e_4\neq 0.
\end{equation}
The defining equation of C$_6$ gives a full rational surface Q$_6$ when $a=b=0$. 
 The successive blowups that produced X$^-_7$ give the following scalings: 
\begin{equation}
\begin{array}{|c|c|c|c|c|c|c|c|c|c|c|c|}
\hline
X^-_7 & s & x& y & e_1 & e_2 & e_3 & e_4 & e_5 & e_6 & e_7\\
\hline
\ell_1& 1 & 1 & 1 & -1 & 0 & 0 & 0 & 0 & 0 & 0 \\
 \hline
\ell_2 & 0 & 1 & 1 & 1 & -1 & 0 & 0 & 0 & 0 & 0 \\
 \hline
\ell_3 & 0 & 0 & 1 & 1 & 0 & -1 & 0 & 0 & 0 & 0 \\
 \hline
\ell_4 & 0 & 0 & 0 & 0 & 1 & 1 & -1 & 0 & 0 & 0 \\
 \hline
\ell_5 & 0 & 0 & 1 & 0 & 0 & 0 & 1 & -1 & 0 & 0 \\
 \hline
\ell_6 & 0 & 0 & 1 & 0 & 1 & 0 & 0 & 0 & -1 & 0 \\
 \hline
\ell_7 & 0 & 0 & 0 & 0 & 1 & 0 & 0 & 0 & 1 & -1 \\
 \hline
\end{array}
\end{equation}
We conveniently redefined them as follows:
\begin{equation}
\begin{array}{|c|c|c|c|c|c|c|c|c|c|c|}
\hline
 X^-_7& s & x & y & e_1 & e_2 & e_3& e_4&e_5&e_6& e_7 \\
 \hline
 \ell_1+\ell_3+\ell_4+\ell_5-\ell_6
 & 1 & 1 & 2 & 0 & 0 & 0 & 0 & -1 & 1 & 0 \\
 \hline
 \ell_2-\ell_3-\ell_4-\ell_5+\ell_7 & 0 & 1 & -1 & 0 & -1 & 0 & 0 & 1 & 1 & -1 \\
 \hline
 \ell_3+\ell_4+\ell_5-\ell_6
  & 0 & 0 & 1 & 1 & 0 & 0 & 0 & -1 & 1 & 0 \\
 \hline
\ell_4+\ell_5-\ell_6
 & 0 & 0 & 0 & 0 & 0 & 1 & 0 & -1 & 1 & 0 \\
 \hline
 \ell_5-\ell_6+\ell_7
 & 0 & 0 & 0 & 0 & 0 & 0 & 1 & -1 & 2 & -1 \\
 \hline
 \ell_2-\ell_3-\ell_4-\ell_5+2\ell_6& 0 & 1 & 1 & 0 & 0 & 0 & 0 & 1 & -2 & 0 \\
 \hline
 \ell_6-\ell_7 & 0 & 0 & 1 & 0 & 0 & 0 & 0 & 0 & -2 & 1 \\
 \hline
\end{array}
\end{equation}
We can then fix $(s,e_1, e_2, e_3, e_4)$ by using $(\ell'_1,\ell'_2,\ell'_3,\ell'_4,\ell'_5)$,  respectively, and after imposing $e_6=0$,  we are left with: 
\begin{equation}
\begin{array}{|l|c|c|c|c|}
\hline
Q_6&  x & y & e_5 & e_7 \\
 \hline
\ell'_6=\ell_2-\ell_3-\ell_4-\ell_5+2\ell_6& 1 & 1 & 1 & 0 \\
 \hline
\ell'_7=  \ell_6-\ell_7 & 0 & 1 & 0 & 1 \\
 \hline
\end{array}
\end{equation}
which shows that Q$_6$ is isomorphic to a Hirzebruch surface $\mathbb{F}_1$.

\subsection{The vertical surface Q$_5$}
The surface Q$_5$ is defined by 
\begin{equation} C_6: \quad e_6=b e_1 e_3 e_4 e_5 e_7 s^2 + a  x =0, \end{equation}
which reduces to
$e_6=0$ over $V(a,b)\cap S$ in Y$_5$. 
The projective coordinates are: 
\begin{align}
\begin{aligned}
[0:0: s]\   [x: 0:e_1 e_3 e_4 e_5 e_7 ]
   [0 :e_1 ]\  
   [0 :e_3]\   [0:e_4]\  [y :e_2 e_7] \  [e_2:e_5 ],
  \end{aligned}
\end{align}
which imply that 
\begin{equation}
s e_1 e_3 e_4 \neq 0.
\end{equation}
The successive blowups defining X$^+_7$ give the scalings:
\begin{equation}
\begin{array}{|c|c|c|c|c|c|c|c|c|c|c|}
 \hline
X^+_7& s & x & y & e_1 &e_2 & e_3 &e_4 & e_5 & e_6 & e_7\\
 \hline
 \ell_1& 1 & 1 & 1 & -1 & 0 & 0 & 0 & 0 & 0 & 0 \\
 \hline
 \ell_2& 0 & 1 & 1 & 1 & -1 & 0 & 0 & 0 & 0 & 0 \\
 \hline
\ell_3&  0 & 0 & 1 & 1 & 0 & -1 & 0 & 0 & 0 & 0 \\
 \hline
\ell_4&  0 & 0 & 0 & 0 & 1 & 1 & -1 & 0 & 0 & 0 \\
 \hline
\ell_5&  0 & 0 & 1 & 0 & 0 & 0 & 1 & -1 & 0 & 0 \\
 \hline
\ell_6&  0 & 0 & 1 & 0 & 1 & 0 & 0 & 0 & -1 & 0 \\
 \hline
\ell_6&  0 & 0 & 0 & 0 & 1 & 0 & 0 & 1 & 0 & -1 \\
 \hline
\end{array}
\end{equation}
which we redefine as follows:
\begin{equation}
\begin{array}{| l |c|c|c|c|c|c|c|c|c|c|}
 \hline
X^+_7& s & x & y & e_1 & e_2& e_3 & e_4 &e_5 & e_6 & e_7 \\
  \hline
\ell'_1= \ell_1+\ell_3+\ell_4+\ell_5-\ell_6& 1 & 1 & 2 & 0 & 0 & 0 & 0 & -1 & 1 & 0 \\
  \hline
\ell'_2=\ell_2-\ell_3-\ell_4-\ell_5+2\ell_6 & 0 & 1 & 1 & 0 & 0 & 0 & 0 & 1 & -2 & 0 \\
 \hline
\ell'_3=\ell_3+\ell_4+\ell_5-\ell_6  & 0 & 0 & 1 & 1 & 0 & 0 & 0 & -1 & 1 & 0 \\
 \hline
\ell'_4= \ell_4+\ell_5-\ell_6 & 0 & 0 & 0 & 0 & 0 & 1 & 0 & -1 & 1 & 0 \\
  \hline
\ell'_5= \ell_5-\ell_6+\ell_7 & 0 & 0 & 0 & 0 & 0 & 0 & 1 & 0 & 1 & -1 \\
  \hline
\ell'_6= \ell_6 
& 0 & 0 & 1 & 0 & 1 & 0 & 0 & 0 & -1 & 0 \\
 \hline
\ell'_7=\ell_7 & 0 & 0 & 0& 0 & 1 & 0 & 0 & 1 & 0 & -1 \\
 \hline
\end{array}
\end{equation}
We can then fix $(s, e_1, e_3 ,e_4)$ using $(\ell'_1, \ell'_3, \ell'_4, \ell'_5)$. After imposing $e_6=0$, we are left with: 
\begin{equation}
\begin{array}{|l|c|c|c|c|c|}
\hline
Q_5 & x & y & e_2 & e_5 & e_7 \\
 \hline
 \ell'_2=\ell_2-\ell_3-\ell_4-\ell_5+2\ell_6&1 & 1 & 0 & 1 & 0 \\
 \hline
 \ell'_6=\ell_6& 0 & 1 & 1 & 0 & 0 \\
 \hline
\ell'_7=\ell_7& 0 & 0 & 1 & 1 & -1 \\
 \hline
\end{array}
\end{equation}
which shows that Q$_5$ is  a Hirzebruch surface $\mathbb{F}_1$ blown-up at a point (namely $e_2=e_5=0$) of its unique curve of self-intersection $-1$.

\subsection{The vertical surface Q$_4$}

The surface Q$_4$ is defined by
\begin{equation}
C_6:\quad e_5=b e_1 e_3 e_4 e_6 e_7^2 s^2 + a x =0,
\end{equation}
which reduces to $e_5=0$ over $V(a,b)\cap S$. 
The projective coordinates are
\begin{align}
\begin{aligned}
[0 :0 : s]\   [x: 0:e_1 e_3 e_4 e_6 e_7^2 ]\  
   [0 :e_1 ]\  
   [0 :e_3]\   [e_6 e_7 y :e_2]\  [y :e_4 e_7] \  [e_4:e_6 ],
  \end{aligned}
\end{align}
which means that we have 
\begin{equation}
s e_1 e_3 \neq 0.
\end{equation}
The successive blowups defining X$''_7$ give the scalings:
\begin{equation}
\begin{array}{|c|c|c|c|c|c|c|c|c|c|c|}
 \hline
X''_7  & s & x & y & e_1 & e_2 & e_3 & e_4 & e_5 & e_6& e_7\\
  \hline
 \ell_1 & 1 & 1 & 1 & -1 & 0 & 0 & 0 & 0 & 0 & 0 \\
  \hline
 \ell_2 & 0 & 1 & 1 & 1 & -1 & 0 & 0 & 0 & 0 & 0 \\
  \hline
 \ell_3 & 0 & 0 & 1 & 1 & 0 & -1 & 0 & 0 & 0 & 0 \\
  \hline
 \ell_4 & 0 & 0 & 0 & 0 & 1 & 1 & -1 & 0 & 0 & 0 \\
  \hline
 \ell_5 & 0 & 0 & 1 & 0 & 1 & 0 & 0 & -1 & 0 & 0 \\
  \hline
 \ell_6 & 0 & 0 & 1 & 0 & 0 & 0 & 1 & 0 & -1 & 0 \\
  \hline
 \ell_7 & 0 & 0 & 0 & 0 & 0 & 0 & 1 & 0 & 1 & -1 \\
  \hline
\end{array}
\end{equation}
We redefine them as follows:
\begin{equation}
\begin{array}{|l|c|c|c|c|c|c|c|c|c|c|}
  \hline
X''_7  & s & x & y &e_1 &e_2 & e_3& e_4 &e_5& e_6 & e_7 \\
   \hline
\ell'_1=\ell_1+\ell_2& 1 & 2 & 2 & 0 & -1 & 0 & 0 & 0 & 0 & 0 \\
  \hline
\ell'_2=\ell_2-\ell_3-\ell_4+2\ell_5 & 0 & 1 & 2 & 0 & 0 & 0 & 1 & -2 & 0 & 0 \\
  \hline
\ell'_3=\ell_3+\ell_4 & 0 & 0 & 1 & 1 & 1 & 0 & -1 & 0 & 0 & 0 \\
  \hline
\ell'_4= \ell_4 & 0 & 0 & 0 & 0 & 1 & 1 & -1 & 0 & 0 & 0 \\
   \hline
\ell'_5= \ell_5 & 0 & 0 & 1 & 0 & 1 & 0 & 0 & -1 & 0 & 0 \\
   \hline
\ell'_6= \ell_6 & 0 & 0 & 1 & 0 & 0 & 0 & 1 & 0 & -1 & 0 \\
   \hline
\ell'_7= \ell_7& 0 & 0 & 0 & 0 & 0 & 0 & 1 & 0 & 1 & -1 \\
   \hline
\end{array}
\end{equation}
which allows us to fix $(s,e_1, e_3)$ using $(\ell'_1, \ell'_2$, $\ell'_3)$. After imposing $e_5=0$, we are left with: 
\begin{equation}
\begin{array}{|l|c|c|c|c|c|c|}
   \hline
Q_4 & x & y &e_2 & e_4& e_6 & e_7 \\
    \hline
\ell_2'= \ell_2-\ell_3-\ell_4+2 \ell_5 & 1 & 2 & 0 & 1 & 0 & 0 \\
    \hline
\ell'_5= \ell_5 & 0 & 1 & 1 & 0 & 0 & 0 \\
    \hline
\ell'_6= \ell_6 & 0 & 1 & 0 & 1 & -1 & 0 \\
    \hline
\ell'_7= \ell_7 & 0 & 0 & 0 & 1 & 1 & -1 \\
    \hline
\end{array}
\end{equation}
which is a Hirzebruch surface $\mathbb{F}_2$ (parametrized by $(x,y,e_2, e_4)$), blown-up at a point $P: y=e_4=0$ of its curve of self-intersection $2$, followed by a blowup of the intersection  point $(e_4=e_5=0)$ of the resulting exceptional fiber and the proper transform of the fiber over the point $P$. 

\section*{Acknowledgements}
 M.E. is supported in part by the National Science Foundation (NSF) grant DMS-1701635 ``Elliptic Fibrations and String Theory.'' 
We would like to thank Patrick Jefferson and Monica Jinwoo Kang for  conversations.

\end{document}